\documentclass[sn-mathphys]{sn-jnl}

\usepackage[style=american]{csquotes}
\usepackage{stmaryrd}
\usepackage{mathtools}
\usepackage{upgreek}
\usepackage{bm}

\catcode`\|=12\usepackage{gnuplot-lua-tikz}

\jyear{2023}

\theoremstyle{plain}
\newtheorem{theorem}{Theorem}
\newtheorem*{theorem*}{Theorem}
\newtheorem{proposition}[theorem]{Proposition}

\newtheorem{corollary}[theorem]{Corollary}
\newtheorem{lemma}[theorem]{Lemma}

\theoremstyle{remark}
\newtheorem*{example}{Example}
\newtheorem*{remark}{Remark}

\theoremstyle{definition}
\newtheorem{definition}{Definition}
\newtheorem{convention}[definition]{Convention}

\expandafter\def\expandafter\thetheorem\expandafter{\expandafter P\thetheorem}
\expandafter\def\expandafter\thedefinition\expandafter{\expandafter D\thedefinition}

\def\eatscanweak#1{}%
\def\smartscanweak#1#2\eatscanweak{}%
\def\eatrestscanweak#1\termscanweak{}%
\def\restscanweak#1#2\termscanweak{{#2}}%
\def\scanweak#1#2{\smartscanweak#2\eatscanweak\eatscanweak\eatrestscanweak#1{#2}\expandafter\scanweak\expandafter#1\restscanweak#2\termscanweak}%
\def\eatrestgroupweak#1\termgroupweak{}%
\long\def\secondgroupweak#1\thirdgroupweak#2#3{}%
\def\eatspacegroupweak. #1\thirdgroupweak{{#1}}%
\def\eattofirst#1\firstgroupweak{}%
\long\def\eattosecond#1\secondgroupweak{}%
\def\smartgroupweak#1#2\eattofirst{}%
\def\spacegroupweak#1 {\smartgroupweak#1\eattofirst\eattofirst\eattosecond}%
\long\def\isgroupweak#1#2#3{\expandafter\eatrestgroupweak\expandafter{\ifx{#1\termgroupweak#2\expandafter\secondgroupweak\else}\termgroupweak\ifx.}\fi\fi
  \spacegroupweak#1. \firstgroupweak#3\secondgroupweak
  \expandafter\isgroupweak\eatspacegroupweak.#1.\thirdgroupweak{#2}{#3}}%
\def\interbraspa{\mskip-1.5mu }%
\def\leftbraspa#1{\mathchar`#1\mskip-\ifx#1(1\else.75\fi mu }%
\def\rightbraspa#1{\mskip-\ifx#1)1\else.75\fi mu \mathchar`#1}%
\def\auxtermgroupbraspa{}%
\def\termgroupbraspa{\auxtermgroupbraspa}%
\def\nnoexpand{\noexpand\noexpand\noexpand}%
\def\nnnoexpand{\nnoexpand\noexpand\nnoexpand}%
\def\firstbraspa#1#2\termbraspa{#1}%
\def\restbraspa#1#2\termbraspa{{#2}}%
\def\auxwillbegroupbraspa#1\wasbraspa\auxtermgroupbraspa{}%
\def\willbegroupbraspa#1{{#1}\auxwillbegroupbraspa#1\expandafter\wasbraspa\auxtermgroupbraspa}%
\def\auxwillbebraspa{}%
\def\willbebraspa{\auxwillbebraspa}%
\def\wasbraspa#1{\ifx#1\willbebraspa\interbraspa\else\expandafter#1\fi}%
\def\eateightbraspa#1#2#3#4#5#6#7#8{\noexpand\wasbraspa}%
\def\isopenbraspa#1#2#3{\ifx#1(#2\else\ifx#1[#2\else#3\fi\fi}%
\def\isclosebraspa#1#2#3{\ifx#1)#2\else\ifx#1]#2\else#3\fi\fi}%
\def\isopenclosebraspa#1#2#3#4{\isopenbraspa{#1}{\isclosebraspa{#2}{#3}{#4}}{#4}}%
\def\parsebraspa#1#2#3#4\termbraspa{\isopenclosebraspa#1#3%
  {\nnoexpand\willbebraspa\leftbraspa#1#2\rightbraspa#3\noexpand\eateightbraspa}%
  {\ifx#1\termgroupbraspa\nnoexpand\auxtermgroupbraspa\else\nnnoexpand#1\fi}}%
\def\checkbraspa#1#2\termbraspa{\isgroupweak{#2}{\nnnoexpand#1}{\expandafter\isgroupweak\restbraspa#2\termbraspa{\nnnoexpand#1}{\parsebraspa#1#2\termbraspa}}}%
\def\iterbraspa#1{\isgroupweak{#1}%
  {\nnoexpand\willbegroupbraspa{\expandafter\expandbraspa\expandafter{\firstbraspa#1\termbraspa}}}%
  {\checkbraspa#1..\termbraspa}}%
\def\expandbraspa#1{\scanweak\iterbraspa{#1\termgroupbraspa}}%
\def\braspa#1{\edef\auxabraspa{\expandbraspa{#1}}\edef\auxbbraspa{\auxabraspa}\edef\auxcbraspa{\auxbbraspa}\auxcbraspa}%
\def\mathdefs{%
  \let\s\braspa%
  \def\_##1{_{\braspa{##1}}}%
  \def\^##1{^{\braspa{##1}}}}%
\everymath\expandafter{\the\everymath\mathdefs}%
\everydisplay\expandafter{\the\everydisplay\mathdefs}%
\mathcode`\;="603B\relax  % we use semicolon in indices; the space that sn-jnl adds after it is unwanted
\def\mathcenter#1#2{\def\mathcenterhelper##1##2{\setbox0=\hbox{$\mathsurround=0pt ##1{##2}$}\hbox to \wd0{\hfil$\mathsurround=0pt ##1{#2}$\hfil}}\mathpalette\mathcenterhelper{#1}}
\def\xsub#1{x_{\scriptscriptstyle(#1)}}
\def\SR{{}^{\mathrm{S}}\!R}
\def\bw#1{b.w.\kern\fontdimen2\font$#1$}  % non-stretchable space of natural width

\begin{document}

\title[On Weyl alignment preserving KK reduction of vacuum]{On Weyl alignment preserving Kaluza--Klein reduction of vacuum}

\author*[1,2]{\fnm{Tom\'a\v s} \sur{Tint\v era}}
	\email{tintera@math.cas.cz}

%\author[1]{\fnm{Vojt\v ech} \sur{Pravda}}
%	\email{pravda@math.cas.cz}

\affil*[1]{%
	\orgname{Institute of Mathematics of the Czech Academy of Sciences},
	\orgaddress{\street{\v Zitn\' a 25}, \postcode{115 67} \city{Prague 1}, \state{Czech Republic}}}

\affil[2]{%
	\orgdiv{Institute of Theoretical Physics},
	\orgname{Faculty of Mathematics and Physics, Charles University in Prague},
	\orgaddress{\street{V Hole\v sovi\v ck\'ach 2}, \postcode{180 00} \city{Prague 8}, \state{Czech Republic}}}

\abstract{We study null alignment properties of Weyl tensors related via Kaluza--Klein reduction of vacuum spacetimes by one spatial Killing direction. Kaluza--Klein reduction is a method that relates spacetimes of different dimensionality. Weyl tensor null alignment is used in a recently proposed generalization of the Petrov algebraic classification of spacetimes to higher dimensions. Concentrating on the case where the two considered null directions are parallel in a gauge where they are perpendicular to the Maxwell potential, we express the relations between Riemann tensor null frame components of the original and reduced spacetime; we do the same for the Weyl tensors and also for optical matrices and non-geodeticities. Based on this, we point out basic consequences regarding reduction of Kundt spacetimes, and of spacetimes admitting a geodetic null direction. Finally, we work out the necessary and sufficient conditions for a Kaluza--Klein lift to preserve Weyl alignment type with respect to the related null directions. In the cases of types III and N with non-vanishing Maxwell field, where both spacetimes turn out to be Kundt, we show explicit solutions for the scalar potential in six dimensions and greater, and discuss some qualitative differences from the four-dimensional case.}

\keywords{Algebraically special spacetimes in higher dimensions, Kaluza--Klein reduction, Algebraic classification}

\maketitle

In this article, we elaborate on conditions for two spacetimes related by a~Kaluza--Klein reduction, to share a common Weyl aligned null direction. We show that in non-trivial cases of alignment types III and N, both spacetimes must be Kundt, and we find explicit forms of the Kaluza--Klein scalar potential necessary to satisfy these alignment conditions. We also point out some interesting qualitative differences of the four-dimensional case, which appears to be less restrictive on the scalar potential.

\section{Introduction}

Kaluza--Klein reduction \cite{nordstrom-1914,kaluza-1921,klein-1926,einstein-bergmann}, being a relation of spacetimes of different dimensionality, is of interest when investigating higher-dimensional algebraic classification. In this work, which expands on the work done in \cite{kk-i-ii}, we concentrate on the simplest case of a reduction of vacuum by one spatial Killing dimension and investigate null alignment \cite{clas-weyl-hd,alignment-special-tensors} \cite[\bgroup\it see also\egroup][]{clas-hd-review} of the related Weyl tensors. We specialize on the case where the two aligned null directions are parallel in a gauge in which they are perpendicular to the Maxwell potential.

While computing frame components of the related Weyl tensors, we demonstrate derivation of the Kaluza--Klein relation of Riemann tensors in a frame formalism, which turns out being easier than the classical coordinate-based derivation. We show how the same technique can be used to derive the relation of optical matrices and non-geodeticities and point out some straightforward consequences regarding twist, geodeticity and Kundt spacetimes.

For types I, II, III and N, we formulate algebraic conditions necessary and sufficient for a Kaluza--Klein lift to be Weyl type preserving with respect to the null directions related via the above condition (i.e. parallel in a \enquote{null gauge}). In the case of type III and N with non-vanishing Maxwell field, we find specific form of the scalar field that these conditions demand in dimension 6 and greater. We point out that in 4 dimensions, the conditions are likely less strict on scalar field behavior and allow for more types of solutions.

\subsection{Alignment conditions summary}
\label{chap:alignment-conditions-summary}

For convenience of the reader, we will summarize here the algebraic alignment conditions for individual types.
For definitions and formalism required for their correct interpretation, the reader is referred to later sections, especially sections \ref{chap:preliminaries} and \ref{chap:reduced-spacetime}.

Supposing $\mathcal{F}_{ab} = 0$, the conditions necessary and sufficient for a Weyl alignment preserving vacuum Kaluza--Klein lift of boost order $b$ are simple:
\begin{equation*}
	\mathrm{bo}_{\langle\bm{\ell}\rangle}\,\Psi_{ab} \le b\,,
\end{equation*}
where
\begin{align*}
	\Psi_{ab} &:= \Phi_{ab} - \Phi\_{(0)(1)}\,g_{ab}\,,\\
	\Phi_{ab} &:= \phi_{;ab} + (\beta-2\alpha)\,\phi_{;a} \phi_{;b}\,.
\end{align*}

On the other hand, if we assume $\mathcal{F}_{ab} \neq 0$, following conditions must be added in order to become sufficient while staying necessary. Some of interesting redundant necessary conditions are given in brackets.
\begin{itemize}
	\item For type I:
	\begin{gather*}
		\mathrm{bo}_{\langle\bm{\ell}\rangle}\,\mathcal{F}_{ab} < 1\\*
		\mathcal{F}\_{(i)(j)} \kappa\^{(j)} = 0\\*
		\mathcal{F}\_{(0)(1)} \kappa\^{(i)} = 0
	\end{gather*}
	\item For geodetic type II:
	\begin{gather*}
		\mathrm{bo}_{\langle\bm{\ell}\rangle}\,\mathcal{F}_{ab} < 1\\*
		\kappa\^{(i)} = 0\\*
		\Big(\mathcal{F}\_{(i)(k)} + \mathcal{F}\_{(0)(1)} g\_{(i)(k)}\Big)\,\Big(L\^{(k)}\_{\hphantom{(k)}(j)} + \alpha\phi\_{;(0)} \delta\^{(k)}\_{(j)}\Big) = \beta \phi\_{;(0)} \mathcal{F}\_{(i)(j)}
	\end{gather*}
	\item For non-geodetic type II:
	\begin{align*}
		\mathrm{bo}_{\langle\bm{\ell}\rangle}\,\mathcal{F}_{ab} &< 1 & \mathcal{F}\_{(i)(k)}\,\bigl(L\^{(k)}\_{\hphantom{(k)}(j)} + (\alpha-\beta)\phi\_{;(0)} \delta\^{(k)}\_{(j)}\bigr) &= 0\\*
		\mathcal{F}\_{(0)(1)} &= 0 & \bigl[\,\left(\Phi\_{;(i)(j)} - \Phi\_{;(0)(1)} g\_{(i)(j)}\right) \kappa\^{(j)} &= 0\,\bigr]\\*
		\kappa\^{(i)} &\neq 0 & \bigl[\,\mathcal{F}\_{(a)(i)} \kappa\^{(i)} &= 0\,\bigr]\\*
		\mathcal{F}\_{(i)(j)} \kappa\^{(j)} &= 0 &&
	\end{align*}
	\item For type III:
	\begin{align*}
		\mathrm{bo}_{\langle\bm{\ell}\rangle}\,\mathcal{F}_{ab} &= -1 & L\_{(i)(j)} &= 0\\*
		\mathrm{bo}_{\langle\bm{\ell}\rangle}\,\phi_{;a} &< 1 & [\,\kappa\^{(i)} &= 0\,]
	\end{align*}
	\item For type N with vanishing $\phi\_{;(i)}$:
	\begin{align*}
		\mathrm{bo}_{\langle\bm{\ell}\rangle}\,\mathcal{F}_{ab} &= -1 & L\_{(i)(j)} &= 0\\*
		\mathrm{bo}_{\langle\bm{\ell}\rangle}\,\phi_{;a} &< 0 & [\,L\_{(i)(1)} &= 0\,]\\*
		\mathcal{F}\_{(1)(i);(j)} &= 0 & [\,\kappa\^{(i)} &= 0\,]\\*
		[\,\mathrm{bo}_{\langle\bm{\ell}\rangle}\,\mathcal{F}_{ab;c} &< -1\,] &
	\end{align*}
	\item For type N with non-vanishing $\phi\_{;(i)}$ (affinely parametrized $\bm{\ell}$ is used):
	\begin{align*}
		\mathrm{bo}_{\langle\bm{\ell}\rangle}\,\mathcal{F}_{ab} &= -1 & L\_{(i)(j)} &= 0 & \phi\_{;(i)} &= \tfrac{1}{\alpha} L\_{(i)(1)}\\*
		\mathrm{bo}_{\langle\bm{\ell}\rangle}\,\phi_{;a} &= 0 & [\,\kappa\^{(i)} &= 0\,] & \mathcal{F}\_{(1)(i)} &= \gamma \phi\_{;(i)}\\*
		&&&& \gamma\_{;(i)} &= -2\beta \gamma \phi\_{;(i)}
	\end{align*}
\end{itemize}

\section{Preliminaries}
\label{chap:preliminaries}

\subsection{Null alignment based classification}

In this section, we establish basic definitions regarding algebraic classification of spacetimes of generic dimension $D$ where $D \ge 4$. For a much more thorough review, see \cite{clas-hd-review}.

\begin{definition}
	We call a frame
	\begin{equation}\label{eq:null-frame}
		\bm{\ell} \equiv \bm{e}\_{(0)},\;\bm{n} \equiv \bm{e}\_{(1)},\;\bm{e}\_{(i)}, \qquad i=2,\dots,D-1
	\end{equation}
	in the tangent space the \emph{null frame}, if the inverse metric takes the form
	\begin{equation}
		{}^{\sharp\sharp}\bm{g} = \bm{\ell} \vee \bm{n} + \sum_{i=2}^{D-1} \bm{e}\_{(i)} \bm{e}\_{(i)}\,.
	\end{equation}
\end{definition}

By indices in brackets, we will denote the tensor frame components\footnote{As an exception, we use brackets also for enumerating the basis vectors.}, as in
\begin{equation}\label{eq:frame-component-def}
	C\_{(0)(1)(i)(j);(k)} = C_{abcd;e}\,\ell^{a} n^{b} e\_{(i)}^{c} e\_{(j)}^{d} e\_{(k)}^{e}\,.
\end{equation}
In these cases, the range of $\lbrace 2, \dots, D-1 \rbrace$ is implied for indices $\scriptstyle \s{(i), (j)}, \dots$. Similarly, indices $\scriptstyle \s{(a), (b)}, \dots$ would take values from $\lbrace 0, \dots, D-1 \rbrace$. For any kind of bracketed indices, Einstein's summation convention is to be applied.
\begin{remark}
	Please note how the expression in \eqref{eq:frame-component-def} denotes frame components of the derivative and not vice versa.
\end{remark}

\begin{definition}
	We define the \emph{boost order} $\mathrm{bo}_{\langle\bm{\ell}\rangle}\,\bm{T}$ of a nonzero tensor $\bm{T}$ with respect to a null direction $\langle\bm{\ell}\rangle$ as usual \cite{clas-weyl-hd}. Additionally, we define $\mathrm{bo}_{\langle\bm{\ell}\rangle}\,\bm{T} = -\infty$ for any $\bm{T}$ that locally vanishes.
\end{definition}

\begin{definition}
	We classify spacetimes based on the boost order of their Weyl tensor $\bm{C}$; we say that a spacetime is of type I, II, III or N with respect to $\bm{\ell}$, if
	\begin{equation}
		\mathrm{bo}_{\langle\bm{\ell}\rangle}\,\bm{C} < b
	\end{equation}
	for $b = 2, \dots, -1$, respectively.
\end{definition}
\begin{remark}
	Note that under this definition, $\text{N} \subset \text{III} \subset \text{II} \subset \text{I}$.
\end{remark}

\subsection{Optical matrix}

Later, we will be working with the optical matrix, non-geodeticity and other Ricci rotation coefficients. Let's recall their definition.
\begin{definition}
	For a given null frame, we define:
	\begin{subequations}
	\begin{alignat}{2}
		\label{eq:ricci-rotation-def}
		L_{ab} &:= \ell_{a;b}\,, &&\\
		\kappa^{a} &:= \ell^{a}_{\hphantom{a};b} \ell^{b} && \mathrlap{\quad\text{\emph{(non-geodeticity of $\bm{\ell}$)}.}}
	\end{alignat}
	\end{subequations}
\end{definition}

\begin{definition}
	For a given null frame, we will call the \emph{optical matrix} a matrix composed of components identical to the frame components $L\_{(i)(j)}$. We will denote the optical matrix as $\rho_{ij}$:
	\begin{equation}
		\rho_{ij} = L\_{(i)(j)}\,.
	\end{equation}
\end{definition}

\begin{definition}
	In the decomposition
	\begin{equation}
		\rho_{ij} = \theta \delta_{ij} + \sigma_{ij} + A_{ij}\,,
	\end{equation}
	of the optical matrix into individual spin-invariant subspaces:
	\begin{subequations}
	\begin{alignat}{3}
		\theta \delta_{ij} &= \tfrac{1}{D-2}\,\rho_{kk} \delta_{ij}\,,\\
		\sigma_{ij} &= \rho_{(ij)} - \theta \delta_{ij}\,,\\
		A_{ij} &= \rho_{[ij]}\,,
	\end{alignat}
	\end{subequations}
	we call
	\begin{subequations}
	\begin{alignat}{2}
		&\hphantom{{}={}} \theta && \mathrlap{\quad\dots\text{\emph{expansion},}}\\
		\sigma^{2} &= \sigma_{ij} \sigma_{ji} && \mathrlap{\quad\dots\text{\emph{shear},}}\\
		\omega^{2} &= -A_{ij} A_{ji} && \mathrlap{\quad\dots\text{\emph{twist}}}
	\end{alignat}
	\end{subequations}
	the \emph{optical scalars}.
\end{definition}

\begin{definition}
	We say that a vector field is \emph{weakly geodetic}, if it is tangent to a geodesic. We reserve the title \emph{geodetic} for vector fields generating affinely parametrized geodesics.
\end{definition}

\begin{corollary}
	$\bm{\ell}$ is weakly geodetic, iff $\kappa\^{(i)} = 0$.
\end{corollary}

\begin{corollary}
	\label{stat:rho-transformation-geodetic}
	Lorentz transformation that preserves the direction of a weakly geodetic $\bm{\ell}$, acts on the optical matrix as $\mathrm{O}(D-2)$:
	\begin{equation}
		\label{eq:rho-transformation-geodetic}
		\rho'_{ij} = \Lambda\_{(i)}\^{\hphantom{(i)}(k)} \Lambda\_{(j)}\^{\hphantom{(j)}(l)} \rho_{kl}\,.
	\end{equation}
\end{corollary}
\begin{remark}
	We can call a weakly geodetic $\bm{\ell}$ shear-free, expansion-free or twist-free without reference to any frame, thanks to corollary \ref{stat:rho-transformation-geodetic}, which ensures that a potential $\bm{\ell}$-preserving frame transformation remains linear in $\rho_{ij}$ and does not mix up the decomposition of $\rho_{ij}$ into spin-invariant subspaces.
\end{remark}

\subsection{Kaluza--Klein reduction}
\label{chap:kk-basics}

In this section, based on the review of \cite{pope-ihplec}, we provide a basic introduction to the method of Kaluza--Klein reduction.

Let's consider a $(D+1)$-dimensional spacetime $(\mathcal{\hat{M}}, \bm{\hat{g}})$. In order to perform a Kaluza--Klein reduction, we need to choose a function $z$, together with a space-like vector field $\bm{\xi}$, such that $(\bm{\mathrm{d}}z, \bm{\xi}) = 1$. Additionally, we will also assume that $\bm{\xi}$ is a Killing vector field of $\bm{\hat g}$ (this is a usual assumption for Kaluza--Klein reduction, called \emph{cylinder condition} \citep{kaluza-1921}, or \emph{massless truncation} \citep{pope-ihplec}). We interpret $z$ as a coordinate function and $\bm{\xi}$ as the corresponding coordinate vector field $\bm{\partial}_{z}$. We also define embedded submanifolds $\mathcal{M}_{z}$ of $\mathcal{\hat{M}}$ by fixing $z = \text{const.}$:
\begin{equation}
	\iota_{z} \colon \mathcal{M}_{z} \rightarrow \mathcal{\hat{M}}\,.
\end{equation}
Let's denote tensor fields on $\mathcal{M}_{z}$ by using small Latin indices (except index $z$, which we will reserve for another use -- see below), as opposed to capital Latin indices which we use for tensors on $\mathcal{\hat{M}}$.

$\delta^{A}_{B} - \xi^{A} \mathrm{d}_{B}z$ is a projector of $T_{x}\mathcal{\hat{M}}$ onto $T_{x}\mathcal{M}_{z}$ that we can use to project tensor fields on $\mathcal{\hat{M}}$ onto $\mathcal{M}_{z}$. We will denote the corresponding tensor operator by $\uppi_{z}$. This projector is, from the definition, an extension of the inverse of the pushforward by the embedding\footnote{More precisely, perceived as an operator on tensor fields, $\iota_{z\star}$ is a pushforward by the corestriction of $\iota_{z}$ to its image.} $\iota_{z}$:
\begin{equation}
\label{eq:projection-identity}
	\uppi_{z} \circ \iota_{z\star} = \mathrm{id}_{T\mathcal{M}_{z}}\,.
\end{equation}
We observe that, for the projection of the metric $\hat{g}_{AB}$ or for the projection of its Riemann tensor $\hat{R}\vphantom{R}^{A}_{\hphantom{A}BCD}$, the precise value of the constant $z$ (present in the definition of $\iota_{z}$) is irrelevant due to $\bm{\partial}_{z}$ being a Killing vector field. For this reason, we will usually denote the embedding, the projector and the submanifold by just $\iota$, $\uppi$ and $\mathcal{M}$.

We can now decompose the metric $\bm{\hat g}$ into three fields $g_{AB}$, $\mathcal{A}_{A}$ and $\phi$ in a way that resembles the method of completing the square (the purpose of which we will see later when choosing an orthonormal basis) \citep{pope-ihplec}:
\begin{subequations}
\label{eq:kk-decomposition}
\begin{align}
	\label{eq:kk-decomposition-metric}
	\bm{\hat{g}} &= \mathrm{e}^{2\alpha\phi} \bm{g} + \mathrm{e}^{2\beta\phi} (\bm{\mathrm{d}}z + \bm{\mathcal{A}})^{2}\,,\\
	\label{eq:kk-decomposition-g}
	\bm{g} &= \iota_{\star} \uppi \bm{g}\,,\\
	\label{eq:kk-decomposition-A}
	\bm{\mathcal{A}} &= \iota_{\star} \uppi \bm{\mathcal{A}}\,.
\end{align}
\end{subequations}

Here, $\alpha$ and $\beta$ are nonvanishing constants. Note that in the chosen coordinates, the components of the fields $g_{AB}$, $\mathcal{A}_{A}$ and $\phi$ are independent of the coordinate $z$, due to $\bm{\partial}_{z}$ being a Killing vector field of $\bm{\hat{g}}$. From the last two equations, we see that the fields $g_{AB}$, $\mathcal{A}_{A}$ and $\phi$ can be identified with corresponding fields $g_{ab}$, $\mathcal{A}_{a}$ and $\phi$ on $\mathcal{M}$, by means of the projector $\uppi$.

\section{Introducing reduced spacetime}
\label{chap:reduced-spacetime}

Since $g_{ab}$ is symmetric and non-degenerate, we can interpret it as a metric on $\mathcal{M}$. We will also see that a certain subgroup of the gauge symmetries of $\hat{g}_{AB}$ is translated on $\mathcal{M}$ as a $\mathrm{U}(1)$ gauge group of $\mathcal{A}_{a}$, which encourages us to interpret $\mathcal{A}_{a}$ as a Maxwell potential.

\subsection{The metric}

\begin{convention}
\label{def:small-indices}
	Small Latin indices (except for index $z$) on tensor fields from $T\mathcal{\hat{M}}$ are to be interpreted as abstract indices on the projection by $\iota_{\star} \circ \uppi$. The index $z$ has a special meaning among other small Latin indices: we use covariant (or contravariant) index $z$ to denote contraction with $\xi^{A}$ (or $\mathrm{d}_{A}z$).
\end{convention}

\begin{proposition}
	$g_{ab}$ is non-degenerate on $\mathcal{M}$, with the inverse
	\begin{equation}
		\label{eq:inverse-g}
		g^{ab} := \mathrm{e}^{2\alpha\phi}\,\hat{g}^{ab}\,,
	\end{equation}
	where $\hat{g}^{AB}$ is the inverse metric on $\mathcal{\hat{M}}$.
\end{proposition}
\begin{proof}
	We start by expressing the explicit transformations between fields on $\mathcal{M}$ and $\mathcal{\hat{M}}$:
	\begin{subequations}
	\label{eq:kk-decomposition-metric-components}
	\begin{align}
		\label{eq:kk-decomposition-metric-components-ab}
		\hat{g}_{ab} &= \mathrm{e}^{2\alpha\phi} g_{ab} + \mathrm{e}^{2\beta\phi} \mathcal{A}_{a} \mathcal{A}_{b}\,,\\
		\label{eq:kk-decomposition-metric-components-az}
		\hat{g}_{az} &= \mathrm{e}^{2\beta\phi} \mathcal{A}_{a}\,,\\
		\label{eq:kk-decomposition-metric-components-zz}
		\hat{g}_{zz} &= \mathrm{e}^{2\beta\phi}\,,
	\end{align}
	\end{subequations}
	and vice versa:
	\begin{subequations}
	\label{eq:lower-fields}
	\begin{align}
		\label{eq:lower-fields-phi}
		\phi &= \tfrac{1}{2\beta} \ln \hat{g}_{zz}\,,\\
		\label{eq:lower-fields-A}
		\mathcal{A}_{a} &= \hat{g}_{zz}^{\hphantom{zz}\!-1} \hat{g}_{az}\,,\\
		\label{eq:lower-fields-g}
		g_{ab} &= \hat{g}_{zz}^{\hphantom{zz}\!-\frac{\alpha}{\beta}} \hat{g}_{ab} - \hat{g}_{zz}^{\hphantom{zz}\!-\frac{\alpha}{\beta}-1} \hat{g}_{az} \hat{g}_{bz}\,.
	\end{align}
	\end{subequations}

	Let's show that
	\begin{equation}
		g^{ab} = \hat{g}_{zz}^{\hphantom{zz}\!\frac{\alpha}{\beta}} \hat{g}^{ab}
	\end{equation}
	is inverse to $g_{ab}$ as in \eqref{eq:lower-fields-g}:
	\begin{equation}
		g_{ac} g^{cb} = \left( \hat{g}_{ac} - \hat{g}_{zz}^{\hphantom{zz}\!-1} \hat{g}_{az} \hat{g}_{cz} \right) \hat{g}^{cb} = \delta_{a}^{b} - \hat{g}_{zz}^{\hphantom{zz}\!-1} \hat{g}_{az} \left( \hat{g}_{zz} \hat{g}^{zb} + \hat{g}_{cz} \hat{g}^{cb} \right) = \delta_{a}^{b}\,.
	\end{equation}
\end{proof}

\begin{corollary}
	The inverse metric $\hat{g}^{AB}$ on $\mathcal{\hat{M}}$ can be expressed in terms of the $D$-dimensional fields as:
	\begin{subequations}
	\begin{align}
		\hat{g}^{ab} &= \mathrm{e}^{-2\alpha\phi} g^{ab}\,,\\
		\hat{g}^{az} &= -\mathrm{e}^{-2\alpha\phi} g^{ab} \mathcal{A}_{b}\,,\\
		\hat{g}^{zz} &= \mathrm{e}^{-2\beta\phi} + \mathrm{e}^{-2\alpha\phi} g^{ab} \mathcal{A}_{a} \mathcal{A}_{b}\,.
	\end{align}
	\end{subequations}
\end{corollary}

We are going to denote some tensors on $\mathcal{\hat{M}}$ with a hat symbol (as in $\hat{L}_{AB}$). In addition to (in some cases) making a connection to corresponding tensors on $\mathcal{M}$, this notation will convey a special convention:
\begin{convention}
	For tensors on $\mathcal{\hat{M}}$ denoted by a hat symbol, we will establish a convention that their indices will be raised and lowered by the metric $\bm{\hat{g}}$. Similarly, indices of tensors on $\mathcal{M}$ will be raised and lowered by $\bm{g}$.
\end{convention}

\subsection{Symmetries}

The component expression of the metric on $\mathcal{\hat{M}}$ \eqref{eq:kk-decomposition-metric-components} possesses a gauge freedom that can generally be parametrized by coordinate transformations combined with global change of scale, another symmetry of Einstein equations. Expressed in terms of the constant scaling parameter $a$ and the vector field $\hat{\omega}^{A}$ generating a small ($\epsilon \rightarrow 0$) coordinate transformation on coordinate functions $\hat{x}^{A}$,
\begin{equation}
	\hat{x}^{A} \mapsto \hat{x}^{A} + \epsilon \mathcal{L}_{\bm{\hat{\omega}}} \hat{x}^{A}\,,\qquad A \in \lbrace 0, \dots, D \rbrace\,,
\end{equation}
the metric transforms as:
\begin{equation}
	\hat{g}_{\!AB} \mapsto \hat{g}_{\!AB} + \epsilon \delta\hat{g}_{\!AB}\,,
\end{equation}
where
\begin{equation}
	\delta\hat{g}_{\!AB} = -\mathcal{L}_{\bm{\hat{\omega}}} \hat{g}_{\!AB} + 2a \hat{g}_{\!AB} = 2a \hat{g}_{\!AB} - \hat{g}_{\!AB,C}\,\hat{\omega}^{C} - \hat{g}_{\!CB}\,\hat{\omega}^{C}_{\hphantom{C}\!,A} - \hat{g}_{\!AC}\,\hat{\omega}^{C}_{\hphantom{C}\!,B}\,.
\end{equation}
Here, the coordinate derivatives can be just any coordinate derivatives, but for convenience, let's interpret them as derivatives with respect to coordinates compatible with $\bm{\partial}_{z}$.
Now, in order for $\bm{\partial}_{z}$ to stay a Killing vector field, the metric components in the new coordinates must be independent of the new $z$; in other words, the transformation must \citep{pope-ihplec} be generated by such $\hat{\omega}^{A}$ that:
\begin{equation}
	\hat{\omega}^{A} = \omega^{A} + (cz + \lambda)\,\xi^{A}\,,
\end{equation}
where $\omega^{A}$ is a vector field projectable to $T\mathcal{M}$:
\begin{equation}
	\mathcal{L}_{\bm{\xi}} \omega^{A} = 0
\end{equation}
and identifiable with such projection:
\begin{equation}
	\omega^{A} \mathrm{d}_{A} z = 0\,,
\end{equation}
and where $\lambda$ is independent of $z$:
\begin{equation}
	\xi^{A} \mathrm{d}_{A} \lambda = 0
\end{equation}
and $c$ is constant:
\begin{equation}
	\mathrm{d}_{A} c = 0\,.
\end{equation}
Note that under such transformations, the Killing vector field $\bm{\xi}$ preserves its direction.

It is a matter of simple substitution to express this gauge transformation in terms of the $D$-dimensional fields:
\begin{subequations}
\begin{align}
	\delta \bm{g} &= -\mathcal{L}_{\bm{\omega}} \bm{g} + 2 s \bm{g}\,,\\
	\delta \bm{\mathcal{A}} &= -\mathcal{L}_{\bm{\omega}} \bm{\mathcal{A}} - \bm{\mathrm{d}} \lambda + (\alpha-\beta) \phi_{0}\,\bm{\mathcal{A}} + s \bm{\mathcal{A}}\,,\\
	\delta \phi &= -\mathcal{L}_{\bm{\omega}} \phi + \phi_{0}\,,
\end{align}
\end{subequations}
where
\begin{subequations}
\begin{align}
	\phi_{0} &= \frac{a-c}{\beta}\,,\\
	s &= a - \frac{\alpha}{\beta} \left(a-c\right).
\end{align}
\end{subequations}
In the terms involving $\bm{\omega}$, we identify the $D$-dimensional coordinate transformation, whereas $\bm{\mathrm{d}} \lambda$ corresponds to a local $\mathrm{U}(1)$ gauge symmetry of $\bm{\mathcal{A}}$. The rest of the transformations pertain to global symmetries: $\phi_{0}$ parametrizes constant shift of $\phi$, compensated by appropriate scaling of $\bm{\mathcal{A}}$, while $s$ is responsible for the global change of scale.

The $D$-dimensional gauge invariance doesn't surprise us, because we constructed the $D$-dimensional fields as tensor fields on $\mathcal{M}$. What is interesting is namely the local $\mathrm{U}(1)$ gauge invariance of $\bm{\mathcal{A}}$, which already justifies our particular choice of $D$-dimensional fields. This allows us to call $\bm{\mathcal{A}}$ a Maxwell potential.

\subsection{Maxwell tensor}

In the next sections, we will develop a toolset to show that, assuming the lifted spacetime $(\mathcal{\hat{M}}, \bm{\hat g})$ is vacuum, we can write the $(D+1)$-dimensional Einstein equations as:
\begin{subequations}\label{eq:reduced-einstein-eq}
\begin{align}
	\label{eq:reduced-einstein-eq-a}
	&R_{ab} = \big(\beta + (D-2)\,\alpha\big)\,\phi_{;ab} + \big(\beta^{2} - 2\alpha\beta - (D-2)\,\alpha^{2} \big)\,\phi_{;a} \phi_{;b} + {}\nonumber\\
	&\hphantom{R_{ab} = {}} + \tfrac{\mathrm{e}^{(2\beta-2\alpha)\phi}}{2} \big(\mathcal{F}^{2}_{ab} + \tfrac{\alpha}{2\beta} \mathcal{F}^{2} g_{ab}\big)\,,\\[3pt]
	\label{eq:reduced-einstein-eq-b}
	&\Big(\mathrm{e}^{((D-4)\alpha + 3\beta)\phi} g^{bc} \mathcal{F}_{ab}\Big)\vphantom{\mathcal{F}}_{;c} = 0\,,\\[3pt]
	\label{eq:reduced-einstein-eq-c}
	&\mathrm{e}^{(2\beta - 2\alpha)\,\phi} \mathcal{F}^{2} = 4\beta g^{ab} \big(\phi_{;ab} + (\beta + (D-2)\,\alpha) \phi_{;a} \phi_{;b}\big)\,,
\end{align}
\end{subequations}
where we have introduced the analog of the Maxwell tensor:
\begin{subequations}
\begin{align}
	\mathcal{F}_{ab} &:= \mathcal{A}_{b;a} - \mathcal{A}_{a;b}\,,\\*
	\mathcal{F}^{2}_{ab} &:= g^{cd} \mathcal{F}_{ac} \mathcal{F}_{bd}\,,\\*
	\mathcal{F}^{2} &:= g^{ab} \mathcal{F}^{2}_{ab}\,.
\end{align}
\end{subequations}
and where by semicolon, we denote the covariant derivative on $\mathcal{M}$, induced by the the metric $\bm{g}$. We will also see how to extend this covariant derivative to the whole $\mathcal{\hat{M}}$.

Recalling the $\mathrm{U}(1)$ gauge symmetry of $\mathcal{A}_{a}$, it is of no surprise to see that the $D$-dimensional field equations also have the $\mathrm{U}(1)$ gauge symmetry, and the electromagnetic potential can be eliminated out of them in favor of the Maxwell tensor.

\begin{remark}
	Choosing
	\begin{equation}\label{eq:alpha-beta-relation}
		\beta + (D-2)\,\alpha = 0\,,
	\end{equation}
	the reduced field equations \eqref{eq:reduced-einstein-eq} could be further simplified.\footnote{The assumption \eqref{eq:alpha-beta-relation} is needed if we wish to interpret $\phi$ as the scalar in the Einstein--Maxwell--scalar system. For this reason, it is assumed in \cite{pope-ihplec}. We will however not assume any restriction on $\alpha$ and $\beta$, apart from both being nonzero.}
\end{remark}

\section{Frame choice}

Given a null frame \eqref{eq:null-frame} in the reduced manifold, it is natural (in the sense that it simplifies computations) to define the null frame in the original manifold the following way \citep{pope-ihplec}:
\begin{subequations}
\label{eq:original-frame}
\begin{align}
	\label{eq:original-frame-a}
	\bm{\hat{e}}\_{[a]} &= \mathrm{e}^{-\alpha\phi}\,(\bm{e}\_{(a)} - \mathcal{A}\_{(a)} \bm{\partial}_{z})\,, \qquad a=0,\dots,D-1\,,\\
	\label{eq:original-frame-z}
	\bm{\hat{e}}\_{[z]} &= \mathrm{e}^{-\beta\phi}\,\bm{\partial}_{z}\,,
\end{align}
\end{subequations}
where $\bm{e}\_{(a)}$ is meant as the pushforward of the null frame \eqref{eq:null-frame} by the embedding $\iota$ mentioned in section \ref{chap:kk-basics}. The notation using the tensor frame components in this frame will be governed by conventions similar to those for the reduced frame \eqref{eq:null-frame}. The range for indices $\scriptstyle \s{[a], [b]}, \dots$ remains $\lbrace 0, \dots, D-1 \rbrace$.

The frame \eqref{eq:original-frame} is indeed a null frame:
\begin{equation}
	{}^{\sharp\sharp}\bm{\hat{g}} = \bm{\hat{\ell}} \vee \bm{\hat{n}} + \sum_{i=2}^{D-1} \bm{\hat{e}}\_{[i]} \bm{\hat{e}}\_{[i]} + \bm{\hat{e}}\_{[z]} \bm{\hat{e}}\_{[z]}\,,
\end{equation}
where
\begin{equation}
	\bm{\hat{\ell}} := \bm{\hat{e}}\_{[0]}\,,\quad \bm{\hat{n}} := \bm{\hat{e}}\_{[1]}\,.
\end{equation}
Its dual is
\begin{subequations}
\label{eq:original-frame-dual}
\begin{align}
	\label{eq:original-frame-dual-a}
	\bm{\hat{\varepsilon}}\^{[a]} &= \mathrm{e}^{\alpha\phi}\,\bm{\varepsilon}\^{(a)}\,, \qquad a=0,\dots,D-1\,,\\
	\label{eq:original-frame-dual-z}
	\bm{\hat{\varepsilon}}\^{[z]} &= \mathrm{e}^{\beta\phi}\,(\bm{\mathrm{d}}z + \bm{\mathcal{A}})\,,
\end{align}
\end{subequations}
where $\bm{\varepsilon}\^{(a)}$ is the dual of $\bm{e}\_{(a)}$ in the subspace orthogonal to $\bm{\partial}_{z}$.

\begin{remark}
	Please note how the individual terms in the completeness relation
	\begin{equation}
		\label{eq:original-metric-completeness}
		\bm{\hat{g}} = \bm{\hat{\varepsilon}}\^{[0]} \vee \bm{\hat{\varepsilon}}\^{[1]} + \sum_{i=2}^{D-1} \bm{\hat{\varepsilon}}\^{[i]} \bm{\hat{\varepsilon}}\^{[i]} + \bm{\hat{\varepsilon}}\^{[z]} \bm{\hat{\varepsilon}}\^{[z]}
	\end{equation}
	correspond in form to the metric decomposition \eqref{eq:kk-decomposition-metric}. This is another manifestation of the beauty of the specific choice of the $D$-dimensional fields $\bm{g}$, $\bm{\mathcal{A}}$ and $\phi$.
\end{remark}

\section{Covariant derivative}

In view of the forthcoming computations, we will find convenient to define an extension of the covariant derivative on the reduced spacetime to the whole $T\mathcal{\hat{M}}$.
\begin{definition}
\label{def:nabla-definition}
	We extend the Levi-Civita connection $\bm{\nabla}$ (later also denoted by a semicolon) on $T\mathcal{M}$ to the entire lifted spacetime as the torsion-free affine connection annihilating, in addition to the reduced metric, both the vector $\bm{\partial}_{z}$ and the form $\bm{\mathrm{d}}z$:
	\begin{subequations}
	\label{eq:nabla-definition}
	\begin{align}
		\label{eq:nabla-definition-torsion}
		\bm{\mathrm{Tor}}(\bm{\nabla}) & = 0\,,\\
		\label{eq:nabla-definition-metric}
		\bm{\nabla} \bm{g} &= 0\,,\\
		\label{eq:nabla-definition-vector}
		\bm{\nabla} \bm{\partial}_{z} &= 0\,,\\
		\label{eq:nabla-definition-form}
		\bm{\nabla} \bm{\mathrm{d}}z &= 0\,.
	\end{align}
	\end{subequations}
\end{definition}
\begin{remark}
	This definition is consistent and unambiguous: the Christoffel symbols $\Gamma^{a}_{\hphantom{a}bc}$ of $\bm{\nabla}$ on $\mathcal{\hat{M}}$ correspond with those on $\mathcal{M}$, while $\Gamma^{z}_{\hphantom{z}BC}$, $\Gamma^{A}_{\hphantom{A}zC}$ and $\Gamma^{A}_{\hphantom{A}Bz}$ are all zero.
\end{remark}
\begin{remark}
	This $\bm{\nabla}$ is indeed an extension in the sense that for any tensor $\bm{T}$ on $T\mathcal{M}$, \enquote{the pushforward of the derivative $\bm{\nabla} \bm{T}$ is equal to the corresponding derivative of the pushforward}; precisely:
	\begin{equation}
		\label{eq:nabla-pushforward}
		\iota_{\star}(\bm{\nabla} \bm{T}) = \bm{\nabla} \iota_{\star}\bm{T}\,.
	\end{equation}
\end{remark}
\begin{remark}
	Thanks to the vanishing of $z$ components of $\Gamma^{A}_{\hphantom{A}BC}$, $\bm{\nabla}$ even \enquote{commutes}, in a sense, with the projector: for any tensor $\bm{T}$ on $T\mathcal{\hat{M}}$ that is projectable to $T\mathcal{M}$:
	\begin{equation}
		\label{eq:nabla-projection-precondition}
		\mathcal{L}_{\bm{\partial}_{z}} \bm{T} = 0\,,
	\end{equation}
	\enquote{the projection of the corresponding derivative $\bm{\nabla} \bm{T}$ is equal to the derivative of the projection}; precisely:
	\begin{equation}
		\label{eq:nabla-projection}
		\uppi(\bm{\nabla} \bm{T}) = \bm{\nabla} \uppi(\bm{T})\,.
	\end{equation}
\end{remark}
Much like when dealing with Christoffel symbols, we will introduce $\bm{\hat{\Gamma}}$:
\begin{definition}
\label{def:nabla-difference}
	We define $\bm{\hat{\Gamma}}$ as the tensor of difference between the original and the reduced covariant derivative:
	\begin{equation}
		\label{eq:nabla-difference}
		\hat{\Gamma}\vphantom{\Gamma}^{A}_{\hphantom{A}BC} v^{C} := \hat{\nabla}_{B} v^{A} - \nabla_{B} v^{A}\,,\mathrlap{\qquad v^{A} \in T\mathcal{\hat{M}}\,.}
	\end{equation}
\end{definition}
It is a common routine to derive the analogue of the expression of Christoffel symbols in terms of the metric:
\begin{proposition}
\label{stat:christoffel}
	The difference between the Levi-Civita connection $\bm{\hat{\nabla}}$ on $T\mathcal{\hat{M}}$ and a torsion-free connection $\bm{\nabla}$ (denoted by a semicolon) is given by
	\begin{equation}
		\label{eq:christoffel}
		\hat{\Gamma}_{ABC} = \tfrac{1}{2}\left( \hat{g}_{AB;C} + \hat{g}_{CA;B} - \hat{g}_{BC;A} \right).
	\end{equation}
\end{proposition}
\begin{proof}
	We follow the process which is standardly realized with a coordinate covariant derivative.

	First we note that due to zero torsion of both connections, the difference tensor is symmetric in the last two indices. For arbitrary $f\colon \mathcal{\hat{M}} \rightarrow \mathbb{R}$:
	\begin{equation}
		\label{eq:christoffel-symmetry}
		\hat{\Gamma}\vphantom{\Gamma}^{A}_{\hphantom{A}[BC]} \mathrm{d}^{\vphantom{A}}_{A} f = -\hat{\nabla}_{[B} \mathrm{d}_{C]} f + \nabla_{[B} \mathrm{d}_{C]} f = 0\,,
	\end{equation}
	where we used the Leibniz rule to see how $\bm{\hat{\nabla}} - \bm{\nabla}$ acts on covectors.

	Using the Levi-Civita property of $\bm{\hat{\nabla}}$ and again the Leibniz rule, we express $\bm{\nabla} \bm{\hat{g}}$:
	\begin{equation}
		\nabla_{C} \hat{g}_{AB} = -(\hat{\nabla} - \nabla)_{C} \hat{g}_{AB} = \hat{\Gamma}\vphantom{\Gamma}^{D}_{\hphantom{D}CA} \hat{g}_{DB} + \hat{\Gamma}\vphantom{\Gamma}^{D}_{\hphantom{D}CB} \hat{g}_{AD}\,,
	\end{equation}
	which can be combined using \eqref{eq:christoffel-symmetry} to obtain \eqref{eq:christoffel}.
\end{proof}

We now have all the necessary tools to express the frame components of $\bm{\nabla} \bm{\hat{g}}$, $\bm{\hat{\Gamma}}$ and finally $\bm{L}$, in terms of the frame components of the $D$-dimensional fields $\bm{g}$, $\bm{\mathcal{A}}$ and $\phi$ in the related frame.

\begin{proposition}
\label{stat:metric-derivative}
	The derivative $\bm{\nabla} \bm{\hat{g}}$ can be expanded as follows:
	\begin{subequations}
	\label{eq:metric-derivative}
	\begin{align}
		\label{eq:metric-derivative-abc}
		\hat{g}\_{[a][b];[c]} &= \mathrm{e}^{-3\alpha\phi} \left(\mathrm{e}^{2\alpha\phi}\right)\_{;(c)} g\_{(a)(b)}\,,\\
		\label{eq:metric-derivative-zab}
		\hat{g}\_{[z][a];[b]} &= \mathrm{e}^{(-2\alpha+\beta)\phi} \mathcal{A}\_{(a);(b)}\,,\\
		\label{eq:metric-derivative-zza}
		\hat{g}\_{[z][z];[a]} &= \mathrm{e}^{(-\alpha-2\beta)\phi} \left(\mathrm{e}^{2\beta\phi}\right)\_{;(a)}\,,\\
		\label{eq:metric-derivative-ABz}
		\hat{g}\_{[A][B];[z]} &= 0\,.
	\end{align}
	\end{subequations}
\end{proposition}
\begin{proof}
	We proceed with a direct proof and in a slightly verbose manner, to demonstrate the usage of the established notation.

	From \eqref{eq:nabla-pushforward} used on \eqref{eq:kk-decomposition-g} and \eqref{eq:kk-decomposition-A}, we see that the derivative $\nabla_{z}$ of the $D$-dimensional fields vanish:
	\begin{subequations}
	\begin{align}
		g_{AB;z} &= 0\,,\\
		\mathcal{A}_{A;z} &= 0\,,\\
		\phi_{;z} &= 0\,.
	\end{align}
	\end{subequations}
	By applying $\nabla_{z}$ on \eqref{eq:kk-decomposition-metric} and using \eqref{eq:nabla-definition-form}, we thus have:
	\begin{equation}
	\label{eq:metric-derivative-intermediate-z}
		\hat{g}_{AB;z} = 0\,.
	\end{equation}
	Therefore,
	\begin{equation}
		\hat{g}\_{[A][B];[C]} = \hat{e}^{A}\_{[A]} \hat{e}^{B}\_{[B]} \hat{e}^{c}\_{[C]}\,\hat{g}_{AB;c}\,.
	\end{equation}
	As this is the first \enquote{serious} usage of the index notation, let's pause here and remind its semantics. Here, \smash{$\hat{e}^{c}\_{[C]}$} is a vector constructed from the frame vector \smash{$\bm{\hat{e}}\_{[C]}$} by taking its projection \smash{$\uppi \bm{\hat{e}}\_{[C]}$}. The contravariant index $\scriptstyle c$ is a usual abstract tensor index. On the other hand, \smash{$\hat{e}^{A}\_{[A]}$} is the frame vector \smash{$\bm{\hat{e}}\_{[A]}$} itself; of course, there is no connection implied by making the frame numbering $\scriptstyle \s{[A]}$ coincide with the contravariant abstract index $\scriptstyle A$.

	Let's start with \eqref{eq:metric-derivative-abc}:
	\begin{equation}
		\hat{g}\_{[a][b];[c]} = \hat{e}^{A}\_{[a]} \hat{e}^{B}\_{[b]} \hat{e}^{c}\_{[c]}\,\hat{g}_{AB;c}\,.
	\end{equation}
	The projection $\uppi$ of \eqref{eq:original-frame-a} is
	\begin{equation}
		\hat{e}^{c}\_{[c]} = \mathrm{e}^{-\alpha\phi} e^{c}\_{(c)}
	\end{equation}
	and to eliminate the frame vectors $\bm{\hat{e}}\_{[a]}$, $\bm{\hat{e}}\_{[b]}$, we use \eqref{eq:original-frame} directly:
	\begin{equation}
		\label{eq:metric-derivative-abc-halfway}
		\hat{g}\_{[a][b];[c]} = \mathrm{e}^{-3\alpha\phi} \left( \hat{g}\_{(a)(b);(c)} - \mathcal{A}\_{(a)} \hat{g}\_{z(b);(c)} - \mathcal{A}\_{(b)} \hat{g}\_{(a)z;(c)} + \mathcal{A}\_{(a)} \mathcal{A}\_{(b)} \hat{g}\_{zz;(c)} \right).
	\end{equation}

	Thanks to \eqref{eq:nabla-projection}, $\uppi \bm{\nabla} \bm{\hat{g}}$ has the same structure as \eqref{eq:kk-decomposition-metric-components-ab}:
	\begin{equation}
		\hat{g}_{ab;c} = \left( \mathrm{e}^{2\alpha\phi} g_{ab} + \mathrm{e}^{2\beta\phi} \mathcal{A}_{a} \mathcal{A}_{b} \right)_{;c}\,.
	\end{equation}
	This projection is enough to express the components $\hat{g}\_{(a)(b);(c)}$ like
	\begin{equation}
		\hat{g}\_{(a)(b);(c)} = \left( \mathrm{e}^{2\alpha\phi} g\_{(a)(b)} + \mathrm{e}^{2\beta\phi} \mathcal{A}\_{(a)} \mathcal{A}\_{(b)} \right)\_{;(c)}\,,
	\end{equation}
	since by definition, $\bm{e}\_{(a)}$ is orthogonal to $\bm{\mathrm{d}}z$:
	\begin{equation}
		e^{z}\_{(a)} = 0\,.
	\end{equation}
	Using \eqref{eq:nabla-definition-metric}, we get
	\begin{equation}
		\label{eq:metric-derivative-intermediate-abc}
		\hat{g}\_{(a)(b);(c)} = \left(\mathrm{e}^{2\alpha\phi}\right)\_{;(c)} g\_{(a)(b)} + \left( \mathrm{e}^{2\beta\phi} \mathcal{A}\_{(a)} \mathcal{A}\_{(b)} \right)\_{;(c)}\,.
	\end{equation}

	To expand the terms with $\hat{g}\_{(a)z;(c)}$ and $\hat{g}\_{zz;(c)}$, we start by realizing that the contraction with $\bm{\partial}_{z}$ commutes with the derivative $\bm{\nabla}$, thanks to \eqref{eq:nabla-definition-vector}.\footnote{Combined with \eqref{eq:nabla-definition-form}, this ensures that using the operator $\bm{\nabla}$ introduces no ambiguity to convention \ref{def:small-indices}.} We are then free to employ \eqref{eq:nabla-projection} again\footnote{Symbolically, $\uppi\mathrm{C}\bm{\nabla}\bm{\hat{g}} = \uppi\bm{\nabla}\mathrm{C}\bm{\hat{g}} = \bm{\nabla}\uppi\mathrm{C}\bm{\hat{g}}$, where $\mathrm{C}$ is the operator of contraction with $\bm{\partial}_{z}$.}, obtaining from \eqref{eq:kk-decomposition-metric-components-az} and \eqref{eq:kk-decomposition-metric-components-zz} respectively:
	\begin{subequations}
	\label{eq:metric-derivative-intermediate}
	\begin{align}
		\label{eq:metric-derivative-intermediate-zab}
		\hat{g}\_{(a)z;(c)} &= \left( \mathrm{e}^{2\beta\phi} \mathcal{A}\_{(a)} \right)\_{;(c)}\,,\\
		\label{eq:metric-derivative-intermediate-zza}
		\hat{g}\_{zz;(c)} &= \left( \mathrm{e}^{2\beta\phi} \right)\_{;(c)}\,.
	\end{align}
	\end{subequations}

	After substituting into \eqref{eq:metric-derivative-abc-halfway}, the terms with $\bm{\mathcal{A}}$ cancel each other out and we arrive at \eqref{eq:metric-derivative-abc}.

	Following the same mechanism, we get
	\begin{subequations}
	\begin{align}
		\hat{g}\_{[z][b];[c]} &= \hat{e}^{A}\_{[z]} \hat{e}^{B}\_{[b]} \hat{e}^{c}\_{[c]}\,\hat{g}_{AB;c} = \mathrm{e}^{(-2\alpha-\beta)\phi} \left( \hat{g}\_{z(b);(c)} - \mathcal{A}\_{(b)} \hat{g}\_{zz;(c)} \right),\\
		\hat{g}\_{[z][z];[c]} &= \hat{e}^{A}\_{[z]} \hat{e}^{B}\_{[z]} \hat{e}^{c}\_{[c]}\,\hat{g}_{AB;c} = \mathrm{e}^{(-\alpha-2\beta)\phi}\,\hat{g}\_{zz;(c)}\,.
	\end{align}
	\end{subequations}
	Substituting again from \eqref{eq:metric-derivative-intermediate}, we arrive at \eqref{eq:metric-derivative-zab} and \eqref{eq:metric-derivative-zza}.

	Lastly, we have
	\begin{equation}
		\hat{g}\_{[A][B];[z]} = \hat{e}^{A}\_{[A]} \hat{e}^{B}\_{[B]} \hat{e}^{c}\_{[z]}\,\hat{g}_{AB;c}\,.
	\end{equation}
	However, the projection $\uppi$ of \eqref{eq:original-frame-z} is
	\begin{equation}
		\hat{e}^{c}\_{[z]} = 0\,,
	\end{equation}
	which leads to \eqref{eq:metric-derivative-ABz}.
\end{proof}

Substituting this result into \eqref{eq:christoffel} from proposition \ref{stat:christoffel}, we get:
\begin{corollary}
\label{stat:reduction-Gamma}
	The frame components of $\bm{\hat{\Gamma}}$ can be expressed as follows:
	\begin{subequations}
	\begin{align}
		\hat{\Gamma}\_{[a][b][c]} &= \left(\mathrm{e}^{-\alpha\phi}\right)\_{;(a)} g\_{(b)(c)} - \left(\mathrm{e}^{-\alpha\phi}\right)\_{;(b)} g\_{(a)(c)} - \left(\mathrm{e}^{-\alpha\phi}\right)\_{;(c)} g\_{(a)(b)}\,,\\
		\hat{\Gamma}\_{[z][a][b]} &= \frac{\mathrm{e}^{(-2\alpha+\beta)\phi}}{2} \left(\mathcal{A}\_{(a);(b)} + \mathcal{A}\_{(b);(a)}\right),\\
		\hat{\Gamma}\_{[a][b][z]} &= \frac{\mathrm{e}^{(-2\alpha+\beta)\phi}}{2} \left(\mathcal{A}\_{(a);(b)} - \mathcal{A}\_{(b);(a)}\right),\\
		\hat{\Gamma}\_{[z][z][a]} &= \mathrm{e}^{(-\alpha-\beta)\phi} \left(\mathrm{e}^{\beta\phi}\right)\_{;(a)}\,,\\
		\hat{\Gamma}\_{[a][z][z]} &= -\mathrm{e}^{(-\alpha-\beta)\phi} \left(\mathrm{e}^{\beta\phi}\right)\_{;(a)}\,,\\
		\hat{\Gamma}\_{[z][z][z]} &= 0\,.
	\end{align}
	\end{subequations}
\end{corollary}

\section{Optical matrix reduction}

\begin{proposition}\label{stat:reduction-L}
	Let $D \ge 4$. Let $\bm{\ell}$ be a null vector field in a $D$-dimensional spacetime with the corresponding $L_{ab}$. Suppose that this spacetime is a Kaluza--Klein reduction of a $(D+1)$-dimensional spacetime. Then the Ricci rotation coefficients $\hat{L}_{AB}$ of this $(D+1)$-dimensional spacetime with respect to $\bm{\hat{\ell}}$, have the following relation to $L_{ab}$:
	\begin{subequations}
	\label{eq:reduction-L}
	\begin{align}
		\label{eq:reduction-L-ab}
		\hat{L}\_{[a][b]} &= \mathrm{e}^{-\alpha\phi} L\_{(a)(b)} - \left( \mathrm{e}^{-\alpha\phi} \right)\_{\!;(0)} g\_{(a)(b)} + \left( \mathrm{e}^{-\alpha\phi} \right)\_{\!;(a)} \ell\_{(b)}\,,\\
		\label{eq:reduction-L-az}
		\hat{L}\_{[a][z]} &= \frac{\mathrm{e}^{(-2\alpha + \beta)\phi}}{2} \mathcal{F}\_{(0)(a)}\,,\\
		\label{eq:reduction-L-za}
		\hat{L}\_{[z][a]} &= \hat{L}\_{[a][z]}\,,\\
		\label{eq:reduction-L-zz}
		\hat{L}\_{[z][z]} &= \mathrm{e}^{(-\alpha - \beta)\phi} \left(\mathrm{e}^{\beta\phi}\right)\_{\!;(0)}\,.
	\end{align}
	\end{subequations}
\end{proposition}
\begin{corollary}\label{stat:reduction-rho}
	For the optical matrix $\rho_{ij}$ of the $D$-dimensional spacetime with respect to some frame, and the optical matrix $\hat{\rho}_{ij}$ of the $(D+1)$-dimensional spacetime with respect to the associated frame, this means:
	\begin{subequations}\label{eq:reduction-rho}
	\begin{align}
		\hat{\rho}_{ij} &= \mathrm{e}^{-\alpha\phi} \rho_{ij} - \left( \mathrm{e}^{-\alpha\phi} \right)\_{\!;(0)} \delta_{ij}\,,\\*
		\hat{\rho}_{iz} &= \hat{\rho}_{zi} = \frac{\mathrm{e}^{(-2\alpha + \beta)\phi}}{2} \mathcal{F}\_{(0)(i)}\,,\\*
		\hat{\rho}_{zz} &= \mathrm{e}^{(-\alpha - \beta)\phi} \left(\mathrm{e}^{\beta\phi}\right)\_{\!;(0)}\,,
	\end{align}
	\end{subequations}
	where $i,j = 2,\dots,D-1$.
\end{corollary}
\begin{corollary}\label{stat:reduction-A}
	In particular, the twist $A_{ij}$ of the $D$-dimensional spacetime with respect to some frame, is proportional to the twist $\hat{A}_{ij}$ of the $(D+1)$-dimensional spacetime with respect to the associated frame; specifically:
	\begin{subequations}
	\begin{alignat}{2}
		\hat{A}_{ij} &= \mathrm{e}^{-\alpha\phi} A_{ij}\,, & \qquad i,j&=2,\dots,D-1\,,\\
		\hat{A}_{iz} &= 0\,, & i&=2,\dots,D-1\,.
	\end{alignat}
	\end{subequations}
\end{corollary}
\begin{remark}
	We can see that $\bm{\ell}$ and $\bm{\hat{\ell}}$ are either both twisting, or both non-twisting (with respect to the corresponding frames\footnote{Generally, we don't require $\bm{\ell}$ to be weakly geodetic; therefore, we are not able to get rid of the frame dependence here -- see corollary \ref{stat:rho-transformation-geodetic}.}).
\end{remark}
\begin{corollary}\label{stat:reduction-kappa}
	It also follows that the non-geodeticity $\bm{\hat{\kappa}}$ of $\bm{\hat{\ell}}$ in the $(D+1)$-dimensional spacetime has the following relation to the non-geodeticity $\bm{\kappa}$ of $\bm{\ell}$ in the $D$-dimensional spacetime:
	\begin{subequations}\label{eq:reduction-kappa}
	\begin{align}
		\hat{\kappa}\^{[a]} &= \mathrm{e}^{-\alpha\phi} \kappa\^{(a)} - \left(\mathrm{e}^{-\alpha\phi}\right)\_{\!;(0)} \ell\^{(a)} = {}\nonumber\\
			{} &= \mathrm{e}^{-2\alpha\phi} \left( \mathrm{e}^{\alpha\phi} \ell\^{(a)} \right)\_{\!;(0)}\,,\\
		\hat{\kappa}\^{[z]} &= 0\,.
	\end{align}
	\end{subequations}
	As a consequence, $\bm{\hat{\ell}}$ is geodetic\footnote{Recall that by geodetic, we always mean affinely geodetic.} in the $(D+1)$-dimensional spacetime if and only if $\mathrm{e}^{\alpha\phi} \bm{\ell}$ is geodetic in the $D$-dimensional spacetime.
\end{corollary}
\begin{remark}
	Obviously, we can relax the affinity condition on the geodesic parametrization, and say that $\bm{\hat{\ell}}$ is weakly geodetic in the $(D+1)$-dimensional spacetime, iff the corresponding $\bm{\ell}$ is weakly geodetic in the $D$-dimensional spacetime.
\end{remark}
\begin{remark}
	As a trivial consequence, a Kaluza--Klein lift of a Kundt spacetime (i.e. one which admits a geodetic null vector field $\bm{\ell}$ with vanishing optical matrix) is also Kundt, whenever boost orders of $\bm{\mathcal{F}}$ and $\bm{\nabla} \phi$ with respect to $\bm{\ell}$ are both at most zero.
\end{remark}
\begin{proof}[Proof of proposition \ref{stat:reduction-L}]
	According to \eqref{eq:original-frame-dual-a}, we have
	\begin{equation}
		\bm{\hat{\varepsilon}}\^{[1]} = \iota_{\star}\big( \mathrm{e}^{\alpha\phi}\bm{\varepsilon}\^{(1)} \big)\,,
	\end{equation}
	which means that the $z$ component of \smash{$\hat{\ell}_{A} = \hat{\varepsilon}\^{[1]}_{A}$} vanishes:
	\begin{equation}
		\hat{\ell}_{z} = 0\,.
	\end{equation}
	We can then use \eqref{eq:nabla-pushforward} to see that
	\begin{subequations}
	\begin{align}
		\hat{\ell}_{A;z} &= 0\,,\\
		\hat{\ell}_{z;B} &= 0\,.
	\end{align}
	\end{subequations}
	In the identity
	\begin{equation}
		\hat{L}_{AB} = \hat{\ell}_{A;B} - \hat{\Gamma}_{CAB} \hat{\ell}\vphantom{\ell}^{C}\,,
	\end{equation}
	the $\hat{\ell}_{A;B}$ can thus be easily prepared:
	\begin{equation}
		\hat{L}\_{[A][B]} = \hat{e}^{a}\_{[A]} \hat{e}^{b}\_{[B]} \left( \mathrm{e}^{\alpha\phi} \ell_{a} \right)_{;b} - \hat{\Gamma}\_{[0][A][B]}
	\end{equation}
	to be expressed in the $D$-dimensional frame components:
	\begin{subequations}
	\label{eq:reduction-L-halfway}
	\begin{align}
		\label{eq:reduction-L-halfway-ab}
		\hat{L}\_{[a][b]} &= \mathrm{e}^{-2\alpha\phi} \left( \mathrm{e}^{\alpha\phi} \ell\_{(a)} \right)\_{;(b)} - \hat{\Gamma}\_{[0][a][b]}\,,\\
		\label{eq:reduction-L-halfway-az}
		\hat{L}\_{[a][z]} &= -\hat{\Gamma}\_{[0][a][z]}\,,\\
		\label{eq:reduction-L-halfway-za}
		\hat{L}\_{[z][a]} &= -\hat{\Gamma}\_{[0][z][a]}\,,\\
		\label{eq:reduction-L-halfway-zz}
		\hat{L}\_{[z][z]} &= -\hat{\Gamma}\_{[0][z][z]}
	\end{align}
	\end{subequations}
	by means of the projection of \eqref{eq:original-frame}:
	\begin{subequations}
	\begin{align}
		\hat{e}\_{[a]}^{a} &= \mathrm{e}^{-\alpha\phi}\,e\_{(a)}^{a}\,, \qquad a=0,\dots,D-1\,,\\
		\hat{e}\_{[z]}^{a} &= 0\,.
	\end{align}
	\end{subequations}
	Substituting from corollary \ref{stat:reduction-Gamma}, we arrive at \eqref{eq:reduction-L-az} and \eqref{eq:reduction-L-zz}, while \eqref{eq:reduction-L-za} is a consequence of the symmetry of $\bm{\hat{\Gamma}}$.

	For \eqref{eq:reduction-L-ab}, we apply the product rule in \eqref{eq:reduction-L-halfway-ab} and identify
	\begin{subequations}
	\begin{gather}
		L\_{(a)(b)} = \ell\_{(a);(b)}\,,\\
		\ell\_{(a)} = g\_{(a)b} \ell^{b} = g\_{(a)(0)}\,.
	\end{gather}
	\end{subequations}
\end{proof}

\section{Weyl tensor reduction}

\subsection{Riemann tensors and their difference}

\begin{definition}
\label{def:riemann-definition}
	For the Riemann tensor $R^{a}_{\hphantom{a}bcd}$ of $(\mathcal{M}, \bm{g})$, we define its extension to $T\mathcal{\hat{M}}$ as the commutator of the extension of $\bm{\nabla}$:
	\begin{equation}
		\label{eq:riemann-definition}
		R^{A}_{\hphantom{A}BCD} v^{B} := \nabla_{C} \nabla_{D} v^{A} - \nabla_{D} \nabla_{C} v^{A}
	\end{equation}
	for any vector field $v^{A}$ on $T\mathcal{\hat{M}}$.
\end{definition}

\begin{proposition}
\label{stat:riemann-pushforward}
	The extended Riemann tensor can be equivalently defined as the pushforward $\iota_{\star}$ of the Riemann tensor of $(\mathcal{M}, \bm{g})$.
\end{proposition}
\begin{proof}
	To see this, one can choose the vector $\bm{v}$ as $\bm{e}\_{(a)}$ and $\bm{\partial}_{z}$, respectively, and use \eqref{eq:nabla-pushforward} and \eqref{eq:nabla-definition-vector}.
\end{proof}
\begin{remark}
	The consistency of definition \ref{def:riemann-definition} can be regarded as a consequence of proposition \ref{stat:riemann-pushforward}.
\end{remark}
\begin{corollary}
\label{stat:riemann-z}
	The $z$ components of $R^{A}_{\hphantom{A}BCD}$ vanish:
	\begin{equation}
	\label{eq:riemann-z}
		R^{A}_{\hphantom{A}BCD} = \iota_{\star} \uppi R^{A}_{\hphantom{A}BCD}\,.
	\end{equation}
\end{corollary}
\begin{proof}
	The subject relation is the identity \eqref{eq:projection-identity} applied on the equivalent definition from proposition \ref{stat:riemann-pushforward}.
\end{proof}

\begin{proposition}
\label{stat:riemann-difference}
	The difference between Riemann tensors of $(\mathcal{\hat{M}}, \bm{\hat{g}})$ and $(\mathcal{M}, \bm{g})$ is
	\begin{equation}
		\label{eq:riemann-difference}
		\hat{R}\vphantom{R}^{A}_{\hphantom{A}BCD} - R^{A}_{\hphantom{A}BCD} = 2 \hat{\Gamma}\vphantom{\Gamma}^{A}_{\hphantom{A}B[D;C]} + 2 \hat{\Gamma}\vphantom{\Gamma}^{A}_{\hphantom{A}E[C} \hat{\Gamma}\vphantom{\Gamma}^{E}_{\hphantom{E}D]B}\,.
	\end{equation}
\end{proposition}
\begin{proof}
	The proof is an application of the common technique that is usually applied on flat $\bm{g}$. We start with the Ricci identity for an arbitrary vector field $v^{A}$:
	\begin{equation}
		\hat{R}\vphantom{R}^{A}_{\hphantom{A}BCD} v^{B} = \hat{\nabla}_{C} \hat{\nabla}_{D} v^{A} - \hat{\nabla}_{D} \hat{\nabla}_{C} v^{A}\,.
	\end{equation}
	After expressing $\bm{\hat{\nabla}}$ using \eqref{eq:nabla-difference}, using the symmetry of $\bm{\hat{\Gamma}}$ and identifying the definition of the extended Riemann tensor \eqref{eq:riemann-definition}, we arrive at the expression that is ultralocal\footnote{Linear dependence can be defined as \emph{ultralocal}, if it is $C^{\infty}\!(\mathcal{\hat{M}}, \mathbb{R})$-linear.} in $\bm{v}$:
	\begin{equation}
		\hat{R}\vphantom{R}^{A}_{\hphantom{A}BCD} v^{B} = \left( R^{A}_{\hphantom{A}BCD} + 2 \hat{\Gamma}\vphantom{\Gamma}^{A}_{\hphantom{A}B[D;C]} + 2 \hat{\Gamma}\vphantom{\Gamma}^{A}_{\hphantom{A}E[C} \hat{\Gamma}\vphantom{\Gamma}^{E}_{\hphantom{E}D]B} \right) v^{B}\,.
	\end{equation}
\end{proof}

For the purposes of algebraic classification, we will be working with the fully covariant form:
\begin{lemma}
\label{stat:riemann-covariant-difference}
	The Riemann tensors are related to each other according to:
	\begin{equation}
		\label{eq:riemann-covariant-difference}
		\hat{R}_{ABCD} = 2 \hat{\Gamma}_{AB[D;C]} + 2 \hat{\Gamma}_{EA[D} \hat{\Gamma}\vphantom{\Gamma}^{E}_{\hphantom{E}C]B} + \hat{g}_{AE} R^{E}_{\hphantom{E}BCD}\,.
	\end{equation}
	In frame components, this means:
	\begin{subequations}
	\label{eq:riemann-reduction}
	\begin{align}
		\label{eq:riemann-reduction-abcd}
		\hat{R}\_{[a][b][c][d]} &= \mathcenter{2\hat{\Gamma}\_{[A][B][[z];[C]]}}{2\hat{\Gamma}\_{[a][b][[d];[c]]}} + \mathcenter{2\hat{\Gamma}\_{[E][A][[z]} \hat{\Gamma}\vphantom{\Gamma}\^{[E]}\_{\hphantom{[E]}[C]][B]}}{2\hat{\Gamma}\_{[E][a][[d]} \hat{\Gamma}\vphantom{\Gamma}\^{[E]}\_{\hphantom{[E]}[c]][b]}} + \mathrm{e}^{-2\alpha\phi} R\_{(a)(b)(c)(d)}\,,\\
		\label{eq:riemann-reduction-ABCz}
		\hat{R}\_{[A][B][C][z]} &= 2\hat{\Gamma}\_{[A][B][[z];[C]]} + 2\hat{\Gamma}\_{[E][A][[z]} \hat{\Gamma}\vphantom{\Gamma}\^{[E]}\_{\hphantom{[E]}[C]][B]}\,.
	\end{align}
	\end{subequations}
\end{lemma}
\begin{proof}
	In \eqref{eq:riemann-difference}, the commutator of index lowering and $\bm{\nabla}$
	\begin{equation}
		\hat{g}_{AE} \hat{\Gamma}\vphantom{\Gamma}^{E}_{\hphantom{E}BD;C} - \hat{\Gamma}_{ABD;C} = -\hat{g}_{AE;C} \hat{\Gamma}\vphantom{\Gamma}^{E}_{\hphantom{E}BD}
	\end{equation}
	is embraced in the second term of \eqref{eq:riemann-covariant-difference}:
	\begin{equation}
		\hat{\Gamma}_{AEC} - \hat{g}_{AE;C} = -\hat{\Gamma}_{EAC}\,.
	\end{equation}

	To infer \eqref{eq:riemann-reduction-abcd}, we first use knowledge of frame components of $\bm{\hat{g}}$ \eqref{eq:original-metric-completeness} and $\bm{g}$:
	\begin{equation}
		g\_{(a)(b)} = \hat{g}\_{[a][b]}\,,
	\end{equation}
	which gives:
	\begin{equation}
		\hat{g}\_{[a][E]} R\^{[E]}\_{\hphantom{[E]}[b][c][d]} = g\_{(a)(e)} R\^{[e]}\_{\hphantom{[e]}[b][c][d]}\,.
	\end{equation}
	Then we change frame for $R^{A}_{\hphantom{A}BCD}$, using \eqref{eq:original-frame-a} and \eqref{eq:original-frame-dual-a} and considering that the $z$ components vanish \eqref{eq:riemann-z}:
	\begin{equation}
		R\^{[e]}\_{\hphantom{[e]}[b][c][d]} = \mathrm{e}^{-2\alpha\phi} R\^{(e)}\_{\hphantom{(e)}(b)(c)(d)}\,.
	\end{equation}

	Finally, as a consequence of \eqref{eq:riemann-z}
	\begin{equation}
		R^{E}_{\hphantom{E}BCz} = 0
	\end{equation}
	and of \eqref{eq:original-frame-z}, we have
	\begin{equation}
		R^{E}\_{\hphantom{E}BC[z]} = 0\,,
	\end{equation}
	which finishes the proof for \eqref{eq:riemann-reduction-ABCz}.
\end{proof}

\subsection{Second derivative of the metric}

To express the relations between frame components of both Riemann tensors using the $D$-di\-men\-sion\-al fields, we need an analogue of proposition \ref{stat:metric-derivative}, with second derivatives of $\bm{\hat{g}}$.

\begin{proposition}
\label{stat:metric-second-derivative}
	The second derivative $\bm{\nabla}^{2} \bm{\hat{g}}$ can be expanded as follows:
	\begin{subequations}
	\label{eq:metric-second-derivative}
	\begin{align}
		\label{eq:metric-second-derivative-abcd}
		&\mathrm{e}^{4\alpha\phi}\,\hat{g}\_{[a][b];[c][d]} = \left( \mathrm{e}^{2\alpha\phi}\right)\_{;(c)(d)} g\_{(a)(b)} + \mathrm{e}^{2\beta\phi} \left( \mathcal{A}\_{(a);(c)} \mathcal{A}\_{(b);(d)} + \mathcal{A}\_{(a);(d)} \mathcal{A}\_{(b);(c)} \right),\\
		\label{eq:metric-second-derivative-zabc}
		&\mathrm{e}^{(3\alpha+\beta)\phi}\,\hat{g}\_{[z][a];[b][c]} = \mathrm{e}^{2\beta\phi} \mathcal{A}\_{(a);(b)(c)} + \left(\mathrm{e}^{2\beta\phi}\right)\_{;(b)} \mathcal{A}\_{(a);(c)} + \left(\mathrm{e}^{2\beta\phi}\right)\_{;(c)} \mathcal{A}\_{(a);(b)}\,,\\
		\label{eq:metric-second-derivative-zzab}
		&\mathrm{e}^{(2\alpha+2\beta)\phi}\,\hat{g}\_{[z][z];[a][b]} = \left(\mathrm{e}^{2\beta\phi}\right)\_{;(a)(b)},\\
		\label{eq:metric-second-derivative-ABCz}
		&\hat{g}\_{[A][B];[C][z]} = 0\,.
	\end{align}
	\end{subequations}
\end{proposition}
\begin{proof}
	By applying $\bm{\nabla}$ on \eqref{eq:metric-derivative-intermediate-z} and using the Riemann tensor \eqref{eq:riemann-definition} to commute the two derivatives, we see that
	\begin{equation}
		\hat{g}_{AB;Cz} = 0\,,
	\end{equation}
	thanks to \eqref{eq:riemann-z}.

	Following the same procedure as in the proof of proposition \ref{stat:metric-derivative}, we immediately see that:
	\begin{subequations}
	\begin{align}
		&\mathrm{e}^{4\alpha\phi}\,\hat{g}\_{[a][b];[c][d]} = \hat{g}\_{(a)(b);(c)(d)} - {\vphantom{g}\mathcal{A}}\_{(a)} \hat{g}\_{z(b);(c)(d)} - {\vphantom{g}\mathcal{A}}\_{(b)} \hat{g}\_{(a)z;(c)(d)} + \mathcal{A}\_{(a)} \mathcal{A}\_{(b)} \hat{g}\_{zz;(c)(d)}\,,\\
		&\mathrm{e}^{(3\alpha+\beta)\phi}\,\hat{g}\_{[z][a];[b][c]} = \hat{g}\_{(a)z;(b)(c)} - {\vphantom{g}\mathcal{A}}\_{(a)} \hat{g}\_{zz;(b)(c)}\,,\\
		&\mathrm{e}^{(2\alpha+2\beta)\phi}\,\hat{g}\_{[z][z];[a][b]} = \hat{g}\_{zz;(a)(b)}\,.
	\end{align}
	\end{subequations}
	Like in the proof of proposition \ref{stat:metric-derivative}, we use \eqref{eq:nabla-projection} to express projections of $\hat{g}_{AB;CD}$ using \eqref{eq:kk-decomposition-metric-components}:
	\begin{equation}
		\label{eq:pi-C-nabla2-commutator}
		\uppi\mathrm{C}\bm{\nabla}^{2}\bm{\hat{g}} = \uppi\bm{\nabla}^{2}\mathrm{C}\bm{\hat{g}} = \bm{\nabla}^{2}\uppi\mathrm{C}\bm{\hat{g}}\,.
	\end{equation}
	This time however, to satisfy the precondition \eqref{eq:nabla-projection-precondition}, we need to show that $\hat{g}_{AB;C}$ is independent of $z$, precisely that
	\begin{equation}
		\label{eq:killing-covariant}
		\mathcal{L}_{\bm{\partial}_z}\!\bm{\nabla} \bm{\hat{g}} = 0\,.
	\end{equation}
	In order not to disturb the flow, let's postpone this discussion to appendix \ref{chap:lie-covariant-commutator}, where we show that \eqref{eq:killing-covariant} indeed holds\footnote{The other possibility would be to realize that we don't need \eqref{eq:killing-covariant} at all, since in \eqref{eq:pi-C-nabla2-commutator}, the projection could be taken only over the indices $\scriptstyle A, B$ of $\hat{g}_{AB;CD}$, with no consequences on the result of \eqref{eq:metric-second-derivative-intermediate}.}. We can then continue by expressing $\bm{\nabla}^{2}\bm{\hat{g}}$ in terms of the $D$-dimensional fields:
	\begin{subequations}
	\label{eq:metric-second-derivative-intermediate}
	\begin{align}
		\hat{g}\_{(a)(b);(c)(d)} &= \left(\mathrm{e}^{2\alpha\phi}\right)\_{;(c)(d)} g\_{(a)(b)} + \left( \mathrm{e}^{2\beta\phi} \mathcal{A}\_{(a)} \mathcal{A}\_{(b)} \right)\_{;(c)(d)},\\
		\hat{g}\_{(a)z;(c)(d)} &= \left( \mathrm{e}^{2\beta\phi} \mathcal{A}\_{(a)} \right)\_{;(c)(d)},\\
		\hat{g}\_{zz;(c)(d)} &= \left( \mathrm{e}^{2\beta\phi} \right)\_{;(c)(d)},
	\end{align}
	\end{subequations}
	which gives rise to
	\begin{subequations}
	\begin{align}
		&\mathrm{e}^{4\alpha\phi}\,\hat{g}\_{[a][b];[c][d]} = \left( \mathrm{e}^{2\alpha\phi}\right)\_{;(c)(d)}  g\_{(a)(b)} + \left( \mathrm{e}^{2\beta\phi} \mathcal{A}\_{(a)} \mathcal{A}\_{(b)} \right)\_{;(c)(d)} - {} \nonumber\\
		& \quad {} - \mathcal{A}\_{(a)}\left( \mathrm{e}^{2\beta\phi} \mathcal{A}\_{(b)} \right)\_{;(c)(d)} - \mathcal{A}\_{(b)}\left( \mathrm{e}^{2\beta\phi} \mathcal{A}\_{(a)} \right)\_{;(c)(d)} + \mathcal{A}\_{(a)}\mathcal{A}\_{(b)}\left( \mathrm{e}^{2\beta\phi} \right)\_{;(c)(d)},\\
		&\mathrm{e}^{(3\alpha+\beta)\phi}\,\hat{g}\_{[z][a];[b][c]} = \left( \mathrm{e}^{2\beta\phi} \mathcal{A}\_{(a)} \right)\_{;(b)(c)} - \mathcal{A}\_{(a)} \left(\mathrm{e}^{2\beta\phi}\right)\_{;(b)(c)},\\
		&\mathrm{e}^{(2\alpha+2\beta)\phi}\,\hat{g}\_{[z][z];[a][b]} = \left(\mathrm{e}^{2\beta\phi}\right)\_{;(a)(b)}.
	\end{align}
	\end{subequations}

	For \eqref{eq:metric-second-derivative-ABCz}, we have
	\begin{equation}
		\hat{g}\_{AB;C[z]} = \mathrm{e}^{-\beta\phi} \hat{g}_{AB;Cz} = 0\,,
	\end{equation}
	again similarly as in the proof of proposition \ref{stat:metric-derivative}.
\end{proof}

For the components of $\bm{\nabla} \bm{\hat{\Gamma}}$ needed to express the Riemann tensor, this means:
\begin{lemma}
	Frame components of $\bm{\nabla} \bm{\hat{\Gamma}}$ can be expressed in terms of frame components of the $D$-dimensional fields and their derivatives. Some of such relations follow:
	\begin{subequations}
	\begin{align}
		&4 \mathrm{e}^{4\alpha\phi}\,\hat{\Gamma}\_{[a][b][[c];[d]]} = \mathrm{e}^{2\beta\phi} \left(\mathcal{F}\_{(a)(b)} \mathcal{F}\_{(c)(d)} - \mathcal{A}\_{(a);(b)} \owedge \mathcal{A}\_{(c);(d)}\right) + {}\nonumber\\
		&\hphantom{4 \mathrm{e}^{4\alpha\phi}\,\hat{\Gamma}\_{[a][b][[c];[d]]} = {}} + \left(\mathrm{e}^{2\alpha\phi}\right)\_{;(a)(b)} \owedge g\_{(c)(d)}\,,\\
		&\mathord{-}2 \mathrm{e}^{(3\alpha+\beta)\phi}\,\hat{\Gamma}\_{[a][b][z];[c]} = \mathrm{e}^{2\beta\phi} \mathcal{F}\_{\!(a)(b);(c)} + \left(\mathrm{e}^{2\beta\phi}\right)\_{;(c)} \mathcal{F}\_{\!(a)(b)} + {}\nonumber\\
		&\hphantom{\mathord{-}2 \mathrm{e}^{(3\alpha+\beta)\phi}\,\hat{\Gamma}\_{[a][b][z];[c]} = {}} + 2 \left(\mathrm{e}^{2\beta\phi}\right)\_{;[(a)} \mathcal{A}\_{(b)];(c)}\,,\\
		&\mathord{-}2 \mathrm{e}^{(-2\alpha-2\beta)\phi}\,\hat{\Gamma}\_{[a][z][z];[b]} = \left(\mathrm{e}^{2\beta\phi}\right)\_{;(a)(b)},\\
		&\hat{\Gamma}\_{[A][B][C];[z]} = 0\,.
	\end{align}
	\end{subequations}
\end{lemma}
\begin{definition}
	By $\owedge$, we mean a \enquote{generalized Kulkarni--Nomizu product}:
	\begin{align}
		A_{ab} \owedge B_{cd} &:= A_{ac} B_{db} + A_{bd} B_{ca} - A_{ad} B_{cb} - A_{bc} B_{da} = {}\nonumber\\
		&\hphantom{:}= 2\left(A_{a[c} B_{d]b} - A_{b[c} B_{d]a}\right).
	\end{align}
\end{definition}
\begin{proof}
	The relations are a direct consequence of \eqref{eq:christoffel} from proposition \ref{stat:christoffel}, differentiated by $\bm{\nabla}$ and expressed as frame components according to proposition \ref{stat:metric-second-derivative}.
\end{proof}

\subsection{Riemann tensor frame components}

Having expressed the frame components of $\bm{\hat{\Gamma}}$ and $\bm{\nabla} \bm{\hat{\Gamma}}$, it is now straightforward to compute the components of the Riemann tensor by means of lemma \ref{stat:riemann-covariant-difference}.
\begin{proposition}
	Frame components of the two Riemann tensors $R_{abcd}$ and $\hat{R}_{ABCD}$ have the following relation to each other:
	\begin{subequations}
	\label{eq:riemann-KK}
	\begin{align}
		\label{eq:riemann-KK-a}
		\mathrm{e}^{2\alpha\phi} \hat{R}\_{[a][b][c][d]} ={}& R\_{(a)(b)(c)(d)} + {}\nonumber\\*
		& {} + \tfrac{\mathrm{e}^{(-2\alpha+2\beta)\phi}}{8} \left(\mathcal{F}\_{(a)(b)} \owedge \mathcal{F}\_{(c)(d)} - 4 \mathcal{F}\_{(a)(b)}\mathcal{F}\_{(c)(d)}\right) + {}\nonumber\\*
		& {} + \left(\mathrm{e}^{\alpha\phi} \left(\mathrm{e}^{-\alpha\phi}\right)\_{;(a)(b)} - \tfrac{\alpha^{2}}{2} g\^{(e)(f)} \phi\_{;(e)} \phi\_{;(f)} g\_{(a)(b)} \right) \owedge g\_{(c)(d)}\,,\\*[5pt]
		\label{eq:riemann-KK-b}
		\mathrm{e}^{2\alpha\phi} \hat{R}\_{[a][b][c][z]} ={}& \mathrm{e}^{(-\alpha+\beta)\phi}\,\Big(-\tfrac{1}{2} \mathcal{F}\_{(a)(b);(c)} + \alpha g\^{(e)(f)} \phi\_{;(e)} \mathcal{F}\_{(f)[(a)} g\_{(b)](c)} + {}\nonumber\\*
		& \hphantom{\mathrm{e}^{(-\alpha+\beta)\phi}\,\Big(} + (\beta-\alpha)\left(\phi\_{;[(a)} \mathcal{F}\_{(b)](c)} - \phi\_{;(c)} \mathcal{F}\_{(a)(b)}\right)\Big)\,,\\*[5pt]
		\label{eq:riemann-KK-c}
		\mathrm{e}^{2\alpha\phi} \hat{R}\_{[a][z][c][z]} ={}& \tfrac{\beta\mathrm{e}^{(2\alpha-\beta)\phi}}{2\alpha-\beta} \left(\mathrm{e}^{(-2\alpha+\beta)\phi}\right)\_{;(a)(c)} + \tfrac{\mathrm{e}^{(-2\alpha+2\beta)\phi}}{4} \mathcal{F}^{2}\_{(a)(c)} - {}\nonumber\\*
		& {} - \alpha\beta g\^{(b)(d)} \phi\_{;(b)} \phi\_{;(d)} g\_{(a)(c)}\,.
	\end{align}
	\end{subequations}
\end{proposition}
\begin{proof}
	By direct substitution from corollary \ref{stat:reduction-Gamma}, we get:
	\begin{align}
		2\hat{\Gamma}\_{[E][a][[d]} \hat{\Gamma}\vphantom{\Gamma}\^{[E]}\_{\hphantom{[E]}[c]][b]} ={}& -\tfrac{1}{8} \mathrm{e}^{(-4\alpha + 2\beta)\phi} (\mathcal{A}\_{(a);(b)} + \mathcal{A}\_{(b);(a)})\owedge(\mathcal{A}\_{(c);(d)} + \mathcal{A}\_{(d);(c)}) - {}\nonumber\\
		& {} - \tfrac{\alpha^{2}}{2} \mathrm{e}^{-2\alpha\phi} g\^{(e)(f)} \phi\_{;(e)} \phi\_{;(f)} g\_{(a)(b)} \owedge g\_{(c)(d)} + {}\nonumber\\
		& {} + 3 \left( \mathrm{e}^{-\alpha\phi} \right)\_{\!;(a)} \left( \mathrm{e}^{-\alpha\phi} \right)\_{\!;(b)} \owedge g\_{(c)(d)}
	\end{align}
	Together with following identities:
	\begin{subequations}
	\begin{gather}
		4 \mathcal{A}\_{(a);(b)} \owedge \mathcal{A}\_{(c);(d)} - (\mathcal{A}\_{(a);(b)} + \mathcal{A}\_{(b);(a)}) \owedge (\mathcal{A}\_{(c);(d)} + \mathcal{A}\_{(d);(c)}) = \mathcal{F}\_{(a)(b)} \owedge \mathcal{F}\_{(c)(d)}\,,\\
		3 \left(\mathrm{e}^{-\alpha\phi}\right)\_{\!;(a)} \left(\mathrm{e}^{-\alpha\phi}\right)\_{\!;(b)} - \frac{1}{2} \left(\mathrm{e}^{2\alpha\phi}\right)\_{\!;(a)(b)} = \mathrm{e}^{\alpha\phi} \left(\mathrm{e}^{-\alpha\phi}\right)\_{\!;(a)(b)},
	\end{gather}
	\end{subequations}
	this leads to \eqref{eq:riemann-KK-a}.

	Using the same technique, we obtain \eqref{eq:riemann-KK-b} and \eqref{eq:riemann-KK-c}.
\end{proof}

\begin{remark}
	Setting $\hat{R}_{AB} = 0$ in
	\begin{subequations}
	\begin{align}
		\hat{R}\_{[a][c]} &= g\^{(b)(d)} \hat{R}\_{[a][b][c][d]} + \hat{R}\_{[a][z][c][z]}\,,\\
		\hat{R}\_{[a][z]} &= -g\^{(b)(c)} \hat{R}\_{[a][b][c][z]}\,,\\
		\hat{R}\_{[z][z]} &= g\^{(a)(c)} \hat{R}\_{[a][z][c][z]}\,,
	\end{align}
	\end{subequations}
	we obtain the Einstein equations \eqref{eq:reduced-einstein-eq} for vacuum $\bm{\hat{g}}$. Specifically, the frame components of the three equations \eqref{eq:reduced-einstein-eq} correspond to the following linear combinations:
	\begin{subequations}
	\begin{align}
		\mathrm{e}^{2\alpha\phi} \hat{R}\_{[a][b]} - \tfrac{\alpha}{\beta} \mathrm{e}^{2\alpha\phi} \hat{R}\_{[z][z]} g\_{(a)(b)} &= 0\,,\\
		2 \mathrm{e}^{((D-1)\alpha + 2\beta)\phi} \hat{R}\_{[a][b]} &= 0\,,\\
		4 \mathrm{e}^{2\alpha\phi} \hat{R}\_{[z][z]} &= 0\,.
	\end{align}
	\end{subequations}
\end{remark}
\begin{definition}
	By $\hat{R}_{AB}$ and $R_{ab}$, we mean the Ricci tensors of the original and reduced spacetime:
	\begin{align}
		\hat{R}_{AB} &= \hat{R}\vphantom{R}^{C}_{\hphantom{C}ACB}\,, & R_{ab} &= R^{c}_{\hphantom{c}acb}\,.
	\end{align}
\end{definition}

\subsection{Weyl tensor frame components}

Finally, we can express the $(D+1)$-dimensional Riemann tensor (which is equal to the Weyl tensor) in terms of the $D$-dimensional Weyl tensor.

\begin{proposition}
	\label{stat:weyl-KK}
	Let $D \ge 4$. Then the Weyl tensor $\hat{C}_{ABCD}$ of a $(D+1)$-dimensional vacuum spacetime has the following relation to the Weyl tensor $C_{abcd}$ of its $D$-dimensional Kaluza--Klein reduction:
	\begin{subequations}\label{eq:weyl-KK}
	\begin{align}
		\label{eq:weyl-KK-a}
		&\mathrm{e}^{2\alpha\phi} \hat{C}\_{[a][b][c][d]} = C\_{(a)(b)(c)(d)} + \tfrac{\mathrm{e}^{(-2\alpha+2\beta) \phi}}{8} \left( \mathcal{F}\_{(a)(b)} \owedge \mathcal{F}\_{(c)(d)} - 4 \mathcal{F}\_{(a)(b)} \mathcal{F}\_{(c)(d)} \right) + {}\nonumber\\
		&\quad {} + \Bigg(\tfrac{\mathrm{e}^{(-2\alpha+2\beta)\phi}}{2(D-2)} \mathcal{F}^{2}\_{(a)(b)} + \tfrac{\beta \mathrm{e}^{(2\alpha-\beta) \phi}}{(D-2)(\beta-2\alpha)} \left( \mathrm{e}^{(-2\alpha+\beta)\phi} \right)\_{\!;(a)(b)} - {}\nonumber\\
		&\quad\hphantom{{} + \Bigg(} - \tfrac{9 \beta \mathrm{e}^{-\left(\frac{D-4}{3}\alpha+\beta\right) \phi}}{2(D-1)(D-2)((D-4) \alpha + 3\beta)} g\^{(e)(f)} \left( \mathrm{e}^{\left(\frac{D-4}{3}\alpha+\beta\right) \phi} \right)\_{\!;(e)(f)} g\_{(a)(b)} \Bigg) \owedge g\_{(c)(d)}\,,\\
		\label{eq:weyl-KK-b}
		&\mathord{-}2\mathrm{e}^{(3\alpha-\beta)\phi} \hat{C}\_{[a][b][c][z]} = (\beta-\alpha)(\phi\_{;(b)} \mathcal{F}\_{(a)(c)} + \phi\_{;(a)} \mathcal{F}\_{(c)(b)} + 2 \phi\_{;(c)} \mathcal{F}\_{(a)(b)}) + \vphantom{\frac{1}{2}}\nonumber\\
		&\hphantom{\mathord{-}2\mathrm{e}^{(3\alpha-\beta)\phi} \hat{C}\_{[a][b][c][z]} = {}} + \alpha g\^{(e)(f)} \phi\_{;(e)} \left( \mathcal{F}\_{(a)(f)} g\_{(b)(c)} + \mathcal{F}\_{(f)(b)} g\_{(a)(c)}\right) + {}\nonumber\\
		&\hphantom{\mathord{-}2\mathrm{e}^{(3\alpha-\beta)\phi} \hat{C}\_{[a][b][c][z]} = {}} + \mathcal{F}\_{(a)(b);(c)}\,,\\[2pt]
		\label{eq:weyl-KK-c}
		&\mathrm{e}^{2\alpha\phi} \hat{C}\_{[a][z][c][z]} = \tfrac{\beta \mathrm{e}^{(2\alpha-\beta)\phi}}{2\alpha-\beta} \left( \mathrm{e}^{(-2\alpha+\beta)\phi}\right)\_{\!;(a)(c)} - \alpha\beta g\^{(b)(d)} \phi\_{;(b)} \phi\_{;(d)} g\_{(a)(c)} + {}\nonumber\\
		&\hphantom{\mathrm{e}^{2\alpha\phi} \hat{C}\_{[a][z][c][z]} = {}} + \tfrac{\mathrm{e}^{(-2\alpha+2\beta)\phi}}{4} \mathcal{F}^{2}\_{(a)(c)}\,.
	\end{align}
	\end{subequations}
\end{proposition}
\begin{proof}
	As the lifted spacetime is vacuum,
	\begin{subequations}
	\begin{align}
		\hat{R}_{ABCD} &= \hat{C}_{ABCD}\,,\\
		R_{abcd} &= C_{abcd} + g_{ab} \owedge S_{cd}\,.
	\end{align}
	\end{subequations}
	Here, $\bm{S}$ is the Schouten tensor of the reduced spacetime:
	\begin{equation}
		S_{ab} \equiv \frac{1}{D-2} \left(R_{ab} - \frac{R}{2\left(D-1\right)}\,g_{ab}\right),\qquad
		\text{where }R \equiv g^{ab} R_{ab}\,,
	\end{equation}
	which can be computed by means of contractions of \eqref{eq:riemann-KK} (those contractions are put in \eqref{eq:reduced-einstein-eq}):
	\begin{subequations}
	\begin{align}
		&(D-2)\,S_{ab} = \left(\beta+\left(D-2\right)\alpha\right) \phi_{;ab} + \left(\beta^{2}-2\alpha\beta-\left(D-2\right)\alpha^{2}\right) \phi_{;a}\phi_{;b} + {}\nonumber\\
		&\qquad {} + \tfrac{\left(D-2\right)\alpha^{2} + 2\alpha\beta}{2}\,g^{ef} \phi_{;e}\phi_{;f} g_{ab} + \tfrac{\mathrm{e}^{(-2\alpha+2\beta)\phi}}{2} \left(\mathcal{F}^{2}_{ab} - \tfrac{3}{4(D-1)}\,\mathcal{F}^{2} g_{ab}\right).
	\end{align}
	\end{subequations}
\end{proof}

Table \ref{tab:weyl-generators} shows generating sets for $(D+1)$-dimensional Weyl tensor frame components with given boost weight \emph{(b.w.)} \citep[\bgroup\it cf.\egroup][table III]{alignment-special-tensors}.
\begin{table}[htbp]
\centering
\caption{Generating Weyl frame components}\label{tab:weyl-generators}
{\edef\baselinestrut{\vrule width0pt height\the\baselineskip\relax}
\begin{tabular}{*{5}{c}}
	\toprule
	\bw2 & \bw1 & \bw0 & \bw{-1} & \bw{-2}\\
	\midrule
	$\hat{C}\_{[0][i][0][j]}$ & $\hat{C}\_{[0][i][j][k]}$ & $\hat{C}\_{[0][i][1][j]}$ & $\hat{C}\_{[1][i][j][k]}$ & $\hat{C}\_{[1][i][1][j]}$\\
	\baselinestrut$\hat{C}\_{[0][i][0][z]}$ & $\hat{C}\_{[0][z][i][z]}$ & $\hat{C}\_{[i][j][k][l]}$ & $\hat{C}\_{[1][z][i][z]}$ & $\hat{C}\_{[1][i][1][z]}$\\
	\baselinestrut& $\hat{C}\_{[0][i][j][z]}$ & $\hat{C}\_{[i][j][k][z]}$ & $\hat{C}\_{[1][i][j][z]}$ &\\
	\baselinestrut&& $\hat{C}\_{[0][1][i][z]}$ &&\\
	\bottomrule
\end{tabular}}
\end{table}
The rest of the components can be expressed as linear combination of such components with the corresponding boost weight, by means of the Weyl tensor symmetries.

\section{Kaluza--Klein alignment conditions}

Assuming the lifted spacetime is vacuum, we would now like to determine the conditions necessary and sufficient for it to be of the same or more special algebraic type as the reduced spacetime, with respect to corresponding null direction. Specifically, for $\mathrm{bo}_{\langle\bm{\hat{\ell}}\rangle}\,\bm{\hat{C}}$ to be at most $b$ provided that $\mathrm{bo}_{\langle\bm{\ell}\rangle}\,\bm{C}$ is at most $b$.\footnote{It should be noted that one can also ask whether $\bm{\hat{g}}$ is more special than $\bm{g}$. For a subset of the cases studied here, this question has been addressed by \cite{ortaggio-warp}.} For ease of expression, we will call this condition a \emph{Weyl alignment preserving Kaluza--Klein lift} to vacuum.

\subsection{Type I and II Kaluza--Klein reductions}

Requiring the corresponding $(D+1)$-dimensional Riemann tensor components to vanish, we arrive at the following results for Kaluza--Klein reductions of types I and II, where for the sake of brevity, we establish the following function of the scalar field:
\begin{equation}
	\Phi :=
	\begin{cases}
		\phi & \text{for } \beta = 2\alpha\,,\\[2pt]
		\frac{\mathrm{e}^{(\beta-2\alpha)\phi}}{\beta-2\alpha} & \text{otherwise.}
	\end{cases}
\end{equation}

\begin{proposition}\label{stat:I}
	Let $D \ge 4$. Let $\bm{\ell}$ be a WAND of a $D$-dimensional spacetime. Suppose that this spacetime is a Kaluza--Klein reduction of a $(D+1)$-dimensional vacuum spacetime. Then $\bm{\hat{\ell}}$ is a WAND of this $(D+1)$-dimensional spacetime, iff the following holds locally:
	\begin{subequations}\label{eq:conditions-I}
	\begin{align}
		\label{eq:conditions-Ia}
		\Phi\_{;(0)(0)} &= 0\,,\\
		\label{eq:conditions-Ib}
		\mathcal{F}\_{(0)(i)} &= 0\,,\\
		\label{eq:conditions-Ic}
		\mathcal{F}\_{(i)(j)} \kappa\^{(j)} &= 0\,,\\
		\label{eq:conditions-Id}
		\mathcal{F}\_{(0)(1)} \kappa\^{(i)} &= 0\,.
	\end{align}
	\end{subequations}
\end{proposition}
Before we prove the proposition, let's recall an elementary linear algebra observation:
\begin{lemma}
\label{stat:kernel-decomposition}
	The kernel of a real matrix with vanishing symmetric traceless part, is a subset of the intersection of the kernels of its symmetric and its antisymmetric part.
\end{lemma}
\begin{proof}[Proof of proposition \ref{stat:I}]
	The relevant equations from proposition \ref{stat:weyl-KK} correspond to the \bw2 column of table \ref{tab:weyl-generators}:
	\begin{subequations}
	\begin{align}
		\mathrm{e}^{2\alpha\phi} \hat{C}\_{[0][i][0][j]} ={}& C\_{[0][i][0][j]} + \tfrac{\beta \mathrm{e}^{(2\alpha-\beta)\phi}}{(D-2)(\beta-2\alpha)} \left( \mathrm{e}^{(-2\alpha+\beta)\phi} \right)\_{;(0)(0)} g\_{(i)(j)} + {}\nonumber\\
		& {} + \tfrac{\mathrm{e}^{(-2\alpha+2\beta)\phi}}{2(D-2)} \mathcal{F}^{2}\_{(0)(0)} g\_{(i)(j)} - \tfrac{3 \mathrm{e}^{(-2\alpha+2\beta)\phi}}{4} \mathcal{F}\_{(0)(i)} \mathcal{F}\_{(0)(j)}\,,\\
		\hat{C}\_{[0][i][0][z]} ={}& \mathord{-}\tfrac{\mathrm{e}^{(-3\alpha+\beta)\phi}}{2} \left(3 (\beta-\alpha) \phi\_{;(0)} \mathcal{F}\_{(0)(i)} + \mathcal{F}\_{(0)(i);(0)} \right).
	\end{align}
	\end{subequations}

	Given the algebraic specialness assumption on $\bm{C}$, the traceless part of the equation
	\begin{equation}
		\hat{C}\_{[0][i][0][j]} = 0
	\end{equation}
	is equivalent to the condition \eqref{eq:conditions-Ib}.\footnote{This wouldn't be necessarily true for $D=3$, which we don't consider here.} The trace of that equation is then equivalent to the condition \eqref{eq:conditions-Ia}. The equation
	\begin{equation}
		\hat{C}\_{[0][i][0][z]} = 0
	\end{equation}
	is then equivalent to
	\begin{equation}
		\mathcal{F}\_{(0)(i);(0)} = 0\,.
	\end{equation}
	However, thanks to \eqref{eq:conditions-Ib} we can use the Leibniz rule to express $\mathcal{F}\_{(0)(i);(0)}$ using derivatives only on the frame vectors:
	\begin{align}
		\mathcal{F}\_{(0)(i);(0)} &= \left(\mathcal{F}_{ab} \ell^{a} e\_{(i)}^{b}\right)\_{\!;(0)} - \mathcal{F}_{ab} \left(\ell^{a} e\_{(i)}^{b}\right)\_{\!;(0)} = -\mathcal{F}\_{(a)(i)} \kappa\^{(a)} - \mathcal{F}\_{(0)(1)} \varepsilon\^{(1)}_{b} \left( e\_{(i)}^{b} \right)\_{\!;(0)} = {}\nonumber\\
		 &= \left( \mathcal{F}\_{(i)(j)} + \mathcal{F}\_{(0)(1)} g\_{(i)(j)} \right) \kappa\^{(j)}.
	\end{align}
	Equivalence with conditions \eqref{eq:conditions-Ic} and \eqref{eq:conditions-Id} is then a consequence of lemma \ref{stat:kernel-decomposition}.
\end{proof}
\begin{remark}
	The condition \eqref{eq:conditions-Ia} is equivalent to claiming that
	\begin{equation}
		\mathrm{bo}_{\langle\bm{\ell}\rangle}\,\bm{\nabla}\bm{\nabla}\Phi < 2\,,
	\end{equation}
	while the condition \eqref{eq:conditions-Ib} is equivalent to claiming that
	\begin{equation}
		\mathrm{bo}_{\langle\bm{\ell}\rangle}\,\bm{\mathcal{F}} < 1\,.
	\end{equation}
	The conditions \eqref{eq:conditions-Ic} and \eqref{eq:conditions-Id} can be replaced by one condition
	\begin{equation}
		\mathrm{bo}_{\langle\bm{\ell}\rangle}\,\bm{\nabla}\bm{\mathcal{F}} < 2\,.
	\end{equation}
\end{remark}

\begin{proposition}\label{stat:II}
	Let $D \ge 4$. Let $\bm{\ell}$ be a mWAND of a $D$-dimensional spacetime. Suppose that this spacetime is a Kaluza--Klein reduction of a $(D+1)$-dimensional vacuum spacetime. Then $\bm{\hat{\ell}}$ is a mWAND of this $(D+1)$-dimensional spacetime, iff the following holds locally:
	\begin{subequations}\label{eq:conditions-II}
	\begin{align}
		\label{eq:conditions-IIa}
		\Phi\_{;(0)(i)} &= 0\,,\\
		\label{eq:conditions-IIb}
		\Big(\mathcal{F}\_{(i)(k)} + \mathcal{F}\_{(0)(1)} g\_{(i)(k)}\Big)\,\Big(L\^{(k)}\_{\hphantom{(k)}(j)} + \alpha\phi\_{;(0)} \delta\^{(k)}\_{(j)}\Big) &= \beta \phi\_{;(0)} \mathcal{F}\_{(i)(j)}\,,
	\end{align}
	\end{subequations}
	together with the conditions \eqref{eq:conditions-I}.
\end{proposition}
\begin{proof}
	We again use proposition \ref{stat:weyl-KK}, this time requiring to vanish both \bw2 and $1$ components of $\bm{\hat{C}}$. Under the assumptions of \eqref{eq:conditions-I}, two additional components from table \ref{tab:weyl-generators} are:
	\begin{subequations}\label{eq:conditions-II-proof}
	\begin{align}
		\label{eq:conditions-II-proof-a}
		\hat{C}\_{[0][z][i][z]} &= \tfrac{\beta \mathrm{e}^{-\beta\phi}}{2\alpha-\beta} \left( \mathrm{e}^{(-2\alpha+\beta)\phi}\right)\_{\!;(0)(i)},\\
		\label{eq:conditions-II-proof-b}
		\hat{C}\_{[0][i][j][z]} &= -\tfrac{\mathrm{e}^{(-3\alpha+\beta)\phi}}{2} \Big( (\alpha-\beta)\phi\_{;(0)} \mathcal{F}\_{(i)(j)} + \alpha \phi\_{;(0)} \mathcal{F}\_{(0)(1)} g\_{(i)(j)} + \mathcal{F}\_{(0)(i);(j)} \Big)\,.
	\end{align}
	\end{subequations}
	The components $\hat{C}\_{[0][i][j][k]}$ vanish identically, assuming that the other $\text{b.w.} \ge 1$ components of $\bm{\hat{C}}$ vanish.

	The equation
	\begin{equation}
		\hat{C}\_{[0][z][i][z]} = 0
	\end{equation}
	is equivalent to the condition \eqref{eq:conditions-IIa}.

	Similarly as in the proof of proposition \ref{stat:I}, we can exploit \eqref{eq:conditions-Ib} while employing the Leibniz rule to express $\mathcal{F}\_{(0)(i);(j)}$ using derivatives only on the frame vectors:
	\begin{equation}\label{eq:F-0ij-leibniz}
		\mathcal{F}\_{(0)(i);(j)} = \mathcal{F}\_{(i)(k)} L\^{(k)}\_{\hphantom{(k)}(j)} + \mathcal{F}\_{(0)(1)} L\_{(i)(j)}\,.
	\end{equation}
	Substituting to \eqref{eq:conditions-II-proof-b}, we arrive at the condition \eqref{eq:conditions-IIb}.
\end{proof}
\begin{remark}
	The conditions \eqref{eq:conditions-Ia}, \eqref{eq:conditions-IIa} are equivalent to claiming that
	\begin{equation}
		\mathrm{bo}_{\langle\bm{\ell}\rangle}\,\bm{\nabla}\bm{\nabla}\Phi < 1\,.
	\end{equation}
\end{remark}
\begin{remark}
	In the following Ricci identity:
	\begin{equation}
		\Phi\_{;(0)(0)(i)} - \Phi\_{;(0)(i)(0)} = R\^{(a)}\_{\hphantom{(a)}(0)(0)(i)} \Phi\_{;(a)}\,,
	\end{equation}
	we can express the left hand side using the Leibniz rule, thanks to \eqref{eq:conditions-Ia} and \eqref{eq:conditions-IIa}:
	\begin{subequations}
	\begin{align}
		\Phi\_{;(0)(0)(i)} &= 0\,,\\
		\Phi\_{;(0)(i)(0)} &= -\left(\Phi\_{;(i)(j)} - \Phi\_{;(0)(1)} g\_{(i)(j)}\right) \kappa\^{(j)}\,,
	\end{align}
	\end{subequations}
	while in the Ricci decomposition of the right hand side, we get rid of $\bm{C}$ due to the algebraic speciality condition and express the $R_{ab}$ using the Einstein equation \eqref{eq:reduced-einstein-eq-a}:
	\begin{align}
		R\_{(a)(0)(0)(i)} &= \tfrac{1}{D-2} \left( g\_{(a)(0)} R\_{(0)(i)} - g\_{(a)(i)} R\_{(0)(0)} \right) = {}\nonumber\\
		{} &= \alpha (\alpha-\beta) \left( g\_{(a)(0)} \phi\_{;(0)} \phi\_{;(i)} - g\_{(a)(i)} \phi\_{;(0)} \phi\_{;(0)} \right),
	\end{align}
	arriving at the following secondary constraint:
	\begin{equation}\label{eq:II-secondary-constraint}
		\left(\Phi\_{;(i)(j)} - \Phi\_{;(0)(1)} g\_{(i)(j)}\right) \kappa\^{(j)} = 0\,.
	\end{equation}
\end{remark}
\begin{remark}
	Antisymmetrizing \eqref{eq:conditions-IIb}, we get:
	\begin{equation}\label{eq:conditions-II-consequence-1}
		\Big(\mathrm{e}^{(2\beta-2\alpha)\phi} \mathcal{F}\_{(i)(j)}\Big)\vphantom{\mathcal{F}}\_{;(0)} = 0\,.
	\end{equation}
	On the other hand, from the trace of \eqref{eq:conditions-IIb}, we obtain:
	\begin{equation}\label{eq:conditions-II-consequence-2}
		\Big(\mathrm{e}^{(3\beta-2\alpha)\phi} \mathcal{F}\_{(0)(1)}\Big)\vphantom{\mathcal{F}}\_{;(0)} = 0\,,
	\end{equation}
	using \eqref{eq:reduced-einstein-eq-b}.
\end{remark}
\begin{remark}
	The conditions \eqref{eq:conditions-I} and \eqref{eq:conditions-II} imply that
	\begin{equation}
		\mathcal{F}\_{(a)(i)} \kappa\^{(i)} = 0\,.
	\end{equation}
\end{remark}
\begin{proof}
	Thanks to \eqref{eq:conditions-Ib} and \eqref{eq:conditions-Ic}, it only remains to show that $\mathcal{F}\_{(1)(i)} \kappa\^{(i)} = 0$. Without loss of generality, we can assume $\mathcal{F}\_{(0)(1)} = 0$, thanks to \eqref{eq:conditions-Id}. But under such assumption, we can use the Leibniz rule on \eqref{eq:conditions-II-consequence-2}, obtaining
	\begin{equation}\label{eq:F-010-leibniz}
		0 = \mathcal{F}\_{(0)(1);(0)} = \mathcal{F}\_{(1)(i)} \kappa\^{(i)}.
	\end{equation}
\end{proof}
\begin{remark}
	The special case of $\bm{\mathcal{A}} = 0$ and $\phi = \text{const.}$ was studied earlier in \cite{ortaggio-warp} as a Killing-case subset of metrics with a warped extra dimension. For this case, the results correspond with those covered by propositions \ref{stat:I} and \ref{stat:II}: both \eqref{eq:conditions-I} and \eqref{eq:conditions-II} are trivially satisfied.
\end{remark}
\begin{example}
	Algebraic specialness of Kaluza--Klein lift of the 4-di\-men\-sion\-al Robinson--Trautman spacetime is discussed in \cite[section 3.2]{real-2013}:
	\begin{equation}
		\bm{g} = -U\,\bm{\mathrm{d}}u^{2} - \bm{\mathrm{d}}u\!\vee\!\bm{\mathrm{d}}r + \tfrac{r^2}{G^2}\,(\bm{\mathrm{d}}x^2 + \bm{\mathrm{d}}y^2)\,,\\
	\end{equation}
	using Kaluza--Klein fields in the following form:
	\begin{subequations}
	\begin{align}
		\mathcal{A} &= C\,\bm{\mathrm{d}}u\,,\\
		\phi &= 0\,,
	\end{align}
	\end{subequations}
	where
	\begin{subequations}
	\begin{align}
		U &= -\tfrac{k}{r} + G^{2} \Delta\ln G - 2r\,(\ln G)_{,u}\,,\\
		0 &= k_{,x} = k_{,y} = k_{,r} = G_{,r} = C_{,r}
	\end{align}
	\end{subequations}
	and where we have introduced the spatial Laplace operator:
	\begin{equation}
		\Delta := \tfrac{\partial^{2}}{\partial x^{2}} + \tfrac{\partial^{2}}{\partial y^{2}}\,.
	\end{equation}
	Suppose that $\bm{g}$ is a Kaluza--Klein reduction of a vacuum spacetime. That means that the equations \eqref{eq:reduced-einstein-eq} must hold:
	\begin{subequations}
	\begin{align}
		G^{2} \left( C_{,x} C_{,x} + C_{,y} C_{,y} \right) &= -2k_{,u} + 6 k (\ln G)_{,u} + G^{2} \Delta(G^2 \Delta \ln G)\,,\\
		\Delta C &= 0\,.
	\end{align}
	\end{subequations}
	The null vector
	\begin{equation}
		\bm{\ell} := \bm{\partial}_{r}
	\end{equation}
	is a mWAND of $\bm{g}$. We can now introduce the following null frame:
	\begin{subequations}
	\begin{align}
		\bm{e}\_{(0)} &= \bm{\partial}_{r}\,,\\
		\bm{e}\_{(1)} &= -\bm{\partial}_{u} + \tfrac{U}{2} \bm{\partial}_{r}\,,\\
		\bm{e}\_{(2)} &= \tfrac{G}{r} \bm{\partial}_{x}\,,\\
		\bm{e}\_{(3)} &= \tfrac{G}{r} \bm{\partial}_{y}
	\end{align}
	\end{subequations}
	and see that the equations \eqref{eq:conditions-I} and \eqref{eq:conditions-II} hold, because the only independent nonzero components of the Maxwell tensor are $\mathcal{F}\_{(1)(i)}$; specifically:
	\begin{subequations}
	\begin{align}
		\mathcal{F}\_{(1)(2)} &= \frac{G}{r}\,C_{,x}\,,\\*
		\mathcal{F}\_{(1)(3)} &= \frac{G}{r}\,C_{,y}\,.
	\end{align}
	\end{subequations}
	Therefore, according to proposition \ref{stat:II}, $\bm{\hat{\ell}} := \bm{\ell}$ is a mWAND\footnote{We will see later in the equation \eqref{eq:conditions-III-d} that $\bm{\ell}$ and $\bm{\hat{\ell}}$ can only be both 3-WANDs (each of its corresponding spacetime), if $\mathcal{F}_{ab} = 0$.} of $\bm{\hat{g}}$\,.
\end{example}

\subsection{Type III Kaluza--Klein reduction}\label{chap:III}

Having a method for deducing the Kaluza--Klein reduction alignment conditions, let's work out the analog for type III.

\begin{proposition}\label{stat:III}
	Let $D \ge 4$. Let $\bm{\ell}$ be a 3-WAND of a $D$-dimensional spacetime. Suppose that this spacetime is a Kaluza--Klein reduction of a $(D+1)$-dimensional vacuum spacetime. Then $\bm{\hat{\ell}}$ is a 3-WAND of this $(D+1)$-dimensional spacetime, iff the following holds locally:
	\begin{subequations}\label{eq:III}
	\begin{align}
		\label{eq:III-a}
		\hat{C}\_{[i][j][k][l]} &= 0\,,\\
		\label{eq:III-b}
		\hat{C}\_{[i][j][k][z]} &= 0\,,\\
		\label{eq:III-c}
		\hat{C}\_{[0][i][1][j]} &= 0\,,\\
		\label{eq:III-d}
		\hat{C}\_{[0][1][i][z]} &= 0\,,
	\end{align}
	\end{subequations}
	together with conditions \eqref{eq:conditions-I} and \eqref{eq:conditions-II}.
\end{proposition}

The system \eqref{eq:III} implies
\begin{subequations}\label{eq:III-consequences}
\begin{align}
	\label{eq:III-consequences-a}
	\hat{C}\_{[i][z][k][z]} &= 0\,,\\
	\label{eq:III-consequences-b}
	\hat{C}\_{[0][1][i][j]} &= 0\,,\\
	\label{eq:III-consequences-c}
	\hat{C}\_{[0][1][0][1]} &= 0\,,\\
	\label{eq:III-consequences-d}
	\hat{C}\_{[0][z][1][z]} &= 0\,,\\
	\label{eq:III-consequences-e}
	\hat{C}\_{[0][i][1][z]} &= 0\,,\\
	\label{eq:III-consequences-f}
	\hat{C}\_{[1][i][0][z]} &= 0\,,
\end{align}
\end{subequations}
due to Weyl tensor symmetries.

Suppose that \eqref{eq:conditions-I} and \eqref{eq:conditions-II} are met.

Under the assumptions demanded by proposition \ref{stat:III}, the system \eqref{eq:III} is equivalent to
\begin{subequations}\label{eq:conditions-III}
\begin{align}
	\label{eq:conditions-III-a}
	\mathcal{F}\_{(i)(j)} &= 0\,,\\
	\label{eq:conditions-III-b}
	\mathcal{F}\_{(0)(1)} &= 0\,,\\
	\label{eq:conditions-III-c}
	\mathcal{F}_{ab}\,\phi\_{;(0)} &= 0\,,\\
	\label{eq:conditions-III-d}
	\mathcal{F}_{ab}\,L\_{(i)(j)} &= 0\,,\\
	\label{eq:conditions-III-e}
	\Phi\_{;(i)(j)} - \Phi\_{;(0)(1)} g\_{(i)(j)} &= 0\,.
\end{align}
\end{subequations}

\begin{remark}
	Condition \eqref{eq:conditions-III-e} can be replaced by a more explicit equivalent:
	\begin{subequations}
	\begin{align}
		\label{eq:conditions-III-phi-01}
		\mathrm{e}^{(2\alpha-\beta)\phi} \Phi\_{;(0)(1)} &= -\alpha\,g^{ab} \phi_{;a} \phi_{;b}\,,\\
		\label{eq:conditions-III-phi-ij}
		\Phi\_{;(i)(j)} &= \tfrac{1}{D-2}\,g\^{(k)(l)} \Phi\_{;(k)(l)} g\_{(i)(j)}\,.
	\end{align}
	\end{subequations}
\end{remark}

\begin{remark}
	Note that for (weakly) geodetic $\bm{\ell}$ and vanishing Maxwell field, the type III conditions \eqref{eq:conditions-III} are already satisfied for type II, thanks to the secondary constraint \eqref{eq:II-secondary-constraint}.
\end{remark}

\subsubsection{Speciality of Maxwell field}

Conditions \eqref{eq:conditions-III-a} and \eqref{eq:conditions-III-b} are equivalent to claiming that $\mathrm{bo}_{\langle\bm{\ell}\rangle}\,\bm{\mathcal{F}} < 0$. To show how these two conditions follow from \eqref{eq:III}, we start with noticing that supposing \eqref{eq:conditions-Ib} is met, the equation \eqref{eq:III-consequences-b} is equivalent to
\begin{equation}\label{eq:III-F-at-least-N-intermediate}
	\mathcal{F}\_{(0)(1)} \mathcal{F}\_{(i)(j)} = 0\,.
\end{equation}
Now, we take a specific linear combination of \eqref{eq:III-consequences-c} and \eqref{eq:III-consequences-d} and simplify it using the $\scriptstyle \s{[z][z]}$ component of the vacuum Einstein equation \eqref{eq:reduced-einstein-eq-c}:
\begin{align}
	0 &= \frac{4}{3} (D-1)\,\mathrm{e}^{(4\alpha-2\beta)\phi} \left( (D-2) \hat{C}\_{[0][1][0][1]} - 2\hat{C}\_{[0][z][1][z]} \right) = {}\nonumber\\
	{} &= \mathcal{F}^{2} - (D-1)(D-4) \mathcal{F}\_{(0)(1)} \mathcal{F}\_{(0)(1)} = {}\nonumber\\
	{} &= g\^{(i)(j)} \mathcal{F}^{2}\_{(i)(j)} - (D-2)(D-3) \mathcal{F}\_{(0)(1)} \mathcal{F}\_{(0)(1)}\,,
\end{align}
where we employed \eqref{eq:conditions-Ib}. Considering \eqref{eq:III-F-at-least-N-intermediate}, we see that \eqref{eq:conditions-III-a} and \eqref{eq:conditions-III-b} follow.

\subsubsection{Geodeticity}

\begin{proposition}
	Conditions \eqref{eq:conditions-III-a}, \eqref{eq:conditions-III-b} and \eqref{eq:conditions-Ib} together with \eqref{eq:reduced-einstein-eq-b} imply that either $\bm{\mathcal{F}}$ vanishes or $\bm{\ell}$ is weakly geodetic:
	\begin{equation}\label{eq:III-geodeticity}
		\mathcal{F}_{ab}\,\kappa\^{(i)} = 0\,.
	\end{equation}
\end{proposition}
\begin{proof}
	We have already learned in \eqref{eq:F-0ij-leibniz} that the assumptions \eqref{eq:conditions-III-a}, \eqref{eq:conditions-III-b} and \eqref{eq:conditions-Ib} imply
	\begin{equation}\label{eq:III-F0ij}
		\mathcal{F}\_{(0)(i);(j)} = 0\,.
	\end{equation}
	Thus
	\begin{equation}
		\mathcal{F}\_{(i)(j);(0)} = \mathcal{F}\_{(0)(j);(i)} - \mathcal{F}\_{(0)(i);(j)} = 0\,,
	\end{equation}
	which after using Leibniz rule gives the condition of linear dependence of $\mathcal{F}\_{(1)(i)}$ and $\kappa\_{(i)}$:
	\begin{equation}\label{eq:F-kappa-dependence}
		\kappa\_{(i)} \mathcal{F}\_{(1)(j)} - \kappa\_{(j)} \mathcal{F}\_{(1)(i)} = 0\,.
	\end{equation}

	On the other hand, contracting \eqref{eq:reduced-einstein-eq-b} with $\bm{\ell}$ and substituting \eqref{eq:III-F0ij} and \eqref{eq:conditions-III-b}, we get
	\begin{equation}
		\mathcal{F}\_{(0)(1);(0)} = 0
	\end{equation}
	which, as we have already learned in \eqref{eq:F-010-leibniz}, is the orthogonality condition of $\mathcal{F}\_{(1)(i)}$ and $\kappa\^{(i)}$:
	\begin{equation}\label{eq:F-kappa-orthogonality}
		\kappa\^{(i)} \mathcal{F}\_{(1)(i)} = 0\,.
	\end{equation}

	Combining \eqref{eq:F-kappa-dependence} and \eqref{eq:F-kappa-orthogonality}, we get the desired result that either $\bm{\mathcal{F}}$ or $\kappa\^{(i)}$ vanishes.
\end{proof}

\subsubsection{Speciality of scalar field gradient}

In order to show how the condition \eqref{eq:conditions-III-c} follows from \eqref{eq:III}, we expand \eqref{eq:III-consequences-e}:
\begin{equation}
	0 = \hat{C}\_{[0][i][1][z]} = -\frac{\mathrm{e}^{(-3\alpha+\beta)\phi}}{2} \left(\beta \phi\_{;(0)} \mathcal{F}\_{(1)(i)} + \mathcal{F}\_{(0)(i);(1)} \right).
\end{equation}
Now, due to \eqref{eq:conditions-Ib}, we have
\begin{align}
	\mathcal{F}\_{(0)(i);(1)} ={}& \left(\mathcal{F}_{ab} \ell^{a} e\_{(i)}^{b}\right)\_{\!;(1)} - \mathcal{F}_{ab} \left(\ell^{a} e\_{(i)}^{b}\right)\_{\!;(1)} = {}\nonumber\\
	{}={}& \mathord{-}\mathcal{F}\_{(a)(i)} \ell\^{(a)}\_{\hphantom{(a)};(1)} = -\mathcal{F}\_{(1)(i)} \ell\_{(0);(1)} = 0\,,
\end{align}
where in the last step, we used the identity
\begin{equation}
	\ell\_{(0);(a)} = \ell\_{(b);(a)} \ell\^{(b)} = \tfrac{1}{2} \left(\ell\_{(b)} \ell\^{(b)}\right)\_{\!;(a)} = 0\,.
\end{equation}
We can conclude that \eqref{eq:conditions-III-c} is indeed a necessary condition for \eqref{eq:III}.

\subsubsection{Vanishing optical matrix}

The condition \eqref{eq:conditions-III-d} is a consequence of linear dependence of $\mathcal{F}\_{(1)(i)}$ and $\ell\_{(i);(j)}$
\begin{equation}
	\mathcal{F}\_{(1)(i)} \ell\_{(j);(k)} - \mathcal{F}\_{(1)(j)} \ell\_{(i);(k)} = 0
\end{equation}
and of orthogonality of these two quantities
\begin{equation}
	\mathcal{F}\_{(1)(i)} \ell\^{(i)}\_{\hphantom{(i)};(j)} = 0\,.
\end{equation}
The linear dependence follows from the equation \eqref{eq:III-b}, considering \eqref{eq:conditions-Ib}, \eqref{eq:conditions-III-a} and \eqref{eq:conditions-III-c}:
\begin{equation}
	0 = 2 \mathrm{e}^{(3\alpha-\beta)\phi} \hat{C}\_{[i][j][k][z]} = -\mathcal{F}\_{(i)(j);(k)} = \mathcal{F}\_{(1)(i)} \ell\_{(j);(k)} - \mathcal{F}\_{(1)(j)} \ell\_{(i);(k)}\,.
\end{equation}
On the other hand, the orthogonality follows from the equation \eqref{eq:III-d}:
\begin{equation}
	0 = -2 \mathrm{e}^{(3\alpha-\beta)\phi} \hat{C}\_{[0][1][i][z]} = \mathcal{F}\_{(0)(1);(i)} = \mathcal{F}\_{(1)(j)} \ell\^{(j)}\_{\hphantom{(j)};(i)}\,.
\end{equation}

\subsubsection{B.w. 0 components of scalar field second gradient}

Condition \eqref{eq:conditions-III-phi-01} follows from \eqref{eq:III-consequences-d}, considering \eqref{eq:conditions-III-a} and \eqref{eq:conditions-III-b}.

The trace of \eqref{eq:III-c} is already satisfied thanks to \eqref{eq:III-consequences-c} and \eqref{eq:III-consequences-d}. On the other hand, its traceless part is equivalent to \eqref{eq:conditions-III-phi-ij}.

The condition \eqref{eq:III-a} is already satisfied given \eqref{eq:conditions-III}, which completes the proof of the other direction of the equivalence.

\subsection{Type III with nonvanishing Maxwell field}
\label{chap:III-kundt}

We see that the case of $\mathcal{F}_{ab} \neq 0$ is an important subcase of type III Weyl alignment preserving Kaluza--Klein reduction. In this section, we will try to elaborate on that case in more detail. Suppose
\begin{equation}
	\mathcal{F}_{ab} \neq 0\,.
\end{equation}
In order to satisfy \eqref{eq:conditions-III-d} and \eqref{eq:III-geodeticity}, the $D$-dimensional spacetime must then be Kundt.

From \eqref{eq:conditions-III-c}, we have
\begin{equation}\label{eq:kundt-phi-0}
	\phi\_{;(0)} = 0
\end{equation}
and applying the Leibniz rule:
\begin{equation}
	\phi\_{;(0)(1)} = -\phi\_{;(i)} L\^{(i)}\_{\hphantom{(i)}(1)}\,.
\end{equation}
The condition \eqref{eq:conditions-III-phi-01} can thus be rewritten as
\begin{equation}\label{eq:kundt-phi-ij}
	g\^{(i)(j)} \phi\_{;(i)} \left( \phi\_{;(j)} - \tfrac{1}{\alpha} L\_{(j)(1)} \right) = 0\,.
\end{equation}
We see that for $\phi\_{;(i)} \neq 0$, the $L\_{(i)(1)}$ would have to be also nonvanishing, and the Kundt spacetime $(\mathcal{M}, \bm{g})$ cannot be recurrent\footnote{By \emph{recurrent Kundt}, we mean such Kundt spacetime that admits a recurrent ($k_{a;b} = k_{a} p_{b}$) null vector field.} \citep{clas-hd-review}.

Knowing that the $D$-dimensional spacetime is Kundt, we can choose coordinates $u$, $r$, $x^{2}$, \dots, $x^{D-1}$ in such a way that $r$ would be the affine parameter along the null geodetic $\bm{\ell}$ and that the metric would take the following form \citep{podolsky-kundt}:
\begin{subequations}\label{eq:kundt-metric}
\begin{align}
	\bm{g} &= \bm{\mathrm{d}}u\!\vee\!\bm{\mathrm{d}}r + 2H\,\bm{\mathrm{d}}u^{2} + W_{i}\,\bm{\mathrm{d}}u\!\vee\!\bm{\mathrm{d}}x^{i} + g_{ij}\,\bm{\mathrm{d}}x^{i}\,\bm{\mathrm{d}}x^{j},\\
	{}^{\sharp\sharp}\bm{g} &= \bm{\partial}_{u}\!\vee\!\bm{\partial}_{r} + (-2H + W^{i} W_{i})\,\bm{\partial}_{r}^{2} - W^{i} \bm{\partial}_{r}\!\vee\!\bm{\partial}_{i} + g^{ij} \bm{\partial}_{i}\!\vee\!\bm{\partial}_{j}\,,
\end{align}
\end{subequations}
where
\begin{subequations}
\begin{gather}
	g_{ik} g^{kj} = \delta_{i}^{j}\,,\\
	W^{i} = g^{ij} W_{j}\,.
\end{gather}
\end{subequations}
Then
\begin{equation}\label{eq:kundt-gij-r}
	g_{ij,r} = 0
\end{equation}
and, according to \eqref{eq:kundt-phi-0}:
\begin{equation}\label{eq:kundt-phi-r}
	\phi_{,r} = 0\,.
\end{equation}
In these coordinates, \eqref{eq:kundt-phi-ij} thus becomes:
\begin{equation}\label{eq:kundt-coords-phi-ij}
	g^{ij} \phi_{,i} \left(\phi_{,j} - \tfrac{1}{2\alpha} W_{j,r}\right) = 0\,,
\end{equation}
which can be obtained for example in the following null frame:
\begin{subequations}\label{eq:kundt-frame}
\begin{align}
	\bm{\varepsilon}\^{(0)} &= \bm{\mathrm{d}}r + H\bm{\mathrm{d}}u + W_{i} \bm{\mathrm{d}}x^{i},& \bm{\ell} &= \bm{\partial}_{r}\,,\\
	\bm{\varepsilon}\^{(1)} &= \bm{\mathrm{d}}u\,,& \bm{n} &= \bm{\partial}_{u} - H \bm{\partial}_{r}\,,\\
	\bm{\varepsilon}\^{(i)} &= \varepsilon\^{(i)}_{j} \bm{\mathrm{d}}x^{j}, & \bm{e}\_{(i)} &= e\_{(i)}^{j}\,(\bm{\partial}_{j} - W_{j} \bm{\partial}_{r})\,,
	\shortintertext{where\footnotemark}
	\bigl(\varepsilon\^{(i)}_{j}\bigr)_{\!,r} &= 0\,, & \bigl(e\^{j}_{(i)}\bigr)_{\!,r} &= 0\,,
\end{align}
\end{subequations}
\footnotetext{This additional condition isn't actually helpful for our calculations, but it is a natural choice, given \eqref{eq:kundt-gij-r}.}%
in which
\begin{equation}
	L\_{(1)(i)} = L\_{(i)(1)} = \tfrac{1}{2} W_{j,r} e^{j}\_{(i)}\,.
\end{equation}
Note that the original condition \eqref{eq:kundt-phi-ij} is frame-invariant, as will be also any further computations.

In fact, thanks to corollaries \ref{stat:reduction-rho} and \ref{stat:reduction-kappa} (and especially the remark below them), the $(D+1)$-dimensional spacetime must be Kundt as well (with respect to $\bm{\hat{\ell}}$, and hence it must be a VSI spacetime \citep{vsi-hd}). Therefore \citep[\bgroup\it cf.\egroup][]{podolsky-kundt-clas},
\begin{subequations}\label{eq:kundt-big-transverse-flat}
\begin{gather}
	0 = \hat{C}_{ijkl} = \hat{R}_{ijkl} = {}^{\mathrm{S}}\!\hat{R}_{ijkl}\,,\\
	0 = \hat{C}_{ijkz} = \hat{R}_{ijkz} = {}^{\mathrm{S}}\!\hat{R}_{ijkz}\,,\\
	0 = \hat{C}_{izkz} = \hat{R}_{izkz} = {}^{\mathrm{S}}\!\hat{R}_{izkz}\,,
\end{gather}
\end{subequations}
where the superscript ${}^{\mathrm{S}}$ marks tensor quantities corresponding to the metric $\hat{g}_{AB}$ (or $g_{ab}$, depending on the presence of the hat symbol) pulled back to the surface of constant $u$ and $r$.
\begin{remark}
	The component calculations of \cite{podolsky-kundt-clas} assume coordinates different than those lifted from \eqref{eq:kundt-metric}. This can be fixed by choosing the gauge of
	\begin{equation}
		\mathcal{A}_{r} = 0
	\end{equation}
	and rescaling the coordinate $r$ to be the affine parameter in the lifted spacetime:
	\begin{equation}
		r \mapsto \mathrm{e}^{2\alpha\phi} r\,.
	\end{equation}
	In the end, the result \eqref{eq:kundt-big-transverse-flat} is invariant under transformations of $r$ and the conclusion \eqref{eq:kundt-transverse-weyl} below is independent of any gauge imposed on $\bm{\mathcal{A}}$.
\end{remark}
Suppose that
\begin{equation}
	D \ge 5\,.
\end{equation}
Now, $\bm{{}^{\mathrm{S}}\!g}$ is a Kaluza--Klein reduction of the flat $\bm{{}^{\mathrm{S}}\!\hat{g}}$. Expressing ${}^{\mathrm{S}}\!C_{ijkl}$ in terms of ${}^{\mathrm{S}}\!\hat{R}_{ijkl}$, ${}^{\mathrm{S}}\!\hat{R}_{ijkz}$ and ${}^{\mathrm{S}}\!\hat{R}_{izkz}$ in a similar way to when we expressed $\hat{R}\_{[A][B][C][D]}$ in terms of $C\_{(a)(b)(c)(d)}$, we would end up with
\begin{equation}\label{eq:kundt-transverse-weyl}
	{}^{\mathrm{S}}\!C_{ijkl} = 0\,,
\end{equation}
because
\begin{equation}
	{}^{\mathrm{S}}\!\nabla_{ij} \phi = \nabla_{ij} \phi
\end{equation}
and the equation coming from \eqref{eq:weyl-KK-a} is invariant under the transformation
\begin{align}
	D &\mapsto D-2\,,\\
	g^{ab} \Phi_{;ab} = D\Phi\_{;(0)(1)} &\mapsto {}^{\mathrm{S}}\!g^{ij} {}^{\mathrm{S}}\!\nabla_{ij} \Phi = g\^{(i)(j)} \Phi\_{;(i)(j)} = (D-2) \Phi\_{;(0)(1)}\,.
\end{align}
For $D \ge 6$, this means, according to the Weyl--Schouten theorem, that the metric ${}^{\mathrm{S}}\!g_{ij}$ is conformally flat (and for $D=4$, the ${}^{\mathrm{S}}\!g_{ij}$ is conformally flat automatically). Therefore, we can choose such coordinates $x^{i}$ that \citep[\bgroup\it see also\egroup][discussing form-invariance of \eqref{eq:kundt-metric} under coordinate transformations]{podolsky-kundt}:
\begin{equation}\label{eq:conformally-flat-metric}
	g_{ij} = a^{2}\,\eta_{ij}\,.
\end{equation}
Here, $\eta_{ij}$ is the Kronecker delta.

From now on in this section, suppose that
\begin{equation}
	D = 4 \quad\vee\quad D \ge 6\,.
\end{equation}
Equation \eqref{eq:conditions-III-e} together with the $\scriptstyle \s{[z][z]}$ component of the vacuum Einstein equation \eqref{eq:reduced-einstein-eq-c} are equivalent to
\begin{subequations}
\begin{align}
	\label{eq:III-phi-01-frame}
	\phi\_{;(0)(1)} &= -\alpha g\^{(k)(l)} \phi\_{;(k)} \phi\_{;(l)}\,,\\
	\label{eq:III-phi-ij-frame}
	\phi\_{;(i)(j)} + (\beta-2\alpha) \phi\_{;(i)} \phi\_{;(j)} &= -\alpha g\^{(k)(l)} \phi\_{;(k)} \phi\_{;(l)}\,g\_{(i)(j)}\,,
\end{align}
\end{subequations}
which is equivalent (remember \eqref{eq:kundt-coords-phi-ij}) to
\begin{subequations}
\begin{align}
	\label{eq:III-phi-01}
	&\eta^{ij} \phi_{,i} \left(\phi_{,j} - \tfrac{1}{2\alpha} W_{j,r}\right) = 0\,,\\
	\label{eq:III-phi-ij}
	&\phi_{,ij} - \phi_{,i}\,(\ln a)_{,j} - \phi_{,j}\,(\ln a)_{,i} + \eta^{kl} \phi_{,k}\,(\ln a)_{,l}\,\eta_{ij} = {}\nonumber\\
	&\qquad {} = (2\alpha - \beta)\,\phi_{,i} \phi_{,j} - \alpha\,\eta^{kl} \phi_{,k} \phi_{,l}\,\eta_{ij}\,.
\end{align}
\end{subequations}
A trivial solution of this system would be $\phi$ that only depends on the coordinate $u$. Assuming the opposite,
\begin{equation}
	\phi_{,i} \neq 0\,,
\end{equation}
the equation \eqref{eq:III-phi-ij} can be explicitly solved for both the conformal factor $a$ and the scalar field $\phi$.

\subsubsection{Conformal factor}

Contracting \eqref{eq:III-phi-ij} with $\eta^{jm} \phi_{,m}$, we get an integrable equation for $a$:
\begin{equation}
	\left((\alpha - \beta)\,\phi + \ln\frac{a}{\sqrt{\eta^{kl} \phi_{,k} \phi_{,l}}}\right)_{\!,i} = 0\,.
\end{equation}
Integrating and substituting back for $a$, we get the following equivalent of \eqref{eq:III-phi-ij}:
\begin{subequations}
\begin{align}
	\label{eq:III-a-equation}
	&a = b\,\sqrt{\phi_{,i} \phi_{,i}}\,\mathrm{e}^{(\beta-\alpha)\phi}\,,\\
	\label{eq:III-phi-equation}
	&\phi_{,k} \phi_{,k} \phi_{,ij} + \phi_{,k} \phi_{,l} \phi_{,kl} \eta_{ij} + \beta \left( \phi_{,k} \phi_{,k} \right)^{2} \eta_{ij} = {}\nonumber\\
	&\quad {} = \beta \phi_{,k} \phi_{,k} \phi_{,i} \phi_{,j} + \phi_{,i} \phi_{,jk} \phi_{,k} + \phi_{,j} \phi_{,ik} \phi_{,k}\,.
\end{align}
\end{subequations}
where $b$ is an arbitrary function of the coordinate $u$ and where for brevity, the Einstein summation convention is used.

\subsubsection{Scalar field}

The equation \eqref{eq:III-phi-equation} is quartic in $\phi$, but under substitution
\begin{equation}\label{eq:psi-definition}
	\psi := \mathrm{e}^{\beta \phi}\,,
\end{equation}
the quartic terms cancel out, resulting in a homogeneous cubic:
\begin{equation}\label{eq:III-psi-equation}
	\psi_{,k} \psi_{,k} \psi_{,ij} + \psi_{,k} \psi_{,l} \psi_{,kl} \eta_{ij} = \psi_{,i} \psi_{,jk} \psi_{,k} + \psi_{,j} \psi_{,ik} \psi_{,k}\,.
\end{equation}

Comparing the components of $\psi_{,ij}$ in a frame aligned with $\psi_{,i}$, it is manifest that each hypersurface of constant $\psi$ has at each point the second fundamental form proportional to the first fundamental form, making the contours totally umbilical, and therefore also of constant mean curvature (in our case where the ambient manifold is flat and Riemannian) \cite{totally-umbilical}. For $D > 4$, this means that each hypersurface of constant $\psi$ is locally either part of a plane, or part of a hypersphere (considering the coordinates $u$ and $r$ fixed). Keeping this in mind, it is now a simple exercise (involving only systems of ordinary differential equations) to solve \eqref{eq:III-psi-equation} for the six qualitatively different arrangements of contours as listed below, showing that in
\begin{equation}
	D \ge 6\,,
\end{equation}
for any point where $\psi \neq 0$, there is\footnote{We assume that $\psi \in \mathrm{C}^{2}\big(\mathbb{R}^{D-2}\big)$ for any fixed $u$ and $r$.} a neighborhood that is a union of sets in which $\psi$ takes at least one of the following forms: either the contours are arranged as concentric spheres:
\begin{subequations}\label{eq:III-psi-solution}
\begin{equation}\label{eq:III-psi-solution-a}
	\psi = C_{0} + C_{1} \ln \left(x - \xsub0\right)^{2},
\end{equation}
or as non-concentric non-overlapping spheres:
\begin{equation}\label{eq:III-psi-solution-b}
	\psi = C_{0} + C_{1} \ln\frac{\left(x - \xsub0\right)^{2}}{\left(x - \xsub1\right)^{2}}\,,
\end{equation}
overlapping spheres:
\begin{equation}\label{eq:III-psi-solution-c}
	\psi = C_{0} + C_{1} \arg \left(\pm(x - \zeta)^{2}\right),
\end{equation}
tangent spheres:
\begin{equation}\label{eq:III-psi-solution-d}
	\psi = C_{0} + \frac{\left(x - \xsub0\right) \cdot \nu}{\left(x - \xsub0\right)^{2}}\,,
\end{equation}
overlapping planes:
\begin{equation}\label{eq:III-psi-solution-e}
	\psi = C_{0} + C_{1} \arg\,\bigl( \pm(x - \xsub0) \cdot \mu \bigr)\,,
\end{equation}
or as non-overlapping planes:
\begin{equation}\label{eq:III-psi-solution-f}
	\psi = C_{0} + x \cdot \nu\,,
\end{equation}
\end{subequations}
where $C_{0}$ and $C_{1}$ are some functions of the coordinate $u$ and where $\xsub0$, $\xsub1$ and $\nu$ are some $(D-2)$-dimensional functions of $u$ and where $\zeta$ and $\mu$ are some $(D-2)$-dimensional complex-valued functions of $u$ satisfying $\mu^{2} = 0$. Conversely, real-valued $\psi$ in any of these forms \eqref{eq:III-psi-solution} is a solution to the equation \eqref{eq:III-psi-equation}.
\begin{remark}
	The two solutions provided for overlapping spheres \eqref{eq:III-psi-solution-c} are equivalent, but with different sets of discontinuity, so they can be glued together as necessary, possibly forming a nontrivial topology with singularity at the $(D-4)$-dimensional ring on $(x - \zeta)^{2} = 0$. Similar situation occurs with the solution for overlapping planes \eqref{eq:III-psi-solution-e}.
\end{remark}

Figure \ref{fig:psi-contours} depicts contour sections of $\psi$ for the different types of $D>4$ solutions of \eqref{eq:III-psi-equation}, in planes respecting the solution symmetries (i.e. either perpendicular to the translation symmetry vector, or tangent with the rotation symmetry vector -- which is aligned with the horizontal axis).
\begin{figure}[htbp]
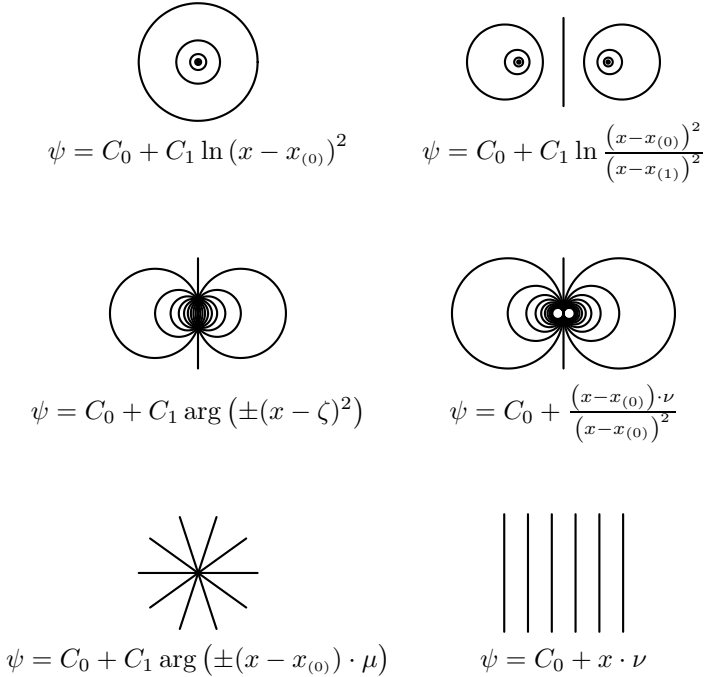

\centering
\begin{tabular}{cc}
\begin{tikzpicture}[gnuplot]
%% generated with GNUPLOT 5.4p1 (Lua 5.4; terminal rev. Jun 2020, script rev. 114)
%% Po 16. ledna 2023, 12:23:11
\path (0.000,0.000) rectangle (1.600,2.500);
\gpcolor{rgb color={0.000,0.000,0.000}}
\gpsetlinetype{gp lt border}
\gpsetdashtype{gp dt solid}
\gpsetlinewidth{2.00}
\draw[gp path] (1.583,1.280)--(1.581,1.329)--(1.577,1.379)--(1.569,1.428)--(1.558,1.476)%
  --(1.544,1.524)--(1.527,1.571)--(1.507,1.616)--(1.484,1.660)--(1.459,1.703)--(1.430,1.744)%
  --(1.400,1.783)--(1.367,1.820)--(1.331,1.855)--(1.294,1.888)--(1.254,1.918)--(1.213,1.945)%
  --(1.170,1.970)--(1.125,1.992)--(1.079,2.011)--(1.032,2.028)--(0.984,2.041)--(0.936,2.051)%
  --(0.886,2.058)--(0.837,2.062)--(0.787,2.063)--(0.737,2.061)--(0.688,2.055)--(0.639,2.046)%
  --(0.591,2.035)--(0.543,2.020)--(0.497,2.002)--(0.452,1.982)--(0.408,1.958)--(0.366,1.932)%
  --(0.325,1.903)--(0.286,1.872)--(0.250,1.838)--(0.216,1.802)--(0.184,1.764)--(0.154,1.724)%
  --(0.127,1.682)--(0.103,1.639)--(0.082,1.594)--(0.063,1.547)--(0.048,1.500)--(0.035,1.452)%
  --(0.026,1.403)--(0.020,1.354)--(0.016,1.304)--(0.016,1.255)--(0.020,1.205)--(0.026,1.156)%
  --(0.035,1.107)--(0.048,1.059)--(0.063,1.012)--(0.082,0.965)--(0.103,0.920)--(0.127,0.877)%
  --(0.154,0.835)--(0.184,0.795)--(0.216,0.757)--(0.250,0.721)--(0.286,0.687)--(0.325,0.656)%
  --(0.366,0.627)--(0.408,0.601)--(0.452,0.577)--(0.497,0.557)--(0.543,0.539)--(0.591,0.524)%
  --(0.639,0.513)--(0.688,0.504)--(0.737,0.498)--(0.787,0.496)--(0.837,0.497)--(0.886,0.501)%
  --(0.936,0.508)--(0.984,0.518)--(1.032,0.531)--(1.079,0.548)--(1.125,0.567)--(1.170,0.589)%
  --(1.213,0.614)--(1.254,0.641)--(1.294,0.671)--(1.331,0.704)--(1.367,0.739)--(1.400,0.776)%
  --(1.430,0.815)--(1.459,0.856)--(1.484,0.899)--(1.507,0.943)--(1.527,0.988)--(1.544,1.035)%
  --(1.558,1.083)--(1.569,1.131)--(1.577,1.180)--(1.581,1.230)--(1.583,1.279);
\draw[gp path] (1.088,1.280)--(1.087,1.298)--(1.085,1.316)--(1.083,1.334)--(1.078,1.352)%
  --(1.073,1.369)--(1.067,1.387)--(1.060,1.403)--(1.051,1.420)--(1.042,1.435)--(1.032,1.450)%
  --(1.020,1.465)--(1.008,1.478)--(0.995,1.491)--(0.981,1.503)--(0.967,1.514)--(0.951,1.524)%
  --(0.936,1.534)--(0.919,1.542)--(0.902,1.549)--(0.885,1.555)--(0.867,1.560)--(0.850,1.563)%
  --(0.831,1.566)--(0.813,1.567)--(0.795,1.568)--(0.777,1.567)--(0.758,1.565)--(0.740,1.562)%
  --(0.723,1.557)--(0.705,1.552)--(0.688,1.545)--(0.672,1.538)--(0.655,1.529)--(0.640,1.519)%
  --(0.625,1.509)--(0.611,1.497)--(0.597,1.485)--(0.585,1.472)--(0.573,1.458)--(0.562,1.443)%
  --(0.552,1.428)--(0.543,1.412)--(0.535,1.395)--(0.529,1.378)--(0.523,1.361)--(0.518,1.343)%
  --(0.515,1.325)--(0.513,1.307)--(0.511,1.289)--(0.511,1.270)--(0.513,1.252)--(0.515,1.234)%
  --(0.518,1.216)--(0.523,1.198)--(0.529,1.181)--(0.535,1.164)--(0.543,1.147)--(0.552,1.131)%
  --(0.562,1.116)--(0.573,1.101)--(0.585,1.087)--(0.597,1.074)--(0.611,1.062)--(0.625,1.050)%
  --(0.640,1.040)--(0.655,1.030)--(0.672,1.021)--(0.688,1.014)--(0.705,1.007)--(0.723,1.002)%
  --(0.740,0.997)--(0.758,0.994)--(0.777,0.992)--(0.795,0.991)--(0.813,0.992)--(0.831,0.993)%
  --(0.850,0.996)--(0.867,0.999)--(0.885,1.004)--(0.902,1.010)--(0.919,1.017)--(0.936,1.025)%
  --(0.951,1.035)--(0.967,1.045)--(0.981,1.056)--(0.995,1.068)--(1.008,1.081)--(1.020,1.094)%
  --(1.032,1.109)--(1.042,1.124)--(1.051,1.139)--(1.060,1.156)--(1.067,1.172)--(1.073,1.190)%
  --(1.078,1.207)--(1.083,1.225)--(1.085,1.243)--(1.087,1.261)--cycle;
\draw[gp path] (0.906,1.280)--(0.905,1.286)--(0.905,1.293)--(0.904,1.300)--(0.902,1.306)%
  --(0.900,1.313)--(0.898,1.319)--(0.895,1.325)--(0.892,1.331)--(0.889,1.337)--(0.885,1.342)%
  --(0.881,1.348)--(0.876,1.353)--(0.871,1.357)--(0.866,1.362)--(0.861,1.366)--(0.855,1.370)%
  --(0.850,1.373)--(0.844,1.376)--(0.837,1.379)--(0.831,1.381)--(0.824,1.383)--(0.818,1.384)%
  --(0.811,1.385)--(0.805,1.385)--(0.798,1.386)--(0.791,1.385)--(0.784,1.384)--(0.778,1.383)%
  --(0.771,1.382)--(0.765,1.380)--(0.759,1.377)--(0.752,1.375)--(0.746,1.371)--(0.741,1.368)%
  --(0.735,1.364)--(0.730,1.360)--(0.725,1.355)--(0.720,1.350)--(0.716,1.345)--(0.712,1.340)%
  --(0.709,1.334)--(0.705,1.328)--(0.702,1.322)--(0.700,1.316)--(0.698,1.309)--(0.696,1.303)%
  --(0.695,1.296)--(0.694,1.290)--(0.694,1.283)--(0.694,1.276)--(0.694,1.269)--(0.695,1.263)%
  --(0.696,1.256)--(0.698,1.250)--(0.700,1.243)--(0.702,1.237)--(0.705,1.231)--(0.709,1.225)%
  --(0.712,1.219)--(0.716,1.214)--(0.720,1.209)--(0.725,1.204)--(0.730,1.199)--(0.735,1.195)%
  --(0.741,1.191)--(0.746,1.188)--(0.752,1.184)--(0.759,1.182)--(0.765,1.179)--(0.771,1.177)%
  --(0.778,1.176)--(0.784,1.175)--(0.791,1.174)--(0.798,1.173)--(0.805,1.174)--(0.811,1.174)%
  --(0.818,1.175)--(0.824,1.176)--(0.831,1.178)--(0.837,1.180)--(0.844,1.183)--(0.850,1.186)%
  --(0.855,1.189)--(0.861,1.193)--(0.866,1.197)--(0.871,1.202)--(0.876,1.206)--(0.881,1.211)%
  --(0.885,1.217)--(0.889,1.222)--(0.892,1.228)--(0.895,1.234)--(0.898,1.240)--(0.900,1.246)%
  --(0.902,1.253)--(0.904,1.259)--(0.905,1.266)--(0.905,1.273)--cycle;
\draw[gp path] (0.839,1.280)--(0.838,1.282)--(0.838,1.284)--(0.838,1.287)--(0.837,1.289)%
  --(0.837,1.292)--(0.836,1.294)--(0.835,1.296)--(0.834,1.298)--(0.832,1.301)--(0.831,1.303)%
  --(0.829,1.305)--(0.828,1.306)--(0.826,1.308)--(0.824,1.310)--(0.822,1.311)--(0.820,1.313)%
  --(0.818,1.314)--(0.816,1.315)--(0.813,1.316)--(0.811,1.317)--(0.809,1.317)--(0.806,1.318)%
  --(0.804,1.318)--(0.801,1.318)--(0.799,1.319)--(0.796,1.318)--(0.794,1.318)--(0.792,1.318)%
  --(0.789,1.317)--(0.787,1.316)--(0.784,1.315)--(0.782,1.314)--(0.780,1.313)--(0.778,1.312)%
  --(0.776,1.311)--(0.774,1.309)--(0.772,1.307)--(0.770,1.306)--(0.769,1.304)--(0.767,1.302)%
  --(0.766,1.300)--(0.765,1.297)--(0.764,1.295)--(0.763,1.293)--(0.762,1.290)--(0.761,1.288)%
  --(0.761,1.286)--(0.761,1.283)--(0.761,1.281)--(0.761,1.278)--(0.761,1.276)--(0.761,1.273)%
  --(0.761,1.271)--(0.762,1.269)--(0.763,1.266)--(0.764,1.264)--(0.765,1.262)--(0.766,1.259)%
  --(0.767,1.257)--(0.769,1.255)--(0.770,1.253)--(0.772,1.252)--(0.774,1.250)--(0.776,1.248)%
  --(0.778,1.247)--(0.780,1.246)--(0.782,1.245)--(0.784,1.244)--(0.787,1.243)--(0.789,1.242)%
  --(0.792,1.241)--(0.794,1.241)--(0.796,1.241)--(0.799,1.240)--(0.801,1.241)--(0.804,1.241)%
  --(0.806,1.241)--(0.809,1.242)--(0.811,1.242)--(0.813,1.243)--(0.816,1.244)--(0.818,1.245)%
  --(0.820,1.246)--(0.822,1.248)--(0.824,1.249)--(0.826,1.251)--(0.828,1.253)--(0.829,1.254)%
  --(0.831,1.256)--(0.832,1.258)--(0.834,1.261)--(0.835,1.263)--(0.836,1.265)--(0.837,1.267)%
  --(0.837,1.270)--(0.838,1.272)--(0.838,1.275)--(0.838,1.277)--cycle;
\draw[gp path] (0.814,1.280)--(0.814,1.281)--(0.814,1.282)--(0.813,1.283)--(0.813,1.284)%
  --(0.813,1.285)--(0.812,1.286)--(0.812,1.287)--(0.811,1.288)--(0.810,1.289)--(0.809,1.290)%
  --(0.809,1.291)--(0.808,1.291)--(0.807,1.292)--(0.806,1.292)--(0.805,1.293)--(0.804,1.293)%
  --(0.803,1.293)--(0.802,1.294)--(0.801,1.294)--(0.800,1.294)--(0.799,1.294)--(0.798,1.294)%
  --(0.797,1.294)--(0.796,1.293)--(0.795,1.293)--(0.794,1.293)--(0.793,1.292)--(0.792,1.292)%
  --(0.792,1.291)--(0.791,1.291)--(0.790,1.290)--(0.789,1.290)--(0.789,1.289)--(0.788,1.288)%
  --(0.787,1.287)--(0.787,1.286)--(0.786,1.285)--(0.786,1.284)--(0.786,1.283)--(0.785,1.282)%
  --(0.785,1.281)--(0.785,1.280)--(0.785,1.279)--(0.785,1.278)--(0.785,1.277)--(0.786,1.276)%
  --(0.786,1.275)--(0.786,1.274)--(0.787,1.273)--(0.787,1.272)--(0.788,1.271)--(0.789,1.270)%
  --(0.789,1.269)--(0.790,1.269)--(0.791,1.268)--(0.792,1.268)--(0.792,1.267)--(0.793,1.267)%
  --(0.794,1.266)--(0.795,1.266)--(0.796,1.266)--(0.797,1.265)--(0.798,1.265)--(0.799,1.265)%
  --(0.800,1.265)--(0.801,1.265)--(0.802,1.265)--(0.803,1.266)--(0.804,1.266)--(0.805,1.266)%
  --(0.806,1.267)--(0.807,1.267)--(0.808,1.268)--(0.809,1.268)--(0.809,1.269)--(0.810,1.270)%
  --(0.811,1.271)--(0.812,1.272)--(0.812,1.273)--(0.813,1.274)--(0.813,1.275)--(0.813,1.276)%
  --(0.814,1.277)--(0.814,1.278)--(0.814,1.279)--cycle;
\draw[gp path] (0.805,1.280)--(0.805,1.281)--(0.804,1.281)--(0.804,1.282)--(0.804,1.283)%
  --(0.803,1.283)--(0.803,1.284)--(0.802,1.284)--(0.801,1.284)--(0.801,1.285)--(0.800,1.285)%
  --(0.799,1.285)--(0.798,1.285)--(0.798,1.284)--(0.797,1.284)--(0.796,1.284)--(0.796,1.283)%
  --(0.795,1.283)--(0.795,1.282)--(0.795,1.281)--(0.794,1.281)--(0.794,1.280)--(0.794,1.279)%
  --(0.794,1.278)--(0.795,1.278)--(0.795,1.277)--(0.795,1.276)--(0.796,1.276)--(0.796,1.275)%
  --(0.797,1.275)--(0.798,1.275)--(0.798,1.274)--(0.799,1.274)--(0.800,1.274)--(0.801,1.274)%
  --(0.801,1.275)--(0.802,1.275)--(0.803,1.275)--(0.803,1.276)--(0.804,1.276)--(0.804,1.277)%
  --(0.804,1.278)--(0.805,1.278)--(0.805,1.279)--cycle;
%% coordinates of the plot area
\gpdefrectangularnode{gp plot 1}{\pgfpoint{0.016cm}{0.496cm}}{\pgfpoint{1.583cm}{2.063cm}}
\end{tikzpicture}
%% gnuplot variables
&\input{plots/b}\\[-17pt]%
$\psi = C_{0} + C_{1} \ln \left(x - \xsub0\right)^{2}$&%
$\psi = C_{0} + C_{1} \ln\frac{\left(x - \xsub0\right)^{2}}{\left(x - \xsub1\right)^{2}}$\\[20pt]%
\input{plots/c}&\input{plots/d}\\[-13pt]%
$\psi = C_{0} + C_{1} \arg\left(\pm(x - \zeta)^{2}\right)$&%
$\psi = C_{0} + \frac{\left(x - \xsub0\right) \cdot \nu}{\left(x - \xsub0\right)^{2}}$\\[20pt]%
\input{plots/e}&\begin{tikzpicture}[gnuplot]
%% generated with GNUPLOT 5.4p1 (Lua 5.4; terminal rev. Jun 2020, script rev. 114)
%% Po 16. ledna 2023, 12:23:24
\path (0.000,0.000) rectangle (1.600,2.500);
\gpcolor{rgb color={0.000,0.000,0.000}}
\gpsetlinetype{gp lt border}
\gpsetdashtype{gp dt solid}
\gpsetlinewidth{2.00}
\draw[gp path] (0.016,0.496)--(0.016,0.512)--(0.016,0.528)--(0.016,0.543)--(0.016,0.559)%
  --(0.016,0.575)--(0.016,0.591)--(0.016,0.607)--(0.016,0.623)--(0.016,0.638)--(0.016,0.654)%
  --(0.016,0.670)--(0.016,0.686)--(0.016,0.702)--(0.016,0.718)--(0.016,0.733)--(0.016,0.749)%
  --(0.016,0.765)--(0.016,0.781)--(0.016,0.797)--(0.016,0.813)--(0.016,0.828)--(0.016,0.844)%
  --(0.016,0.860)--(0.016,0.876)--(0.016,0.892)--(0.016,0.908)--(0.016,0.923)--(0.016,0.939)%
  --(0.016,0.955)--(0.016,0.971)--(0.016,0.987)--(0.016,1.003)--(0.016,1.018)--(0.016,1.034)%
  --(0.016,1.050)--(0.016,1.066)--(0.016,1.082)--(0.016,1.097)--(0.016,1.113)--(0.016,1.129)%
  --(0.016,1.145)--(0.016,1.161)--(0.016,1.177)--(0.016,1.192)--(0.016,1.208)--(0.016,1.224)%
  --(0.016,1.240)--(0.016,1.256)--(0.016,1.272)--(0.016,1.287)--(0.016,1.303)--(0.016,1.319)%
  --(0.016,1.335)--(0.016,1.351)--(0.016,1.367)--(0.016,1.382)--(0.016,1.398)--(0.016,1.414)%
  --(0.016,1.430)--(0.016,1.446)--(0.016,1.462)--(0.016,1.477)--(0.016,1.493)--(0.016,1.509)%
  --(0.016,1.525)--(0.016,1.541)--(0.016,1.556)--(0.016,1.572)--(0.016,1.588)--(0.016,1.604)%
  --(0.016,1.620)--(0.016,1.636)--(0.016,1.651)--(0.016,1.667)--(0.016,1.683)--(0.016,1.699)%
  --(0.016,1.715)--(0.016,1.731)--(0.016,1.746)--(0.016,1.762)--(0.016,1.778)--(0.016,1.794)%
  --(0.016,1.810)--(0.016,1.826)--(0.016,1.841)--(0.016,1.857)--(0.016,1.873)--(0.016,1.889)%
  --(0.016,1.905)--(0.016,1.921)--(0.016,1.936)--(0.016,1.952)--(0.016,1.968)--(0.016,1.984)%
  --(0.016,2.000)--(0.016,2.016)--(0.016,2.031)--(0.016,2.047)--(0.016,2.063);
\draw[gp path] (0.329,0.496)--(0.329,0.512)--(0.329,0.528)--(0.329,0.543)--(0.329,0.559)%
  --(0.329,0.575)--(0.329,0.591)--(0.329,0.607)--(0.329,0.623)--(0.329,0.638)--(0.329,0.654)%
  --(0.329,0.670)--(0.329,0.686)--(0.329,0.702)--(0.329,0.718)--(0.329,0.733)--(0.329,0.749)%
  --(0.329,0.765)--(0.329,0.781)--(0.329,0.797)--(0.329,0.813)--(0.329,0.828)--(0.329,0.844)%
  --(0.329,0.860)--(0.329,0.876)--(0.329,0.892)--(0.329,0.908)--(0.329,0.923)--(0.329,0.939)%
  --(0.329,0.955)--(0.329,0.971)--(0.329,0.987)--(0.329,1.003)--(0.329,1.018)--(0.329,1.034)%
  --(0.329,1.050)--(0.329,1.066)--(0.329,1.082)--(0.329,1.097)--(0.329,1.113)--(0.329,1.129)%
  --(0.329,1.145)--(0.329,1.161)--(0.329,1.177)--(0.329,1.192)--(0.329,1.208)--(0.329,1.224)%
  --(0.329,1.240)--(0.329,1.256)--(0.329,1.272)--(0.329,1.287)--(0.329,1.303)--(0.329,1.319)%
  --(0.329,1.335)--(0.329,1.351)--(0.329,1.367)--(0.329,1.382)--(0.329,1.398)--(0.329,1.414)%
  --(0.329,1.430)--(0.329,1.446)--(0.329,1.462)--(0.329,1.477)--(0.329,1.493)--(0.329,1.509)%
  --(0.329,1.525)--(0.329,1.541)--(0.329,1.556)--(0.329,1.572)--(0.329,1.588)--(0.329,1.604)%
  --(0.329,1.620)--(0.329,1.636)--(0.329,1.651)--(0.329,1.667)--(0.329,1.683)--(0.329,1.699)%
  --(0.329,1.715)--(0.329,1.731)--(0.329,1.746)--(0.329,1.762)--(0.329,1.778)--(0.329,1.794)%
  --(0.329,1.810)--(0.329,1.826)--(0.329,1.841)--(0.329,1.857)--(0.329,1.873)--(0.329,1.889)%
  --(0.329,1.905)--(0.329,1.921)--(0.329,1.936)--(0.329,1.952)--(0.329,1.968)--(0.329,1.984)%
  --(0.329,2.000)--(0.329,2.016)--(0.329,2.031)--(0.329,2.047)--(0.329,2.063);
\draw[gp path] (0.643,0.496)--(0.643,0.512)--(0.643,0.528)--(0.643,0.543)--(0.643,0.559)%
  --(0.643,0.575)--(0.643,0.591)--(0.643,0.607)--(0.643,0.623)--(0.643,0.638)--(0.643,0.654)%
  --(0.643,0.670)--(0.643,0.686)--(0.643,0.702)--(0.643,0.718)--(0.643,0.733)--(0.643,0.749)%
  --(0.643,0.765)--(0.643,0.781)--(0.643,0.797)--(0.643,0.813)--(0.643,0.828)--(0.643,0.844)%
  --(0.643,0.860)--(0.643,0.876)--(0.643,0.892)--(0.643,0.908)--(0.643,0.923)--(0.643,0.939)%
  --(0.643,0.955)--(0.643,0.971)--(0.643,0.987)--(0.643,1.003)--(0.643,1.018)--(0.643,1.034)%
  --(0.643,1.050)--(0.643,1.066)--(0.643,1.082)--(0.643,1.097)--(0.643,1.113)--(0.643,1.129)%
  --(0.643,1.145)--(0.643,1.161)--(0.643,1.177)--(0.643,1.192)--(0.643,1.208)--(0.643,1.224)%
  --(0.643,1.240)--(0.643,1.256)--(0.643,1.272)--(0.643,1.287)--(0.643,1.303)--(0.643,1.319)%
  --(0.643,1.335)--(0.643,1.351)--(0.643,1.367)--(0.643,1.382)--(0.643,1.398)--(0.643,1.414)%
  --(0.643,1.430)--(0.643,1.446)--(0.643,1.462)--(0.643,1.477)--(0.643,1.493)--(0.643,1.509)%
  --(0.643,1.525)--(0.643,1.541)--(0.643,1.556)--(0.643,1.572)--(0.643,1.588)--(0.643,1.604)%
  --(0.643,1.620)--(0.643,1.636)--(0.643,1.651)--(0.643,1.667)--(0.643,1.683)--(0.643,1.699)%
  --(0.643,1.715)--(0.643,1.731)--(0.643,1.746)--(0.643,1.762)--(0.643,1.778)--(0.643,1.794)%
  --(0.643,1.810)--(0.643,1.826)--(0.643,1.841)--(0.643,1.857)--(0.643,1.873)--(0.643,1.889)%
  --(0.643,1.905)--(0.643,1.921)--(0.643,1.936)--(0.643,1.952)--(0.643,1.968)--(0.643,1.984)%
  --(0.643,2.000)--(0.643,2.016)--(0.643,2.031)--(0.643,2.047)--(0.643,2.063);
\draw[gp path] (0.956,0.496)--(0.956,0.512)--(0.956,0.528)--(0.956,0.543)--(0.956,0.559)%
  --(0.956,0.575)--(0.956,0.591)--(0.956,0.607)--(0.956,0.623)--(0.956,0.638)--(0.956,0.654)%
  --(0.956,0.670)--(0.956,0.686)--(0.956,0.702)--(0.956,0.718)--(0.956,0.733)--(0.956,0.749)%
  --(0.956,0.765)--(0.956,0.781)--(0.956,0.797)--(0.956,0.813)--(0.956,0.828)--(0.956,0.844)%
  --(0.956,0.860)--(0.956,0.876)--(0.956,0.892)--(0.956,0.908)--(0.956,0.923)--(0.956,0.939)%
  --(0.956,0.955)--(0.956,0.971)--(0.956,0.987)--(0.956,1.003)--(0.956,1.018)--(0.956,1.034)%
  --(0.956,1.050)--(0.956,1.066)--(0.956,1.082)--(0.956,1.097)--(0.956,1.113)--(0.956,1.129)%
  --(0.956,1.145)--(0.956,1.161)--(0.956,1.177)--(0.956,1.192)--(0.956,1.208)--(0.956,1.224)%
  --(0.956,1.240)--(0.956,1.256)--(0.956,1.272)--(0.956,1.287)--(0.956,1.303)--(0.956,1.319)%
  --(0.956,1.335)--(0.956,1.351)--(0.956,1.367)--(0.956,1.382)--(0.956,1.398)--(0.956,1.414)%
  --(0.956,1.430)--(0.956,1.446)--(0.956,1.462)--(0.956,1.477)--(0.956,1.493)--(0.956,1.509)%
  --(0.956,1.525)--(0.956,1.541)--(0.956,1.556)--(0.956,1.572)--(0.956,1.588)--(0.956,1.604)%
  --(0.956,1.620)--(0.956,1.636)--(0.956,1.651)--(0.956,1.667)--(0.956,1.683)--(0.956,1.699)%
  --(0.956,1.715)--(0.956,1.731)--(0.956,1.746)--(0.956,1.762)--(0.956,1.778)--(0.956,1.794)%
  --(0.956,1.810)--(0.956,1.826)--(0.956,1.841)--(0.956,1.857)--(0.956,1.873)--(0.956,1.889)%
  --(0.956,1.905)--(0.956,1.921)--(0.956,1.936)--(0.956,1.952)--(0.956,1.968)--(0.956,1.984)%
  --(0.956,2.000)--(0.956,2.016)--(0.956,2.031)--(0.956,2.047)--(0.956,2.063);
\draw[gp path] (1.270,0.496)--(1.270,0.512)--(1.270,0.528)--(1.270,0.543)--(1.270,0.559)%
  --(1.270,0.575)--(1.270,0.591)--(1.270,0.607)--(1.270,0.623)--(1.270,0.638)--(1.270,0.654)%
  --(1.270,0.670)--(1.270,0.686)--(1.270,0.702)--(1.270,0.718)--(1.270,0.733)--(1.270,0.749)%
  --(1.270,0.765)--(1.270,0.781)--(1.270,0.797)--(1.270,0.813)--(1.270,0.828)--(1.270,0.844)%
  --(1.270,0.860)--(1.270,0.876)--(1.270,0.892)--(1.270,0.908)--(1.270,0.923)--(1.270,0.939)%
  --(1.270,0.955)--(1.270,0.971)--(1.270,0.987)--(1.270,1.003)--(1.270,1.018)--(1.270,1.034)%
  --(1.270,1.050)--(1.270,1.066)--(1.270,1.082)--(1.270,1.097)--(1.270,1.113)--(1.270,1.129)%
  --(1.270,1.145)--(1.270,1.161)--(1.270,1.177)--(1.270,1.192)--(1.270,1.208)--(1.270,1.224)%
  --(1.270,1.240)--(1.270,1.256)--(1.270,1.272)--(1.270,1.287)--(1.270,1.303)--(1.270,1.319)%
  --(1.270,1.335)--(1.270,1.351)--(1.270,1.367)--(1.270,1.382)--(1.270,1.398)--(1.270,1.414)%
  --(1.270,1.430)--(1.270,1.446)--(1.270,1.462)--(1.270,1.477)--(1.270,1.493)--(1.270,1.509)%
  --(1.270,1.525)--(1.270,1.541)--(1.270,1.556)--(1.270,1.572)--(1.270,1.588)--(1.270,1.604)%
  --(1.270,1.620)--(1.270,1.636)--(1.270,1.651)--(1.270,1.667)--(1.270,1.683)--(1.270,1.699)%
  --(1.270,1.715)--(1.270,1.731)--(1.270,1.746)--(1.270,1.762)--(1.270,1.778)--(1.270,1.794)%
  --(1.270,1.810)--(1.270,1.826)--(1.270,1.841)--(1.270,1.857)--(1.270,1.873)--(1.270,1.889)%
  --(1.270,1.905)--(1.270,1.921)--(1.270,1.936)--(1.270,1.952)--(1.270,1.968)--(1.270,1.984)%
  --(1.270,2.000)--(1.270,2.016)--(1.270,2.031)--(1.270,2.047)--(1.270,2.063);
\draw[gp path] (1.583,0.496)--(1.583,0.512)--(1.583,0.528)--(1.583,0.543)--(1.583,0.559)%
  --(1.583,0.575)--(1.583,0.591)--(1.583,0.607)--(1.583,0.623)--(1.583,0.638)--(1.583,0.654)%
  --(1.583,0.670)--(1.583,0.686)--(1.583,0.702)--(1.583,0.718)--(1.583,0.733)--(1.583,0.749)%
  --(1.583,0.765)--(1.583,0.781)--(1.583,0.797)--(1.583,0.813)--(1.583,0.828)--(1.583,0.844)%
  --(1.583,0.860)--(1.583,0.876)--(1.583,0.892)--(1.583,0.908)--(1.583,0.923)--(1.583,0.939)%
  --(1.583,0.955)--(1.583,0.971)--(1.583,0.987)--(1.583,1.003)--(1.583,1.018)--(1.583,1.034)%
  --(1.583,1.050)--(1.583,1.066)--(1.583,1.082)--(1.583,1.097)--(1.583,1.113)--(1.583,1.129)%
  --(1.583,1.145)--(1.583,1.161)--(1.583,1.177)--(1.583,1.192)--(1.583,1.208)--(1.583,1.224)%
  --(1.583,1.240)--(1.583,1.256)--(1.583,1.272)--(1.583,1.287)--(1.583,1.303)--(1.583,1.319)%
  --(1.583,1.335)--(1.583,1.351)--(1.583,1.367)--(1.583,1.382)--(1.583,1.398)--(1.583,1.414)%
  --(1.583,1.430)--(1.583,1.446)--(1.583,1.462)--(1.583,1.477)--(1.583,1.493)--(1.583,1.509)%
  --(1.583,1.525)--(1.583,1.541)--(1.583,1.556)--(1.583,1.572)--(1.583,1.588)--(1.583,1.604)%
  --(1.583,1.620)--(1.583,1.636)--(1.583,1.651)--(1.583,1.667)--(1.583,1.683)--(1.583,1.699)%
  --(1.583,1.715)--(1.583,1.731)--(1.583,1.746)--(1.583,1.762)--(1.583,1.778)--(1.583,1.794)%
  --(1.583,1.810)--(1.583,1.826)--(1.583,1.841)--(1.583,1.857)--(1.583,1.873)--(1.583,1.889)%
  --(1.583,1.905)--(1.583,1.921)--(1.583,1.936)--(1.583,1.952)--(1.583,1.968)--(1.583,1.984)%
  --(1.583,2.000)--(1.583,2.016)--(1.583,2.031)--(1.583,2.047)--(1.583,2.063);
%% coordinates of the plot area
\gpdefrectangularnode{gp plot 1}{\pgfpoint{0.016cm}{0.496cm}}{\pgfpoint{1.583cm}{2.063cm}}
\end{tikzpicture}
%% gnuplot variables
\\[-13pt]%
$\psi = C_{0} + C_{1} \arg\,\bigl(\pm(x - \xsub0) \cdot \mu \bigr)$&%
$\psi = C_{0} + x \cdot \nu$%
\end{tabular}
\smallskip
\caption{Potentials compatible with $D \ge 6$ type III alignment preserving KK reduction}\label{fig:psi-contours}
\end{figure}

\begin{remark}
	Having found the explicit form of $\psi$, we could proceed to compute also the explicit form of the conformal factor $a^{2}$ by means of the relation \eqref{eq:III-a-equation}.
\end{remark}

\begin{remark}
	In $D = 4$, the equation \eqref{eq:III-psi-equation} admits more solutions, such as for example a linear combination of three fields $\psi$ in the form of \eqref{eq:III-psi-solution-a}. But from the trace of \eqref{eq:III-psi-equation}, we know that all such solutions must be harmonic (unlike for $D > 4$):
	\begin{equation}
		\psi_{,ii} = 0\,.
	\end{equation}
\end{remark}

\subsubsection{Maxwell field}

In frame components, the speciality condition for $\bm{\mathcal{F}}$ reads:
\begin{equation}
	\bm{\mathcal{F}} = \mathcal{F}\_{(1)(i)}\,\bm{\varepsilon}\^{(1)}\!\wedge \bm{\varepsilon}\^{(i)},
\end{equation}
which is translated to coordinates as
\begin{equation}
	\bm{\mathcal{F}} = \mathcal{F}_{ui}\,\bm{\mathrm{d}}u \wedge \bm{\mathrm{d}}x\^{(i)},
\end{equation}
or in the contravariant version:
\begin{equation}\label{eq:kundt-F-alignment-contravariant}
	{}^{\sharp\sharp}\bm{\mathcal{F}} = a^{-2} \mathcal{F}_{ui}\,\bm{\partial}_{r} \wedge \bm{\partial}_{i}\,.
\end{equation}

Until now, we have been neglecting the equation \eqref{eq:reduced-einstein-eq-b} for $\bm{\mathcal{F}}$. In coordinates, it becomes
\begin{equation}
	\Big( a^{D-2}\,\mathrm{e}^{((D-4)\alpha + 3\beta)\phi} \mathcal{F}^{IJ} \Big)_{\!,J} = 0\,,
\end{equation}
where the indices $\scriptstyle I$, $\scriptstyle J$ label components in coordinates $r$, $u$, $x^{i}$.
Looking at \eqref{eq:kundt-F-alignment-contravariant}, we see that the $u$ component is satisfied identically. The $i$ component demands that $\mathcal{F}_{ui}$ is independent of the coordinate $r$:
\begin{equation}
	\mathcal{F}_{ui,r} = 0\,.
\end{equation}
Finally, for the $r$ component, we get after substituting for $a$ using \eqref{eq:III-a-equation} and \eqref{eq:psi-definition}:
\begin{equation}\label{eq:kundt-maxwell-equation}
	\Big( \psi^{3} \left(\psi_{,j}\psi_{,j}\right)^{\!\frac{D-4}{2}}\!\mathcal{F}_{ui} \Big)\vphantom{\mathcal{F}}_{,i} = 0\,.
\end{equation}

\begin{remark}
	Applying the Poincar\'e lemma on $\mathcal{F}_{[ab;c]} = 0$, which in coordinates reduces to
	\begin{equation}
		\mathcal{F}_{ui,j} = \mathcal{F}_{uj,i}\,,
	\end{equation}
	we see that one can locally always choose a gauge in which $\mathcal{F}_{ui}$ is a gradient of the potential:
	\begin{equation}
		\mathcal{F}_{ui} = -\mathcal{A}_{u,i}\,.
	\end{equation}
\end{remark}

\subsubsection{Einstein equation}

We also have to satisfy the equation \eqref{eq:reduced-einstein-eq-a}. The components of the Ricci tensor at its left-hand side are expressed in \cite{podolsky-kundt-clas}. The $rr$ component is satisfied trivially. The $ri$ component says that $W_{i}$ is a polynomial in $r$ of order at most $1$. The $ru$ component says that $H$ is a polynomial in $r$ of order at most $2$ and expresses the highest order in terms of the remaining components of $g_{ab}$. The $ij$ component reads as follows:
\begin{equation}\label{eq:kundt-R-ij}
	R_{ij} = \SR_{ij} + W_{(i,j)r} - \tfrac{1}{2} W_{i,r} W_{j,r} - \tfrac{1}{2} g^{kl} W_{k,r} (2 g_{l(i,j)} - g_{ij,l})\,,
\end{equation}
where $g^{ij} \equiv {}^{\mathrm{S}}\!g^{ij}$ and $\SR_{ij}$ is the Ricci tensor of the $(D-2)$-dimensional Riemannian metric ${}^{\mathrm{S}}\!g_{ij}$.

We won't deal here directly with the $ui$ and $uu$ equations, the only ones that restrict the $u$-de\-pen\-dence of the fields. Instead, we will extract an interesting information from them, which is unrelated to $u$-de\-pen\-dence -- using a trick that involves a particular \bw0 Weyl tensor component. Expressing the component $\hat{R}_{ruij}$ in terms of the metric's partial derivatives, we get\footnote{The computation \eqref{eq:kundt-Wir-closed-computation} is to be performed in gauge $\mathcal{A}_{r} = 0$, but the result about $W_{i,r}$ being a closed form, is gauge-invariant.} \citep[\bgroup \it cf.\egroup][]{podolsky-kundt-clas}:
\begin{equation}\label{eq:kundt-Wir-closed-computation}
	0 = \hat{C}_{ruij} = \hat{R}_{ruij} = \mathrm{e}^{2\alpha\phi} W_{[i,j]r}\,.
\end{equation}
This means that $W_{i,r}$ is locally a gradient of some function (let us call it $-2 \ln f$, for reasons which will become apparent later):
\begin{equation}
	W_{i,r} = \left( -2 \ln f \right)_{,i}\,.
\end{equation}
Substituting this into \eqref{eq:kundt-R-ij}, we get (after multiplying by the nonvanishing $f$) a linear equation for $f$:
\begin{equation}\label{eq:III-ij-1}
	2f_{,ij} - g^{kl} f_{,k} (2g_{l(i,j)} - g_{ij,l}) + (R_{ij} - \SR_{ij}) f = 0\,.
\end{equation}
Here, $R_{ij}$ is to be substituted from \eqref{eq:reduced-einstein-eq-a} using $\phi_{;ij}$ from \eqref{eq:III-phi-ij}:
\begin{equation}
	R_{ij} = (D-2)(\alpha-\beta)\alpha\,\phi_{,i} \phi_{,j} - \big(\beta+(D-2)\alpha\big)\alpha\,s\,\eta_{ij}\,,
\end{equation}
while the conformally flat $\SR_{ij}$ can be expressed easily by means of the conformal factor $a^{2}$, which is related to $\phi$ by \eqref{eq:III-a-equation}; after further eliminating $\phi_{,ij}$ using \eqref{eq:III-phi-equation}:
\begin{align}
	\SR_{ij} ={}& (D-4) \left(\alpha(\alpha-\beta)\,\phi_{,i} \phi_{,j} - \frac{s_{\!,ij}}{2s} + 3\,\frac{s_{\!,i} s_{\!,j}}{4s^{2}}\right) - {}\nonumber\\*
	& {} - \left((\alpha-\beta)\big((D-4)\alpha + (D-3)\beta\big) s + \frac{s_{\!,kk}}{2s} + (D-6)\frac{s_{\!,k} s_{\!,k}}{4s^{2}}\right) \eta_{ij}\,.
\end{align}
In the above expressions, $s$ is defined as the \enquote{square} of the gradient of $\phi$:
\begin{equation}
	s := \phi_{,i} \phi_{,i}\,.
\end{equation}

In \eqref{eq:III-ij-1}, we can also express the spatial part of the metric by means of \eqref{eq:conformally-flat-metric}:
\begin{equation}\label{eq:III-ij-2}
	2f_{,ij} - f_{,i} (\ln{a^{2}})_{,j} - f_{,j} (\ln{a^{2}})_{,i} + f_{,k} (\ln{a^{2}})_{,k}\,\eta_{ij} + (R_{ij} - \SR_{ij}) f = 0\,.
\end{equation}
The conformal factor $a^{2}$ can be in turn expressed by means of \eqref{eq:III-a-equation} in terms of $\phi$. Using the relation \eqref{eq:III-phi-01}, we can extract from this equation information unrelated to $f$. Indeed, contracting \eqref{eq:III-ij-2} with $\eta^{jk} \phi_{,k}$ and substituting from \eqref{eq:III-a-equation}, we get after some manipulation:
\begin{equation}
	\eqref{eq:III-ij-2}_{ij} \phi_{,j} = (2\alpha\beta s\,\eta_{ij} + R_{ij} - \SR_{ij})\,f\phi_{,j} - 2s^{-1}\eqref{eq:III-phi-equation}_{ij} f_{,j}\,,
\end{equation}
where $\text{(x)}_{ij}$ marks the left hand side of the corresponding equation rearranged as $\text{(x)}_{ij} = 0$. This means that $\phi_{,i}$ is an eigenvector of the matrix $\SR_{ij} - R_{ij}$ with the corresponding eigenvalue equal to $2\alpha\beta s$:
\begin{equation}\label{eq:phi-eigenvector-of-ricci}
	(R_{ij} - \SR_{ij})\,\phi_{,j} + 2 \alpha\beta s\,\phi_{,i} = 0\,.
\end{equation}

On the other hand, under a partial substitution of
\begin{equation}
	\varsigma := \frac{1}{\sqrt{\psi_{,i} \psi_{,i}}}\,,
\end{equation}
the overall expression for the \enquote{residual} Ricci curvature $\SR_{ij} - R_{ij}$ simplifies to:
\begin{multline}\label{eq:III-SRij-Rij}
	\SR_{ij} - R_{ij} - 2\alpha\beta\,\phi_{,i}\phi_{,j} + 2\alpha^{2}\,(\phi_{,i} \phi_{,j} - s \eta_{ij}) = {}\\*
	{} = (D-4) \left(\frac{\varsigma_{,ij}}{\varsigma} - \frac{\varsigma_{,k} \varsigma_{,k}}{\varsigma^{2}} \eta_{ij}\right) + \left(\frac{\varsigma_{,kk}}{\varsigma} - \frac{\varsigma_{,k} \varsigma_{,k}}{\varsigma^{2}}\right) \eta_{ij}\,.
\end{multline}
With this relation, it is a matter of direct calculation to show that for all the solutions \eqref{eq:III-psi-solution}, the scalar field gradient $\phi_{,i}$ is indeed an eigenvector of the \enquote{residual} Ricci curvature $\SR_{ij} - R_{ij}$:
\begin{equation}\label{eq:kundt-eigenvector-in-psi}
	\psi_{,i} \in \ker \varsigma_{,ij}\,.
\end{equation}
Similarly, for all solutions \eqref{eq:III-psi-solution}, it can also be shown that
\begin{equation}
	\label{eq:kundt-eigenvalue}
	\psi_{,i} \left(\vphantom{\frac{\frac{1}{1}}{1}}(D-4) \left(\frac{\varsigma_{,ij}}{\varsigma} - \frac{\varsigma_{,k} \varsigma_{,k}}{\varsigma^{2}} \eta_{ij}\right) + \left(\frac{\varsigma_{,kk}}{\varsigma} - \frac{\varsigma_{,k} \varsigma_{,k}}{\varsigma^{2}}\right) \eta_{ij} \right) \psi_{,j} = 0\,,
\end{equation}
meaning that all of the solutions \eqref{eq:III-psi-solution} are compatible with the condition \eqref{eq:phi-eigenvector-of-ricci}.

\begin{remark}
	For $D=4$, other solutions are be possible. For example the strictly $D=4$ solution of \eqref{eq:III-psi-equation}
	\begin{equation}\label{eq:pure-4d-psi}
		\psi = C_{0} + \sum_{i=1}^{3} C_{i} \ln \left(x - \xsub i\right)^{2},
	\end{equation}
	does not generally satisfy even \eqref{eq:kundt-eigenvector-in-psi}, which however does not\footnote{thanks to the $D-4$ factor in \eqref{eq:III-SRij-Rij}} disqualify it as a scalar field compatible with type III Weyl alignment preserving Kaluza--Klein reduction.
\end{remark}

\subsection{Type N Kaluza--Klein reduction}

\begin{proposition}\label{prop:N}
	Let $D \ge 4$. Let $\bm{\ell}$ be a 4-WAND of a $D$-dimensional spacetime. Suppose that this spacetime is a Kaluza--Klein reduction of a $(D+1)$-dimensional vacuum spacetime. Then $\bm{\hat{\ell}}$ is a 4-WAND of this $(D+1)$-dimensional spacetime, iff the following holds locally:
	\begin{subequations}\label{eq:conditions-N}
	\begin{gather}
		\label{eq:conditions-N-a}
		\Phi\_{;(1)(i)} = 0\,,\\
		\label{eq:conditions-N-b}
		(\beta-\alpha)\Big(\phi\_{;(i)}\mathcal{F}\_{(1)(j)} + 2\phi\_{;(j)}\mathcal{F}\_{(1)(i)}\Big) + \alpha g\^{(k)(l)} \phi\_{;(k)} \mathcal{F}\_{(1)(l)} g\_{(i)(j)} + \mathcal{F}\_{(1)(i);(j)} = 0\,,
	\end{gather}
	\end{subequations}
	together with conditions \eqref{eq:conditions-I}, \eqref{eq:conditions-II} and \eqref{eq:conditions-III}.
\end{proposition}
\begin{proof}
	Let's discuss the three generating Weyl tensor components for \bw{-1} as listed in table~\ref{tab:weyl-generators}. The condition \eqref{eq:conditions-N-a} is equivalent to the equation
	\begin{equation}
		\hat{C}\_{[1][z][i][z]} = 0\,,
	\end{equation}
	while the equation
	\begin{equation}
		\hat{C}\_{[1][i][j][k]} = 0
	\end{equation}
	is consistent with it. The condition \eqref{eq:conditions-N-b} is equivalent to the equation
	\begin{equation}
		\hat{C}\_{[1][i][j][z]} = 0\,.
	\end{equation}
\end{proof}

\begin{remark}
	For $\mathcal{F}_{ab} = 0$, the condition \eqref{eq:conditions-N-b} is satisfied trivially and the type N Kaluza--Klein alignment conditions reduce to
	\begin{subequations}
	\begin{align}
		\Phi\_{;(0)(0)} &= 0\,,\\
		\Phi\_{;(0)(i)} &= 0\,,\\
		\Phi\_{;(1)(i)} &= 0\,,\\
		\Phi\_{;(i)(j)} - \Phi\_{;(0)(1)} g\_{(i)(j)} &= 0\,,
	\end{align}
	\end{subequations}
	which is equivalent to
	\begin{equation}
		\mathrm{bo}_{\langle\bm{\ell}\rangle}\,(\Phi_{;ab} - \Phi\_{;(0)(1)} g_{ab}) < -1\,.
	\end{equation}
\end{remark}

\begin{proposition}
	The trace and the antisymmetric part of \eqref{eq:conditions-N-b} respectively are equivalent to
	\begin{subequations}
	\begin{align}
		\label{eq:conditions-N-b-trace-frame}
		g\^{(i)(j)} \mathcal{F}\_{(1)(i)} \left(\phi\_{;(j)} - \tfrac{1}{\alpha} L\_{(j)(1)}\right) &= 0\,,\\
		\label{eq:conditions-N-b-antisymmetric-frame}
		\mathcal{F}\_{(1)[(i)} \left( L\_{(j)](1)} + (\beta-\alpha)\,\phi\_{;(j)]} \right) &= 0\,.
	\end{align}
	\end{subequations}
\end{proposition}
\begin{proof}
	Using the Leibniz rule, we express
	\begin{equation}
		\mathcal{F}\_{(0)(1);(1)} = \mathcal{F}\_{(1)(i)} L\^{(i)}\_{\hphantom{(i)}(1)}
	\end{equation}
	and with help of the Maxwell equation \eqref{eq:reduced-einstein-eq-b}, obtain
	\begin{multline}
		g\^{(i)(j)} \mathcal{F}\_{(1)(i);(j)} = g^{bc} \mathcal{F}\_{(1)b;c} - \mathcal{F}\_{(1)(0);(1)} = {}\\
		{} = -((D-4)\,\alpha + 3\beta)\,g\^{(i)(j)}\mathcal{F}\_{(1)(i)}\phi\_{;(j)} + \mathcal{F}\_{(1)(i)}L\^{(i)}\_{\hphantom{(i)}(1)}\,,
	\end{multline}
	which allows us to write the trace of \eqref{eq:conditions-N-b} as \eqref{eq:conditions-N-b-trace-frame}.

	For the antisymmetric part of \eqref{eq:conditions-N-b}, we note that
	\begin{equation}
		\mathcal{F}\_{(1)(i);(j)} - \mathcal{F}\_{(1)(j);(i)} = -\mathcal{F}\_{(i)(j);(1)}
	\end{equation}
	and continue in the same vein as in the case of the trace, using Leibniz rule and frame orthogonality to convert the problem to Ricci rotation coefficients:
	\begin{equation}
		\mathcal{F}\_{(i)(j);(1)} = -\mathcal{F}\_{(1)(i)} L\_{(j)(1)} + \mathcal{F}\_{(1)(j)} L\_{(i)(1)}\,,
	\end{equation}
	The antisymmetric part of \eqref{eq:conditions-N-b} then becomes \eqref{eq:conditions-N-b-antisymmetric-frame}.
\end{proof}

From now on in this section, we will suppose that
\begin{equation}
	\mathcal{F}_{ab} \neq 0\,.
\end{equation}
Then the $D$-dimensional spacetime is Kundt, as we already mentioned in section \ref{chap:III-kundt}.

The equation \eqref{eq:conditions-N-b-antisymmetric-frame} says that $\mathcal{F}\_{(1)(i)}$ and $L\_{(i)(1)} + (\beta-\alpha)\,\phi\_{;(i)}$ are linearly dependent:
\begin{equation}\label{eq:N-L-phi-propto-F}
	L\_{(i)(1)} + (\beta-\alpha)\,\phi\_{;(i)} \propto \mathcal{F}\_{(1)(i)}\,.
\end{equation}
Substituting this into \eqref{eq:conditions-N-b-trace-frame}, we can eliminate $\mathcal{F}\_{(1)(i)}$:
\begin{equation}
	g\^{(i)(j)} \left( L\_{(i)(1)} + (\beta-\alpha)\,\phi\_{;(i)} \right) \left( \alpha\phi\_{;(j)} - L\_{(j)(1)} \right) = 0\,.
\end{equation}
This can be simplified using \eqref{eq:kundt-phi-ij}, to get
\begin{equation}\label{eq:proportional-gradient-norms-frame}
	g\^{(i)(j)} \left( \alpha^{2} \phi\_{;(i)} \phi\_{;(j)} - L\_{(i)(1)} L\_{(j)(1)} \right) = 0\,.
\end{equation}
Now we can see that
\begin{equation}
	L\_{(i)(1)} = \alpha \phi\_{;(i)}\,,
\end{equation}
since, using \eqref{eq:kundt-phi-ij} and \eqref{eq:proportional-gradient-norms-frame},
\begin{multline}
	g\^{(i)(j)} \big( L\_{(i)(1)} - \alpha \phi\_{;(i)} \big) \big( L\_{(j)(1)} - \alpha \phi\_{;(j)} \big) = {}\\
	{} = g\^{(i)(j)} \big( L\_{(i)(1)} L\_{(j)(1)} + \alpha^{2}\phi\_{;(i)} \phi\_{;(j)} - 2\alpha L\_{(i)(1)} \phi\_{;(j)} \big) = 0\,.
\end{multline}
Also, substituting this back to \eqref{eq:N-L-phi-propto-F}, we see that $\phi\_{;(i)}$ must be proportional to $\mathcal{F}\_{(1)(i)}$.
\begin{corollary}
	As a necessary condition for type N Weyl alignment preserving Kaluza--Klein reduction of vacuum with nonvanishing Maxwell field, following must hold:
	\begin{subequations}
	\begin{gather}
		\label{eq:N-L-phi-relation}
		L\_{(i)(1)} = \alpha \phi\_{;(i)}\,,\\
		\label{eq:N-phi-propto-F}
		\phi\_{;(i)} \propto \mathcal{F}\_{(1)(i)}\,.
	\end{gather}
	\end{subequations}
\end{corollary}

The $\scriptstyle\s{(1)}$~component of the Maxwell equation \eqref{eq:reduced-einstein-eq-b} is now equivalent to the trace of \eqref{eq:conditions-N-b}, while the $\scriptstyle\s{(0)}$~component is satisfied identically and the $\scriptstyle\s{(i)}$~component demands that
\begin{equation}\label{eq:N-maxwell-0}
	\mathcal{F}\_{(1)(i);(0)} = 0\,.
\end{equation}

For $\phi\_{;(i)} = 0$, the condition \eqref{eq:conditions-N-b} reduces to
\begin{equation}
	\mathcal{F}\_{(1)(i);(j)} = 0\,,
\end{equation}
which, together with \eqref{eq:N-maxwell-0}, demands that the Maxwell field depends only on $u$ as defined by \eqref{eq:kundt-metric}:
\begin{subequations}
\begin{align}
	\mathcal{F}_{ui,r} &= 0\,,\\
	\mathcal{F}_{ui,j} &= 0\,.
\end{align}
\end{subequations}

From now on in this section, let's assume that
\begin{equation}\label{N-phi-i-nonvanishing}
	\phi\_{;(i)} \neq 0\,.
\end{equation}
Condition \eqref{eq:N-phi-propto-F} can then be written as
\begin{equation}\label{eq:N-F-propto-phi}
	\mathcal{F}\_{(1)(i)} = \gamma \phi\_{;(i)}\,,
\end{equation}
where $\gamma$ is some function on $\mathcal{M}$. Note that the choice of $\gamma$ is frame-dependent. Using this result and the Kundt property \eqref{eq:conditions-III-d} while applying the Leibniz rule as usual, we can express
\begin{equation}
	\mathcal{F}\_{(1)(i);(j)} = \phi\_{;(i)} \gamma\_{;(j)} + \gamma \phi\_{;(i)(j)} - \gamma \phi\_{;(i)} L\_{(1)(j)}\,.
\end{equation}
In a frame with affinely parametrized $\bm{\ell}$, we have
\begin{equation}
	L\_{(1)(i)} = L\_{(i)(1)}\,,
\end{equation}
which can be shown for example using a frame that satisfies \eqref{eq:kundt-frame}. Using \eqref{eq:N-L-phi-relation}, the condition \eqref{eq:conditions-N-b} then becomes
\begin{equation}
	\gamma \phi\_{;(i)(j)} + (3\beta-2\alpha)\,\gamma \phi\_{;(i)} \phi\_{;(j)} + \alpha \gamma g\^{(k)(l)} \phi\_{;(k)} \phi\_{;(l)} g\_{(i)(j)} + \phi\_{;(i)} \gamma\_{;(j)} = 0\,,
\end{equation}
which can be further simplified using \eqref{eq:III-phi-ij-frame}:
\begin{equation}
	2\beta \gamma \phi\_{;(i)} \phi\_{;(j)} + \phi\_{;(i)} \gamma\_{;(j)} = 0\,.
\end{equation}
As we assume $\phi\_{;(i)}$ to be nonvanishing \eqref{N-phi-i-nonvanishing}, it can be factored out:
\begin{equation}
	2\beta \gamma \phi\_{;(i)} + \gamma\_{;(i)} = 0\,,
\end{equation}
or in terms of $\psi$:
\begin{equation}
	(\psi^{2} \gamma)\_{;(i)} = 0\,.
\end{equation}
\begin{corollary}\label{stat:N-with-F-and-phi-i}
	Following conditions are necessary and sufficient for type N Weyl alignment preserving Kaluza--Klein reduction of vacuum with nonvanishing Maxwell field and with scalar field gradient of boost order at least zero:
	\begin{align*}
		\mathrm{bo}_{\langle\bm{\ell}\rangle}\,(\Phi_{;ab} - \Phi\_{;(0)(1)} g_{ab}) &< -1 & L\_{(i)(j)} &= 0 & \phi\_{;(i)} &= \tfrac{1}{\alpha} L\_{(i)(1)}\\*
		\mathrm{bo}_{\langle\bm{\ell}\rangle}\,\mathcal{F}_{ab} &= -1 & (\kappa\^{(i)} &= 0) & \mathcal{F}\_{(1)(i)} &= \gamma \phi\_{;(i)}\\*
		\mathrm{bo}_{\langle\bm{\ell}\rangle}\,\phi_{;a} &= 0 &&& \gamma\_{;(i)} &= -2\beta \gamma \phi\_{;(i)}\,,
	\end{align*}
	where a frame with affinely parametrized $\bm{\ell}$ is used and where $\gamma$ is a frame-dependent function on $\mathcal{M}$.
\end{corollary}

With the help of \eqref{eq:kundt-frame}, which is already affinely parametrized, the conditions translate easily to coordinates of \eqref{eq:kundt-metric}:
\begin{subequations}
\begin{align}
	\label{eq:N-phi-i-coordinates}
	\phi_{,i} &= \tfrac{1}{2\alpha} W_{i,r}\,,\\*
	\mathcal{F}_{ui} &= \gamma \phi_{,i}\,,\\*
	\gamma_{,i} &= -2 \beta \gamma \phi_{,i}\,.
\end{align}
\end{subequations}
For $D = 4$ or $D \ge 6$, the transverse metric is conformally flat and we can substitute everything into the transverse Einstein equation \eqref{eq:kundt-R-ij}. Thanks to \eqref{eq:N-phi-i-coordinates}, we can choose
\begin{equation}
	f = \mathrm{e}^{-\alpha\phi}\,.
\end{equation}
Substituting into \eqref{eq:III-ij-2}, we immediately reveal the same pattern as in \eqref{eq:III-phi-ij}, which can be used to show that
\begin{equation}
	2\alpha(\beta-\alpha) \phi_{,i}\phi_{,j} + 2\alpha^{2} \phi_{,k} \phi_{,k} \eta_{ij} + R_{ij} - \SR_{ij} = 0\,.
\end{equation}
But this expression is (up to a sign) precisely the left hand side of \eqref{eq:III-SRij-Rij}. Requiring the right hand side to also vanish:
\begin{equation}\label{eq:N-phi-compatibility-condition}
	(D-4) \left(\frac{\varsigma_{,ij}}{\varsigma} - \frac{\varsigma_{,k} \varsigma_{,k}}{\varsigma^{2}} \eta_{ij}\right) + \left(\frac{\varsigma_{,kk}}{\varsigma} - \frac{\varsigma_{,k} \varsigma_{,k}}{\varsigma^{2}}\right) \eta_{ij} = 0\,,
\end{equation}
we find that for $D=4$, the condition is already satisfied, given \eqref{eq:kundt-eigenvalue}; however in $D \ge 6$, not all of the scalar field solutions \eqref{eq:III-psi-solution} are compatible. Specifically, all of \eqref{eq:III-psi-solution-a}, \eqref{eq:III-psi-solution-b}, \eqref{eq:III-psi-solution-c} and \eqref{eq:III-psi-solution-e} have nonvanishing off-diagonal components of $\varsigma_{,ij}$, making them incompatible with $D \ge 6$ type~N alignment conditions (for $\mathcal{F}_{ab} \neq 0$). The other two solutions however, are not ruled out by this condition: for \eqref{eq:III-psi-solution-f}, $\varsigma$ is constant (not taking into account the $u$ coordinate) and for \eqref{eq:III-psi-solution-d}, $\varsigma_{,ij}$ turns out to be symmetric and with zero traceless part:
\begin{equation}
	\varsigma_{,ij} = \frac{2}{\left| \nu \right|}\,\eta_{ij}\,,
\end{equation}
which is enough to satisfy \eqref{eq:N-phi-compatibility-condition}, given \eqref{eq:kundt-eigenvalue}. We can conclude that to be compatible with $D \ge 6$ type N alignment preserving Kaluza--Klein reduction with $\mathcal{F}_{ab} \neq 0$, the scalar field must be in either of the following two forms:
\begin{subequations}\label{eq:N-phi-solution}
\begin{align}
	\phi &= \frac{1}{\beta} \ln\left(C_{0} + \frac{\left(x - \xsub0\right) \cdot \nu}{\left(x - \xsub0\right)^{2}}\right)
	\shortintertext{or}
	\phi &= \frac{1}{\beta} \ln\left(C_{0} + x \cdot \nu\right).
\end{align}
\end{subequations}

\section{Conclusion}

We investigated preservation of Weyl tensor null alignment algebraic types within a~Kaluza--Klein reduction of a vacuum spacetime along a Killing vector field where the two aligned null directions are parallel \eqref{eq:original-frame-a} in a gauge where they are perpendicular to the Maxwell potential.

First, we demonstrated an elegant way to derive the Kaluza--Klein relation of Riemann tensors \eqref{eq:riemann-KK} and Weyl tensors \eqref{eq:weyl-KK}, employing a frame formalism; we further used this technique to derive the relations of optical matrices \eqref{eq:reduction-rho} and also of non-geodeticities \eqref{eq:reduction-kappa}, revealing some interesting consequences regarding twist and Kundt spacetimes. Later, we presented convenient algebraic conditions (summarized in section \ref{chap:alignment-conditions-summary}) necessary and sufficient for the Kaluza--Klein lift to be Weyl type preserving with respect to the related null directions, individually for all of types I, II, III and N. Finally, for types III and N in $D \neq 5$ with non-vanishing Maxwell field, we inquired specific forms of scalar field $\phi$ that these conditions demand: we found that in $D \geq 6$, the conditions are likely more strict (see \eqref{eq:III-psi-solution} and figure \ref{fig:psi-contours} for type III, or \eqref{eq:N-phi-solution} for type N) on $\phi$ behavior than in $D = 4$ (see for example \eqref{eq:pure-4d-psi}).

\backmatter

\bmhead{Acknowledgments}
This work was done with patient support of the author's supervisor Mgr. Vojt\v ech Pravda, Ph.D., DSc., who also suggested the topic.

\section*{Declarations}

\bmhead{Funding}
This work has been supported by research plan RVO: 67985840.
The work was done under the auspices of the Albert Einstein Center for Gravitation and Astrophysics, Czech Republic.

\bmhead{Data availability}
Data sharing is not applicable to this article as no datasets were generated or analyzed during the current study.

\bmhead{Author contributions}
The author claims to be responsible for the following roles related to this work: conceptualization, methodology, formal analysis, visualization, validation, writing.

\begin{appendices}

\section{Commutator of Lie and covariant derivative}
\label{chap:lie-covariant-commutator}

Being faced with the rather trivial problem of determining $\mathcal{L}_{\bm{\partial}_z}\!\bm{\nabla} \bm{\hat{g}}$, instead of resorting to reflections on coordinate components of $\bm{\nabla}$, we remember the explicit formula for the commutator of a Lie derivative and a torsion-free covariant derivative, as demonstrated in \citep{cov-lie-commut}. For completeness, let's work out a generalization, admitting nonvanishing torsion. Later, on the contrary, we will focus on the special case of metric covariant derivative, which allows us to use the Ricci identity to completely eliminate the Riemann tensor from the expression.
\begin{theorem}
	The commutator of a Lie derivative along an arbitrary vector field $\bm{\xi}$ and an arbitrary covariant derivative $\bm{\nabla}$ with Riemann curvature $R^{a}_{\hphantom{a}bcd}$ and torsion $T^{a}_{\hphantom{a}bc}$, is given by a~tensor field $\gamma^{a}_{\hphantom{a}bc}$:
	\begin{align}
		\label{eq:lie-covariant-commutator}
		(\mathcal{L}_{\bm{\xi}} \nabla_{a} - \nabla_{a} \mathcal{L}_{\bm{\xi}}) A^{b\dots c}_{d\dots e} ={}& \gamma^{b}_{\hphantom{b}a\mskip-3mu f} A^{f\mskip-3mu \dots c}_{d\dots e} + \dots + \gamma^{c}_{\hphantom{c}a\mskip-3mu f} A^{b\dots\mskip-3mu f}_{d\dots e} - {}\nonumber\\*
		 & {} - \gamma^{f}_{\hphantom{f}ad} A^{b\dots c}_{f\mskip-3mu \dots e} - \dots - \gamma^{f}_{\hphantom{f}ae} A^{b\dots c}_{d\dots \mskip-3mu f}\,,
	\end{align}
	specifically:
	\begin{equation}
		\label{eq:lie-covariant-commutator-tensor}
		\gamma^{a}_{\hphantom{a}bc} = R^{a}_{\hphantom{a}cdb} \xi^{d} + \nabla_{b} (\nabla_{c} \xi^{a} + T^{a}_{\hphantom{a}dc} \xi^{d})\,.
	\end{equation}
\end{theorem}
\begin{proof}
	To keep the analysis lucid, let's work without help of any coordinate covariant derivative (contrary to the derivation in \citep{cov-lie-commut}).

	We observe that the subject commutator
	\begin{equation}
		\mathrm{M}_{a} := \mathcal{L}_{\bm{\xi}} \nabla_{a} - \nabla_{a} \mathcal{L}_{\bm{\xi}}
	\end{equation}
	satisfies
	\begin{subequations}
	\begin{align}
		\bm{\mathrm{M}}(\bm{A}+\bm{B}) &= \bm{\mathrm{M}}\bm{A} + \bm{\mathrm{M}}\bm{B}\,,\\
		\bm{\mathrm{M}}(\bm{A}\bm{B}) &= (\bm{\mathrm{M}}\bm{A})\bm{B} + \bm{A}(\bm{\mathrm{M}}\bm{B})\,,\\
		\bm{\mathrm{M}}\bm{\delta} &= 0\,,\\
		\bm{\mathrm{M}}\phi &= 0\,,
	\end{align}
	\end{subequations}
	where $\bm{A}$, $\bm{B}$ are tensor fields, $\bm{\delta}$ is the Kronecker delta and $\phi$ is a scalar function.
	$\bm{\mathrm{M}}$ is linear and obeys Leibniz rule due to both $\mathcal{L}_{\bm{\xi}}$ and $\bm{\nabla}$ being linear and obeying Leibniz rule. Similarly, $\bm{\mathrm{M}}$ annihilates Kronecker delta (i.e. it commutes with contraction), because both $\mathcal{L}_{\bm{\xi}}$ and $\bm{\nabla}$ annihilate Kronecker delta. Finally, $\bm{\mathrm{M}}$ annihilates a scalar function because in this case, the covariant derivative is also the exterior derivative, which is known to commute with $\mathcal{L}$.

	From this fact alone, it is evident that $\bm{\mathrm{M}}$ is determined by some tensor $\gamma^{a}_{\hphantom{a}bc}$, as described by \eqref{eq:lie-covariant-commutator}. Let's derive the explicit form of $\bm{\gamma}$, by looking at the action of $\bm{\mathrm{M}}$ on an arbitrary vector field $A^{a}$. We have
	\begin{subequations}
	\begin{align}
		\mathcal{L}_{\bm{\xi}} A^{b} &= \xi^{c} \nabla_{c} A^{b} + \beta^{b}_{\hphantom{b}d} A^{d}\,,\\
		\mathcal{L}_{\bm{\xi}} \nabla_{a} A^{b} &= \xi^{c} \nabla_{c} \nabla_{a} A^{b} + \beta^{b}_{\hphantom{b}d} \nabla_{a} A^{d} - \beta^{d}_{\hphantom{d}a} \nabla_{d} A^{b}\,,
	\end{align}
	\end{subequations}
	where
	\begin{equation}
		\beta^{a}_{\hphantom{a}b} = - \nabla_{b} \xi^{a} - T^{a}_{\hphantom{a}cb} \xi^{c}\,.
	\end{equation}
	For $\bm{\mathrm{M}}$, this means:
	\begin{align}
		\mathrm{M}_{a} A^{b} &= \xi^{c} \nabla_{c} \nabla_{a} A^{b} + \beta^{b}_{\hphantom{b}d} \nabla_{a} A^{d} - \beta^{d}_{\hphantom{d}a} \nabla_{d} A^{b} - \nabla_{a} (\xi^{c} \nabla_{c} A^{b} + \beta^{b}_{\hphantom{b}d} A^{d}) = {}\nonumber\\
		{} &= 2\xi^{c} \nabla_{[c} \nabla_{a]} A^{b} - \beta^{d}_{\hphantom{d}a} \nabla_{d} A^{b} - (\nabla_{a} \xi^{c}) \nabla_{c} A^{b} - (\nabla_{a} \beta^{b}_{\hphantom{b}d}) A^{d} = {}\nonumber\\
		{} &= 2\xi^{c} \nabla_{[c} \nabla_{a]} A^{b} + T^{d}_{\hphantom{d}ca} \xi^{c} \nabla_{d} A^{b} - (\nabla_{a} \beta^{b}_{\hphantom{b}d}) A^{d}\,.
	\end{align}
	In the first two terms, we identify the Riemann tensor using the Ricci identity:
	\begin{equation}
		2\nabla_{[c} \nabla_{a]} A^{b} + T^{d}_{\hphantom{d}ca} \nabla_{d} A^{b} = R^{b}_{\hphantom{b}dca} A^{d}\,,
	\end{equation}
	completing the proof of \eqref{eq:lie-covariant-commutator-tensor}:
	\begin{equation}
		\gamma^{b}_{\hphantom{b}ad} A^{d} = \mathrm{M}_{a} A^{b} = (R^{b}_{\hphantom{b}dca} \xi^{c} - \nabla_{a} \beta^{b}_{\hphantom{b}d}) A^{d}\,.
	\end{equation}
\end{proof}

\begin{remark}
	This is enough to deem \eqref{eq:killing-covariant} proven, as in this case $\gamma^{A}_{\hphantom{A}BC} = 0$ thanks to corollary \ref{stat:riemann-z} and to definition of $\bm{\nabla}$ \eqref{eq:nabla-definition-vector}.
\end{remark}

\begin{remark}
	For a torsion-free covariant derivative, $\gamma^{a}_{\hphantom{a}bc}$ is symmetric in the last two indices:
	\begin{equation}
		\label{eq:lie-covariant-commutator-symmetry}
		\gamma^{a}_{\hphantom{a}bc} = \gamma^{a}_{\hphantom{a}cb}\,,
	\end{equation}
	due to Bianchi identity $R^{a}_{\hphantom{a}[bcd]} = 0$. This explains how our result reduces to the form presented in \cite{cov-lie-commut}:
	\begin{equation}
		\gamma^{a}_{\hphantom{a}bc} = R^{a}_{\hphantom{a}bdc} \xi^{d} + \nabla_{c} \nabla_{b} \xi^{a}\,.
	\end{equation}
\end{remark}

For the sake of curiosity, let's apply this formula to the special case of a metric covariant derivative, to eventually reveal a structure which resembles that of the Christoffel symbols.
\begin{corollary}
	The commutator of a Lie derivative along an arbitrary vector field $\bm{\xi}$ and a~Levi-Civita connection $\bm{\nabla}$ is given by \eqref{eq:lie-covariant-commutator}, where
	\begin{align}
		\gamma^{a}_{\hphantom{a}bc} &= g^{ad} (\xi_{d;cb} + \xi_{b;dc} - \xi_{b;cd}) = {}\nonumber\\
		{} &= g^{ad} (\xi_{d;bc} + \xi_{c;db} - \xi_{c;bd})\,.
	\end{align}
\end{corollary}
\begin{proof}
	The first equality comes from the metric-induced Riemann tensor symmetry and Ricci identity:
	\begin{equation}
		R_{acdb} \xi^{d} = R_{bdca} \xi^{d} = \xi_{b;ac} - \xi_{b;ca}\,,
	\end{equation}
	while the second equality (corresponding to the symmetry \eqref{eq:lie-covariant-commutator-symmetry}) reflects the fact that the exact form $\xi_{[a;b]}$ is closed:
	\begin{equation}
		\xi_{[a;bc]} = 0\,.
	\end{equation}
\end{proof}

\end{appendices}

\bibliography{paper}

\end{document}